\documentclass[11pt]{article}

\usepackage[margin=1in]{geometry}
\usepackage{amsmath,amssymb,amsthm,amsxtra}
\usepackage{mathrsfs, mathtools}
\usepackage{bbm}
\usepackage{bm}
\usepackage[shortlabels]{enumitem}
\usepackage{microtype}
\usepackage[T1]{fontenc}
\usepackage{xcolor}
\usepackage{comment}
\usepackage{algorithm}
\usepackage{algpseudocode}
\usepackage{float}
\usepackage{hyperref}

\numberwithin{equation}{section}

\newtheorem{theorem}{Theorem}[section]
\newtheorem*{theorem*}{Theorem}
\newtheorem{definition}[theorem]{Definition}
\newtheorem{lemma}[theorem]{Lemma}
\newtheorem{corollary}[theorem]{Corollary}

\newtheorem{remark}[theorem]{Remark}
\newtheorem{proposition}[theorem]{Proposition}

\newtheorem{fact}[theorem]{Fact}

\newtheorem{problem}[theorem]{Problem}

\newcommand{\cE}{\mathcal{E}}

\newcommand{\cG}{\mathcal{G}}

\newcommand{\cI}{\mathcal{I}}

\newcommand{\cN}{\mathcal{N}}

\newcommand{\cT}{\mathcal{T}}

\newcommand{\CC}{\mathbb{C}}

\newcommand{\EE}{\mathbb{E}}

\newcommand{\RR}{\mathbb{R}}
\newcommand{\E}{\operatorname{\mathbbm{E}}}

\newcommand{\C}{\mathbb{C}}

\newcommand{\1}{\mathbf{1}}

\newcommand{\set}[1]{\left\{#1\right\}}
\newcommand{\BigO}{\mathcal{O}}

\newcommand{\poly}{\mathrm{poly}}

\makeatletter
\DeclareRobustCommand\Equiv{\mathrel{%
  \mathchoice
    {\Equiv@\textfont\displaystyle{.45}}
    {\Equiv@\textfont\textstyle{.45}}
    {\Equiv@\scriptfont\scriptstyle{.5}}
    {\Equiv@\scriptscriptfont\scriptscriptstyle{.55}}
}}
\newcommand{\Equiv@}[3]{%
  \rlap{\raisebox{#3\fontdimen5#12}{$\m@th#2 = $}}%
  \raisebox{-#3\fontdimen5#12}{$\m@th#2 = $}%
}
\makeatother

\newcommand{\Var}{\operatorname{\mathsf{Var}}}

\newcommand{\bitm}{\begin{itemize}[leftmargin=*]}
\newcommand{\eitm}{\end{itemize}}

\newcommand{\benm}{\begin{enumerate}[leftmargin=*]}
\newcommand{\eenm}{\end{enumerate}}

\definecolor{forestgreen}{rgb}{0.13, 0.55, 0.13}

\newcommand{\abs}[1]{\left|#1\right|}

\newcommand{\inparen}[1]{\left(#1\right)}             
\newcommand{\inbraces}[1]{\left\{#1\right\}}           

\newcommand{\arcsinh}{\operatorname{arcsinh}}

\title{Computational Thresholds for Balanced and Fixed-Slice Independent Sets in Bipartite Graphs}
\author{Ijay Narang\thanks{School of Computer Science, Georgia Institute of Technology, inarang3@gatech.edu}
 \and
 Will Perkins\thanks{School of Computer Science, Georgia Institute of Technology, wperkins3@gatech.edu}
\and
Yuzhou Wang\thanks{School of Mathematics, Georgia Institute of Technology, ywang3694@gatech.edu}
\and
Timothy L.~H.~Wee\thanks{School of Mathematics, Georgia Institute of Technology, timothy.wee@gatech.edu}}
\date{\today}

\begin{document}
\maketitle

\begin{abstract}
Motivated by recent work of Kocurek, Oveis Gharan, and Tjowasi
\cite{kocurek2026sampling}, which gives an efficient sampling algorithm for the hard-core model on random
regular bipartite graphs by decomposing into fixed-size slices, we study the worst-case tractability of approximate counting and sampling of  fixed-size slices for  bipartite independent set problems.
Let \(G=(L\sqcup R,E)\) be a bipartite graph  with \(|L|=|R|=n\) and maximum degree $\Delta$. The fixed-slice problem asks to sample uniformly from independent sets satisfying \(|I\cap L|=\alpha_L n\) and
\(|I\cap R|=\alpha_R n\). We show that if the overall density $\alpha$ lies in the interval $(\frac{1}{\Delta}, \tfrac{1}{2})$, and the densities on the two sides are more balanced than the typical phase densities of a random $\Delta$-regular bipartite graph, then there is no FPRAS or efficient sampling scheme unless $\mathbf{NP}=\mathbf{RP}$. 

We then study a related fugacity model in which the  densities are not fixed,
but the independent set is required to be balanced between the two sides of the bipartition. For \(\lambda>0\), the balanced
hard-core model is the ordinary hard-core model with fugacity \(\lambda\),
conditioned on the event \(|I\cap L|=|I\cap R|\). We prove that this model has
the same computational threshold as the  hard-core model on general
bounded-degree graphs. That is, for every fixed \(\Delta\ge 3\), if
\(\lambda<\lambda_c(\Delta)\), then the balanced partition function admits an
FPTAS and the balanced hard-core distribution admits an efficient sampling
scheme. Conversely, if
\(\lambda>\lambda_c(\Delta)\), then no FPRAS or efficient sampler exists on this graph class unless
$\mathbf{NP}=\mathbf{RP}$. 
\end{abstract}

\thispagestyle{empty}

\newpage
\setcounter{page}{1}

\tableofcontents

\section{Introduction}
\label{secIntro}

The hard-core model originates from the study of lattice gas systems in statistical physics and has received much attention in theoretical computer science and related fields. Given a graph $G$ and fugacity $\lambda >0$, the partition function of the hard-core model is given by
$$Z_G(\lambda):= \sum_{I \in \cI_G} \lambda^{|I|},$$ where $\cI_G$ is the set of all independent sets of $G$. When $\lambda = 1$, the partition function equals $|\cI_G|$, the number of independent sets of $G$. Associated to $Z_G(\lambda)$ is the probability measure on $\cI_G$ given by 
$
    \mu_{G,\lambda}(I) = \frac{\lambda^{|I|}}{Z_G(\lambda)}.
$

Algorithmically, there are two main tasks associated with the hard-core model.  The first is to compute (exactly or approximately) $Z_G(\lambda)$; the second is to output an independent set with distribution (close to) $\mu_{G,\lambda}$.  Even in restricted settings (including bounded-degree or  bipartite graphs) computing $Z_G(\lambda)$ is $\#\textbf{P}$-hard and so most attention is focused on when it is possible to efficiently compute $Z_G(\lambda)$ to within an $\epsilon$ relative error.

More precisely, we collect the following definitions, which will be used throughout the paper. A \emph{fully polynomial-time approximation scheme} (FPTAS) is a deterministic algorithm that given $G$ and $\epsilon \in (0,1)$ outputs
\(\widehat Z\) satisfying
\[
(1-\epsilon) Z_G(\lambda)\le \widehat Z \le (1+\epsilon) Z_G(\lambda)
\]
in time polynomial in \(|V(G)|\) and \(1/\epsilon\).  A \emph{fully polynomial-time randomized approximation scheme} (FPRAS) is a randomized algorithm that outputs such an approximation with probability at least $2/3$ and runs in time polynomial in $|V(G)|$ and $1/\epsilon$.  An \emph{efficient sampling scheme}, in our context, is a randomized algorithm that runs in time polynomial in $|V(G)|$ and $1/\epsilon$, and outputs an independent set according to a distribution $\hat{\mu}$ satisfying $\|\mu_{G,\lambda}-\hat{\mu}\|_{\mathrm{TV}}<\epsilon.$

Let $\mathcal{G}_\Delta$  be the set of all graphs of maximum degree at most $\Delta$. For $G \in \mathcal{G}_\Delta$, the computational tractability of approximating $Z_G(\lambda)$ is well understood. If $\lambda < \lambda_c(\Delta)$, then there exists an FPTAS and efficient sampling scheme for the hard-core model~\cite{weitz2006counting}. On the other hand, if $\lambda > \lambda_c(\Delta)$, then there is no FPRAS or efficient sampling scheme unless $\mathbf{NP}=\mathbf{RP}$ \cite{Sly10,sly2012computational,galanis2016inapproximability}. Here $\lambda_c(\Delta) := \frac{(\Delta-1)^{\Delta-1}}{(\Delta-2)^\Delta}$ is known as the uniqueness threshold, as it
corresponds to the boundary between uniqueness and non-uniqueness of the infinite volume hard-core Gibbs measure on the  $\Delta$-regular tree (see, e.g.~\cite{georgii2011gibbs}).

When $G$ is restricted to be bipartite, however, proving hardness becomes substantially more delicate, since the \textbf{NP}-hard optimization problems used in hardness reductions become tractable on bipartite graphs. This difficulty is reflected in the status of approximately counting independent sets in bipartite graphs, known as $\#\textbf{BIS}$, which is one of the central problems in approximate counting. It  is neither known to admit an FPRAS nor known to be \textbf{NP}-hard to approximate \cite{DyerGoldbergGreenhillJerrum2004, DyerGoldbergJerrum2010}.

For bipartite $G \in \mathcal{G}_\Delta$, \cite{cai2016hardness} showed that approximating $Z_G(\lambda)$ is as hard as $\#\textbf{BIS}$ whenever $\lambda>\lambda_c(\Delta)$. The algorithmic side has seen some progress for special classes of instances including bipartite graphs with degree bounds on one side~\cite{liu2015fptas,chen2023uniquenessrapidmixingbipartite} and lattices~\cite{helmuth2019algorithmic,cannon2026pirogov} and expander graphs~\cite{JenssenKeevashPerkins2020,liao2019counting,jenssen2023approximately,jenssen2026refined} for large values of $\lambda$.

Additionally, random $\Delta$-regular bipartite graphs are canonical average-case instances.  Combined with the algorithmic results for $\lambda = \Omega(\log \Delta/\Delta)$~\cite{chen2022sampling}, following~\cite{JenssenKeevashPerkins2020,liao2019counting}, the recent work \cite{kocurek2026sampling} giving efficient algorithms for $\lambda = O(\Delta^{-1/2})$,
shows that, for random $\Delta$-regular bipartite graphs (for $\Delta$ a sufficiently large constant), one can efficiently
sample from the hard-core model for all fugacities. A key feature of the
algorithm and analysis of~\cite{kocurek2026sampling}  is the use of fixed-size slices. Instead of directly sampling from the hard-core measure $\mu_{G,\lambda}$, they sample from both the one-sided and two-sided slices of bipartite independent sets; that is, the hard-core measure conditioned on the number of left-occupied vertices being $k$ (the one-sided slice) and the hard-core measure conditioned on the number of left-occupied vertices being $k_L$ and the number of right-occupied vertices being $k_R$.

For a bipartite graph
$G=(L\sqcup R,E)$ with $|L|=|R|=n$ and densities $(\alpha_L,\alpha_R)\in[0,1]^2$, a fixed slice consists of the independent sets
\[ \mathcal{I}_{\alpha_L, \alpha_R}(G) := \{I \in \cI_G : 
    |I\cap L|= \lfloor \alpha_L n \rfloor, \hspace{0.1cm}|I\cap R|=\lfloor \alpha_R n \rfloor \}.
\]
We write $\mathrm{FixedSlice}(\alpha_L,\alpha_R)$ for the associated problems of approximately counting $Z^{\mathrm{slice}}_G(\alpha_L,\alpha_R):=
|\mathcal I_{\alpha_L,\alpha_R}|$ and approximately sampling from $\mu^{\mathrm{slice}}_{G,\alpha_L,\alpha_R}$, the uniform distribution on $\mathcal I_{\alpha_L,\alpha_R}$. The algorithmic analysis of \cite{kocurek2026sampling} proceeds by proving the spectral independence
estimates needed to obtain rapid mixing on the fixed $(\alpha_L, \alpha_R)$-slices used to sample from the hard-core model.

Thus, one may hope that sampling via fixed slices could lead to progress on $\#\textbf{BIS}$. Our first main result shows that this approach fails in the
worst case: there exist fixed bipartite instances for which the corresponding fixed-slice problem is \textbf{NP}-hard.

To state the main result, we need some notation related to densities.   For $\lambda > \lambda_c(\Delta)$, the hard-core model on the infinite $\Delta$-regular tree exhibits non-uniqueness.  Two distinct semi-translation-invariant Gibbs measures are obtained by taking the limit of finite Gibbs measures with even-occupied and odd-occupied boundary conditions respectively.  We write $\alpha_+(\lambda)$ and $\alpha_-(\lambda)$ for the probability the root of the tree is occupied under these respective measures.  When $\lambda \le \lambda_c(\Delta)$ we have $\alpha_+(\lambda) = \alpha_-(\lambda)$ since there is a unique infinite-volume Gibbs measure; when $\lambda> \lambda_c(\Delta)$, we have  $\alpha_+(\lambda) > \alpha_-(\lambda)$.  We define $\alpha(\lambda) := (\alpha_+(\lambda) + \alpha_-(\lambda))/2$.  This is also the limiting expected density of the hard-core model on the random $\Delta$-regular bipartite graph~\cite{sly2012computational}.  Next we let $\alpha_c(\Delta) := \alpha(\lambda_c(\Delta)) = \frac{1}{\Delta}$.  Finally note that the map $\lambda \in (\lambda_c, \infty) \mapsto \alpha(\lambda) \in (1/\Delta,1/2)$ is strictly increasing, and hence invertible on the relevant range. We denote its inverse by $\lambda(\alpha)$.

Let $\cG_{\Delta,\Delta}$ denote the class of bipartite graphs
$G=(L\sqcup R,E)$ with maximum degree at most $\Delta$ and $|L|=|R|$. 

\begin{theorem} \label{Thm:FixedDensityHardness}
Fix an integer $\Delta\ge 3$ and $\alpha_L,\alpha_R\in(0,1)$. Let $\alpha = (\alpha_L + \alpha_R)/2 $.  If  $\alpha\in(\alpha_c(\Delta),1/2)$ and
\begin{equation}
    \frac{\alpha_-(\lambda(\alpha))}{\alpha_+(\lambda(\alpha))}
    <
    \frac{\alpha_L}{\alpha_R}
    <
    \frac{\alpha_+(\lambda(\alpha))}{\alpha_-(\lambda(\alpha))},
\end{equation}
then, unless $\mathbf{NP}=\mathbf{RP}$, there is no FPRAS or efficient sampling scheme for $\mathrm{FixedSlice}(\alpha_L,\alpha_R)$ for  inputs $G \in \cG_{\Delta,\Delta}$.
\end{theorem}

Thus, fixed-slice decompositions do not by themselves alleviate worst-case
hardness in bipartite graphs: some fixed slices are already \textbf{NP}-hard to approximately count and sample. The hard
slices in our theorem are necessarily off the phase-aligned ratios that are
typical under the relevant hard-core Gibbs measure on random bipartite graphs, so this result does not
rule out slice-based algorithms for $\#\textbf{BIS}$. Rather, it identifies a
worst-case obstruction and motivates the question of which slices are
tractable.  We discuss the relationship between different regimes of fixed slices and how they relate to our hardness results and the work of \cite{kocurek2026sampling} in Section \ref{subsec:future}.

The fixed-slice model fixes both coordinates \((|I\cap L|,|I\cap R|)\).
We next consider a softer but still global constraint, in which only their
difference $B(I) :=|I\cap L|- |I\cap R|  $ is fixed.   A natural such constraint is \textit{balance}: we  require the independent set to use the two
sides equally. Indeed, balance constraints introduce a related source of hardness on bipartite graphs. Although the maximum independent set problem is polynomial-time solvable on bipartite graphs, maximum balanced independent set is \textbf{NP}-hard \cite{feige2002relations}. Moreover, \cite{perkins2024hardness} showed that, in random bipartite graphs, maximum balanced independent set exhibits the same kind of statistical--computational gap as maximum independent set in random graphs, namely hardness for local and low-degree algorithms.

Let $\cI_G^{\mathrm{bal}}$ denote the collection of balanced independent sets of $G$; that is $\cI_G^{\mathrm{bal}} = \{I \in \cI_G : 
    |I\cap L|= |I\cap R| \}$.
For $\lambda>0$, define the balanced partition function and the associated Gibbs distribution by
\begin{align*}
    Z_G^{\mathrm{bal}}(\lambda)
    &= \sum_{I\in \cI_G^{\mathrm{bal}}} \lambda^{|I|}, \\
    \mu_{G,\lambda}^{\mathrm{bal}}(I)
    &= \frac{\lambda^{|I|}}{Z_G^{\mathrm{bal}}(\lambda)},
    \qquad I\in \cI_G^{\mathrm{bal}}.
\end{align*}

Given this model, a natural question is whether approximating
$Z_G^{\mathrm{bal}}(\lambda)$ for $G\in\mathcal{G}_{\Delta,\Delta}$
exhibits the same hardness threshold as approximating the hard-core
partition function on general bounded-degree graphs. We answer this
question affirmatively by proving the following theorem.

\begin{theorem} \label{thm:comp_thresh} Fix an integer $\Delta \ge 3$ and a fugacity $\lambda > 0$. Then the following hold.
\begin{itemize}
    \item If $\lambda < \lambda_c(\Delta),$ there is an FPTAS for $Z_G^{\mathrm{bal}} (\lambda)$ and an efficient sampling scheme for $\mu_{G, \lambda}^{\mathrm{bal}}$ on inputs $G \in \mathcal{G}_{\Delta,\Delta}$.
    \item If $\lambda>\lambda_c(\Delta)$, then unless $\mathbf{NP}=\mathbf{RP}$, there is no FPRAS for $Z_G^{\mathrm{bal}}(\lambda)$ or efficient sampling scheme for $\mu_{G,\lambda}^{\mathrm{bal}}$ on inputs $G\in\mathcal{G}_{\Delta,\Delta}$.
\end{itemize}
\end{theorem}
For the positive algorithmic result, we in fact treat a larger class of graphs than $\mathcal{G}_{\Delta,\Delta}$: bipartite graphs of maximum degree $\Delta$ for which the ratio $|L|/|R|$ is not too far from $1$; see~\eqref{eq:tilde_graph_class} below.
We do not address the critical case $\lambda=\lambda_c(\Delta)$, but results for the hard-core model~\cite{chen2025rapid} suggest that this case should be tractable. 

The challenge to proving the positive algorithm result of Theorem \ref{thm:comp_thresh} is that the event \(|I\cap L|=|I\cap R|\) can be atypical and exponentially rare under the unconstrained hard-core measure, so naive rejection sampling may fail to be efficient. The main idea is to introduce an exponential tilt $e^{tB(I)}$ that makes the exact balanced event more typical. At a nearly centering tilt, balance has probability of order $n^{-1/2}$, making rejection sampling plausible.  This exponential tilt, however, takes us away from the usual hard-core model to a model with different fugacities for vertices in $L$ and $R$, and the main part of the analysis is showing that we can sample efficiently from this measure.

Above \(\lambda_c(\Delta)\), the same global constraint becomes a source of hardness. We build upon the phase coexistence gadget framework of Sly~\cite{Sly10} and extensions to models with global constraints in~\cite{davies2023approximately,carlson2022computational}. The hard-core model on such a  gadget (derived from random bipartite graphs) displays two phases: one in which the left side is more heavily occupied, and one in which the right side is more heavily occupied. By replacing each vertex with a gadget, the phase of each gadget can be used as a binary label. Exact balance then forces these labels to appear in equal numbers, while the edges between gadgets reward labelings with few crossing edges. This allows the balanced partition function to encode the bisection problem. 

Previous work on the worst-case complexity of approximate counting and sampling under global constraints include~\cite{davies2023approximately,jain2023optimal} on sampling independent sets of specified size on bounded degree graphs and~\cite{carlson2022computational,kuchukova2025fast} on sampling from the ferromagnetic Ising model at fixed magnetization on bounded degree graphs.

\subsection{Overview of the techniques}

\subsubsection{Algorithms below uniqueness}

We fix $\lambda < \lambda_c(\Delta)$ and consider the balanced hard-core model $\mu_{G,\lambda}^{\text{bal}}$ on $G \in \cG_{\Delta,\Delta}$. This model is equivalent to the usual hard-core model conditioned on the \emph{balance} $B(I)$ of an independent set being zero, where
\begin{equation}\label{eq:balance_definition}
    B(I):=|I\cap L|-|I\cap R|.
\end{equation}

Our strategy revolves around a relaxed version of this balanced model: a \emph{tilted} hard-core model with partition function
\begin{align}\label{eq:tiltedHardcorePartition_intro}
     Z_G(\lambda;t) := \sum_{I \in \cI_G} \lambda^{\abs{I}} e^{t B(I)} = \sum_{I \in \cI_G} \inparen{\lambda e^t}^{\abs{I \cap L}} \inparen{\lambda e^{-t}}^{\abs{I \cap R}},
\end{align}
which is a bivariate hard-core model with fugacities $\lambda e^t$ and $\lambda e^{-t}$ on $L$ and $R$ respectively. Write $\mu_{G,\lambda,t}$ for the tilted measure associated to $Z_{G}(\lambda;t)$. The link is that $e^{tB(I)} = 1$ whenever $B(I) = 0$, so the conditional law of $\mu_{G,\lambda,t}$ given $B(I) = 0$ is exactly $\mu_{G,\lambda}^{\text{bal}}$.

The tilt $t$ gives us flexibility to favor occupation in $L$ and $R$ by setting $t$ to be  positive or  negative respectively. A monotonicity argument shows that there is a unique $t^*$ for which $\EE_{\mu_{G,\lambda,t^*}} B(I) = 0$.

Our sampling algorithm is relatively straightforward. At a high level, the steps are:
\begin{enumerate}[(i)]
    \item We first perform noisy binary search by empirical sampling from the tilted measure to find a $\widetilde{t}$ that  approximates $t^*$, in the sense $|\EE_{\mu_{G,\lambda,\widetilde{t}}} B(I) | = O(1)$.
    \item Then we rejection sample from $\mu_{G,\lambda,\widetilde{t}}$ until we obtain an independent set $I$ with $B(I) = 0$.
\end{enumerate}

The proof that this sampling scheme for $\mu_{G,\lambda}^{\text{bal}}$ can be made efficient uses several technical ingredients. We show that sampling from $\mu_{G,\lambda,t}$ can be done efficiently by a self-avoiding walk (SAW) tree approach building upon \cite{weitz2006counting}. We defer technical definitions of the SAW tree and its boundary conditions to Section~\ref{sec:sampling_from_tilted}. Nevertheless, we highlight informally a key step which is to establish a form of correlation decay, \emph{strong spatial mixing} (SSM), on the SAW tree.

For a boundary condition \(\tau\) on the SAW tree rooted at vertex \(v\), let \(R_v^\tau\) denote the resulting root occupation ratio, i.e.~the ratio of the probabilities that the root is occupied to unoccupied.

\begin{theorem*}[Informal version of Theorem \ref{thm:ssm}]
Fix $\lambda<\lambda_c(\Delta)$. For every fixed $t$, the tilted hard-core model \eqref{eq:tiltedHardcorePartition_intro} satisfies SSM on the SAW tree. That is, if two boundary conditions $\tau,\tau'$ first differ at distance $\ell$ from the root, then
\[
    |R_v^\tau-R_v^{\tau'}|
    \le
    C e^{-c\ell},
\]
where $C,c>0$ depend only on $\Delta,\lambda,t$.
\end{theorem*}

In particular, if \(t\) is restricted to a fixed compact interval, the constants \(C,c\) may be chosen uniformly over all graphs, roots, and boundary conditions. Consequently, truncating the SAW tree at logarithmic depth gives a polynomial-time oracle for the marginal occupation probabilities in the tilted model. A standard self-reduction argument then yields an efficient sampler for $\mu_{G,\lambda,t}$.

 The argument is rather delicate and standard estimates do not apply: while the base fugacity $\lambda$ is below $\lambda_c(\Delta)$,  one of the tilted fugacities $\lambda e^t$ and $\lambda e^{-t}$ may rise above this threshold. We instead exploit the bipartite structure and show that the occupation-ratio recursion on the SAW tree is contractive in a two-level sense, corresponding to a recursion on $G$ that maps from $L \to R \to L$ or vice versa. We show this contraction after passing to the coordinate system induced by the potential function $x \mapsto \arcsinh \sqrt{x}$ (also used in \cite{sinclair2017spatialmixingconnectiveconstant,chen2026zerofreenesshardcoreconnective}).

For \(G\in\cG_{\Delta,\Delta}\), the centering tilt $t^*$ can indeed be shown to live in a compact interval that does not scale with $\abs{V(G)}$. In addition to being essential for the SSM arguments, this compactness allows us to establish a zero-freeness result for a complex version of $Z_G(\lambda;t)$ in some region around the real line. The latter leads to a local central limit theorem for $B(I)$, which in turn gives the acceptance probability estimate
\[
    \Pr_{\mu_{G,\lambda,\widetilde{t}}}\{B(I)=0\}=\Omega\!\left(|V(G)|^{-1/2}\right).
\]
This estimate provides the guarantee that rejection sampling succeeds in polynomial time.

We now turn to the FPTAS for $Z_G^{\text{bal}}(\lambda)$. The tilted hard-core model plays a key role again, along with the identity
\[
    Z_G^{\mathrm{bal}}(\lambda)
    =
    Z_G(\lambda;t)\Pr_{\mu_{G,\lambda,t}}\{B(I)=0\}.
\]
The three tasks are thus (1) to find a nearly centered tilt \(\widetilde t\) so that the  probability factor is inverse-polynomially large, (2) to give an FPTAS for \(Z_G(\lambda;\widetilde t)\), and (3) to approximate this probability factor efficiently.

In task (1) we replace the noisy binary search used in the sampling procedure with a more complicated deterministic procedure, based on deterministic estimates of $\E_{\mu_{G,\lambda,t}}B(I)$. Tasks (2) and (3) essentially recycle the ingredients from the sampling proof, with task (3) ``algorithmizing'' the local CLT proof, in the spirit of \cite{jain2021approximatecountingsamplinglocal}.

\begin{remark}
Chen--Liu--Yin \cite{chen2023uniquenessrapidmixingbipartite} study bivariate hard-core models on bipartite graphs related to the tilted model $Z_G(\lambda;t)$. They allow for a larger class of bipartite graphs requiring only a degree bound on one side. They give a uniqueness condition phrased in terms of the fixed points of a map $x \mapsto \lambda_L\left(1+\lambda_R(1+x)^{-w}\right)^{-(\Delta - 1)}$, where $(\lambda_L,\lambda_R)$ are the left and right fugacities, and where $w > 0$ is a branching parameter for the (possibly) unbounded degree side. Under this condition, they prove spectral independence, influence decay, and rapid mixing.

Our tilted measure analysis pertains to a subfamily of such models, where $(\lambda_L,\lambda_R)=(\lambda e^t,\lambda e^{-t})$ traces a curve $\lambda_L \lambda_R = \lambda^2$ in the bivariate fugacity space. Here we prove auxiliary SSM results for the associated SAW tree, zero-freeness for the tilted partition function $Z_G(\lambda;t)$, and local CLT results for the balance variable $B(I)$, in order to handle the exact balance model $Z^{\text{bal}}_{G}(\lambda)$. Furthermore, our analysis is tight, as witnessed by $\lambda_L = \lambda_R = \lambda_c(\Delta)$. 
\end{remark}

\subsubsection{Hardness above uniqueness}

Above the uniqueness threshold, both hardness results are proved by adapting the phase-coexistence gadget framework of \cite{Sly10, carlson2022computational}.
Briefly, the gadget is obtained from a random bipartite $\Delta$-regular graph by a modification that creates designated terminal vertices. In the non-uniqueness regime $\lambda>\lambda_c(\Delta)$, the hard-core measure on this gadget has two dominant phases: in the $(+)$-phase the left side is more heavily occupied, while in the $(-)$-phase the right side is more heavily occupied. 
Given an input graph $H$, we replace each vertex $x\in V(H)$ by a gadget copy $G_x$ and connect the terminals of different copies according to the edges of $H$. The phase vector $Y=(Y_x)_{x\in V(H)}\in\{+,-\}^{V(H)}$ then plays the role of a spin configuration on $H$.

For the balanced model, the global constraint $|I\cap L|=|I\cap R|$ forces the phase vector to be nearly balanced. Informally, before the inter-gadget edges are imposed, phase vectors with unequal numbers of $(+)$-phase and $(-)$-phase gadgets have exponentially small contribution after conditioning on exact balance, while balanced phase vectors retain enough mass to dominate the constrained partition function. Thus the exact balance constraint converts the possible phase vectors into bisections of the input graph.

The inter-gadget edges then encode the cut objective. 
Terminal occupations are approximately independent under the phase-conditioned gadget measures, so the probability that all inter-gadget edges are legal factors over the edges of $H$. 
With our choice of terminal matchings, an edge whose endpoints have the same phase contributes a larger factor than an edge whose endpoints have opposite phases. 
Hence, among balanced phase vectors, the dominant contribution to the balanced partition function comes from phase vectors minimizing the number of cut edges. 
This allows us to recover the minimum bisection value from sufficiently accurate multiplicative approximations to the balanced partition function.

The fixed-slice hardness proof reuses the same construction and terminal compatibility calculation. 
The difference is that the constraint now fixes the two side occupations separately, rather than only their difference. 
We choose the gadget fugacity and the number of isolated vertices so that the target slice is centered around phase vectors with a prescribed number of plus phases. 
The analogue of the balance point probability estimate is now a two-dimensional point probability estimate for $(|I\cap L|,|I\cap R|)$. 
Phase vectors with the wrong number of plus phases are exponentially suppressed, while correctly centered phase vectors retain enough mass. 
The same inter-gadget compatibility factor then encodes the corresponding fixed-cardinality cut problem.

\subsection{Fixed-Size Sampling and Future Directions} \label{subsec:future}

\label{subsec:fixed-slices-random-graphs}

Schematically, the decomposition used in \cite{kocurek2026sampling} can be viewed as
\begin{align} \label{eq:shayan_decomp}
Z_G(\lambda)
&\approx
\sum_{\substack{I\in\cI_G:\\
|I\cap X|\le \alpha_0 |X|,\ |I\cap Y|\le \alpha_0 |Y|}}
\lambda^{|I|} +
\sum_{\alpha_0 |X|<k\le \beta_0 |X|}
\sum_{\substack{S\subseteq X:\\ |S|=k}}
\lambda^k(1+\lambda)^{|Y\setminus N(S)|} \nonumber \\
&+
\sum_{\alpha_0 |Y|<k\le \beta_0 |Y|}
\sum_{\substack{S\subseteq Y:\\ |S|=k}}
\lambda^k(1+\lambda)^{|X\setminus N(S)|}
\end{align}
where
\[
    \alpha_0 := \frac{\log \Delta}{(2+o_\Delta(1))\Delta},
    \qquad
    \beta_0 := \Theta(\lambda)
\]

The first sum (low-low) is approximated by sampling from each fixed two-sided slice via a down-up walk. The second and third sums are approximated by running a different down-up walk on one-sided slices. When $G$ is a random $\Delta$-regular graph, their analysis shows that these walks mix in polynomial time. 

Now, consider the hardness result of Theorem \ref{Thm:FixedDensityHardness}. In particular, for each fixed $0<\epsilon<1$, the low-low slices $(\rho_L,\rho_R)=((1-\epsilon)\rho,(1+\epsilon)\rho)$ are hard in the worst-case by Theorem~\ref{Thm:FixedDensityHardness} for $\rho_\epsilon(\Delta)<\rho\le \alpha_0/(1+\epsilon)$, where $\rho_\epsilon(\Delta):=\inf\{\rho>\alpha_c(\Delta):\alpha_-(\lambda(\rho))/\alpha_+(\lambda(\rho))<(1-\epsilon)/(1+\epsilon)\}$, while phase-aligned slices are not ruled out by our theorem.

It would be interesting to understand exactly which fixed slices are hard. One future direction is an understanding of the complement region. That is the following problem,

\begin{problem} \label{prob:reg}
    Fix an integer $\Delta\ge 3$ and $\alpha\in(\alpha_c(\Delta),1/2)$. If $(\alpha_L + \alpha_R)/2 = \alpha$ and
\begin{equation}
    \frac{\alpha_L}{\alpha_R} \in \bigg( 0, 
    \frac{\alpha_-(\lambda(\alpha))}{\alpha_+(\lambda(\alpha))} \bigg] \cup \bigg[    \frac{\alpha_+(\lambda(\alpha))}{\alpha_-(\lambda(\alpha))}, \infty \bigg),
\end{equation} does $\mathcal{I}_{\alpha_L, \alpha_R}$ admit an FPRAS and efficient sampling scheme?
\end{problem}

In particular, an answer to Problem \ref{prob:reg} would give an explicit boundary for the computability of the fixed-slice problem above the uniqueness threshold. Currently, it is unclear if the boundary is the true threshold for computability or an artifact of our proof.

Lastly, we remark that the comparison between Theorem \ref{Thm:FixedDensityHardness} and the algorithm of \cite{kocurek2026sampling} should be understood with the quantifiers in mind. The fixed-slice algorithmic results of \cite{kocurek2026sampling} are proved in the random regular setting and in the large-\(\Delta\) regime, whereas Theorem~\ref{Thm:FixedDensityHardness} is a worst-case hardness result valid for every fixed \(\Delta\ge 3\).

\subsection{Organization}

The rest of the paper is organized as follows. In Section~\ref{sec:below_uniqueness}, we prove correlation decay for the tilted measures, establish zero-freeness and a local central limit theorem for the balance variable, and combine these ingredients to obtain the sampler for $\mu^{\text{bal}}_{G,\lambda}$ below the uniqueness threshold.

In Section~\ref{sec:fptas_balanced_partition}, we prove the FPTAS below the uniqueness threshold. The proof recycles the sampling ingredients, replacing the noisy binary search by deterministic bisection and using Fourier inversion to approximate the probability of exact balance.

In Section~\ref{sec:hardness}, we prove the hardness result above the uniqueness threshold. We construct the bipartite gadgets, analyze their phase behavior under the balanced constraint, and give the reduction from the minimum bisection problem. Lastly, in Section~\ref{sec:density_hardness}, we study the fixed-slice model and prove hardness for approximately counting and sampling fixed slices at prescribed densities above the tree uniqueness density.

\section{\texorpdfstring{Sampling when \(\lambda < \lambda_c(\Delta)\)}{lambda below uniqueness}}
\label{sec:below_uniqueness}

The goal of this section is to prove the sampling assertion in
Theorem~\ref{thm:comp_thresh}. In fact, we will prove the assertion for the
slightly larger class of bipartite graphs, for a fixed constant $\gamma \geq 1$
\begin{equation}\label{eq:tilde_graph_class}
\widetilde{\cG}_{\Delta,\Delta}^{\gamma}
    :=
    \left\{
        G=(L\sqcup R,E):
        \Delta(G)\le \Delta,\quad
        L,R\neq\varnothing,\quad
        \frac{\max\{|L|,|R|\}}{\min\{|L|,|R|\}}\leq \gamma
    \right\}.
\end{equation}

\begin{proposition}
\label{prop:bounded_imbalance_sampler}
Fix \(\Delta\ge3\), \(\lambda<\lambda_c(\Delta)\), and $\gamma \geq 1$. Then the balanced
hard-core distribution \(\mu_{G,\lambda}^{\mathrm{bal}}\) admits an efficient
sampling scheme on inputs
\(G\in\widetilde{\cG}_{\Delta,\Delta}^{\gamma}\).
\end{proposition}

\begin{remark}
The proposition above strengthens the sampling assertion of
Theorem~\ref{thm:comp_thresh}, since
\(\cG_{\Delta,\Delta}\subseteq\widetilde{\cG}_{\Delta,\Delta}^{\gamma}\).
The same extension will be proved for the FPTAS in
Section~\ref{sec:fptas_balanced_partition}. The condition on the bipartition ratio in
\(\widetilde{\cG}_{\Delta,\Delta}^\gamma\) keeps the centering tilt in a  fixed compact interval.
Graphs whose bipartition sizes differ by a factor larger than, say, $\Delta^\Delta$ can  be
handled by combining the tilting framework with the algorithms for unbalanced
bipartite hard-core models of Cannon and
Perkins~\cite{cannon2019counting}; we omit that analysis here.
Alternatively, it seems plausible that the present approach could be extended
to tilts \(t=t(n)\) outside a fixed compact interval, but this would require
suitable modifications to the zero-freeness and local central limit arguments.
In particular, the uniform linear variance lower bound in
Lemma~\ref{lem:balance_variance_bounds} is then no longer available. We have
not pursued this extension.
\end{remark}

Throughout this section, fix
\(G \in\widetilde{\cG}_{\Delta,\Delta}^\gamma\) and write
\(n:=|V(G)|=|L|+|R|\). Recall from
Section~\ref{secIntro} the balance \(B(I)\) in
\eqref{eq:balance_definition}, the tilted partition function
\(Z_G(\lambda;t)\) in \eqref{eq:tiltedHardcorePartition_intro}, and its
associated tilted measure \(\mu_{G,\lambda,t}\). We write
\(\lambda_L(t)=\lambda e^t\), \(\lambda_R(t)=\lambda e^{-t}\), and often
abbreviate \(\mu_{G,\lambda,t}\) as \(\mu_{\lambda,t}\) when \(G\) is clear from context.

After possibly exchanging the two sides of the bipartition, we assume throughout
this section that \(|L|\ge |R|\). This exchange leaves the balanced independent
sets, and hence \(Z_G^{\mathrm{bal}}(\lambda)\), unchanged; it only sends
\(B\) to \(-B\) and \(t\) to \(-t\). For 
\(G\in\widetilde{\cG}_{\Delta,\Delta}^\gamma\), under this convention we have
\begin{equation}
    1\le \frac{|L|}{|R|}\leq\gamma.
    \label{eq:bounded_imbalance_regime}
\end{equation}

Our approach for sampling from \(\mu^{\mathrm{bal}}_{G,\lambda}\) is given in
Algorithm~\ref{alg:balanced_hardcore_sampler_template}. We defer the explicit
choice of parameters, including the tilt window \(\cT\) and the
\(\mathsf{TiltedSampler}\) subroutine, to
the subsequent subsections. Briefly, the algorithm first uses
empirical samples from the tilted measures \(\mu_{\lambda,t}\) to binary search
for a tilt \(\widetilde t\) whose expected balance is close to zero. We then
rejection sample from \(\mu_{\lambda,\widetilde t}\) until we obtain a sample
\(I\) with \(B(I)=0\).

The correctness of the rejection step follows directly from the definition of
the tilted measure. Under \(\mu_{\lambda,t}\), every independent set \(I\)
receives weight \(\lambda^{|I|}e^{tB(I)}\). On the event \(B(I)=0\), the tilt
factor is equal to \(1\), so the conditional law of \(\mu_{\lambda,t}\) given
\(B(I)=0\) exactly coincides with $\mu^{\text{bal}}_{G,\lambda}$.

\begin{algorithm}[!ht]
\caption{Balanced hard-core sampler template}
\label{alg:balanced_hardcore_sampler_template}
\begin{algorithmic}[1]
\Require A bipartite graph \(G=(L\sqcup R,E)\), activity
\(\lambda<\lambda_c(\Delta)\), accuracy parameter \(\epsilon>0\)
\Ensure A sample from the balanced hard-core model, up to total variation error
\(\epsilon\)
\State \label{line:balanced_sampler_template_regime_choice}Choose a tilt window \(\mathcal T\), a \(\mathsf{TiltedSampler}\) subroutine, centering tolerance \(\theta\), and parameters \(J,N,M,\tau\).
\State Run \(J\) steps of binary search on \(\mathcal T\) to find an approximate zero of \(m(t):=\E_{\mu_{\lambda,t}}B(I)\):
\Statex \begin{minipage}{0.95\linewidth}
\begin{enumerate}[leftmargin=2em,itemsep=0pt,topsep=2pt]
    \item At current midpoint \(t\) of $\cT$, call \(\mathsf{TiltedSampler}\) \(N\)
    times with TV error at most \(\tau\) to generate independent sets
    \(I_1,\dots,I_N\).
    \item Compute \(\widehat m=\frac{1}{N}\sum_{i=1}^N B(I_i)\).
    \item If \(|\widehat m|\le\theta\), set \(\widetilde t=t\) and stop; if
    \(\widehat m>0\), keep the lower half of the window; if \(\widehat m<0\),
    keep the upper half.
\end{enumerate}
\end{minipage}
\State If the search does not stop early, let \(\widetilde t\) be the midpoint of the final window.
\State Call \(\mathsf{TiltedSampler}\) \(M\) times with tilt \(\widetilde t\)
and total variation error at most \(\tau\), obtaining samples
\(I'_1,\dots,I'_M\).
\State Return the first \(I'_i\) satisfying \(B(I'_i)=0\). If no such sample
appears, return \(\varnothing\).
\end{algorithmic}
\end{algorithm}

At a high level, to prove that Algorithm~\ref{alg:balanced_hardcore_sampler_template} fulfills
the sampling assertion in Theorem~\ref{thm:comp_thresh}, we establish the
following three facts.
\begin{enumerate}[leftmargin=*]
    \item The binary search requires only polynomially many samples and returns a
    nearly centered tilt
    \(\widetilde t\), meaning that
    \(m(\widetilde t)=\E_{\mu_{\lambda,\widetilde t}}B(I)\) is sufficiently
    close to zero. This follows from monotonicity of \(m(t)\) and concentration
    of the empirical estimates used in the search.
    \item Every call to \(\mathsf{TiltedSampler}\) can be implemented in
    polynomial time for the tilts queried by the search. For
    \(G\in\widetilde{\cG}_{\Delta,\Delta}^\gamma\), these tilts remain in a fixed compact
    window. Strong spatial mixing on the SAW tree throughout this window then
    yields an efficient sampler for the tilted measure.
    \item The rejection step succeeds with polynomially many trials. Concretely,
    we show the key acceptance probability estimate
        \[
        \Pr_{\mu_{\lambda,\widetilde{t}}}\{B(I)=0\}
        =
        \Omega\!\left(\frac1{\sqrt n}\right).
    \]
    This will follow from a zero-freeness result of the tilted partition function
    \(Z_G(\lambda;t)\) in a complex neighborhood of the compact tilt window,
    which yields a local central limit theorem for \(B(I)\). Consequently,
    \(M=\Theta(\sqrt n\log(1/\epsilon))\) rejection trials suffice.
\end{enumerate}

Point 1 is relatively straightforward and is treated in
Section~\ref{sec:sampling_preliminaries}. Points 2 and 3 constitute the main
work and are treated in the subsequent subsections.

\subsection{Preliminaries}
\label{sec:sampling_preliminaries}

We record two elementary facts. The first controls the
empirical means appearing in the search step of Algorithm \ref{alg:balanced_hardcore_sampler_template}.

\begin{lemma}
\label{lem:empirical-balance-concentration}
Fix a bipartite graph \(G=(L\sqcup R,E)\) on \(n\) vertices, a tilt
\(t\in\mathbb R\), an accuracy parameter \(\eta>0\), and a confidence
parameter \(\xi\in(0,1)\). Let
$
N\ge 2n^2\eta^{-2}\log(2/\xi),
$
and let \(I_1,\dots,I_N\) be independent (exact) samples from the tilted hard-core
measure \(\mu_{\lambda,t}\). Then
\[
\Pr\!\left(\left|\frac1N\sum_{j=1}^N B(I_j)-\E_{\mu_{\lambda,t}}[B(I)]\right|\ge \eta\right)
\le \xi.
\]
\end{lemma}
\begin{proof}
For each \(j\), the random variable \(B(I_j)\) takes values in the interval \([-n,n]\), since
$
-n\le |I_j\cap L|-|I_j\cap R|\le n.
$
Therefore, using that $\sum_{j=1}^N (2n)^2 = 4Nn^2$, Hoeffding's inequality gives
\[
\Pr\!\left(\left|\frac1N\sum_{j=1}^N B(I_j)-\E_{\mu_{\lambda,t}}[B(I)]\right|\ge \eta\right)
\le
2\exp\!\left(
-\frac{2N^2\eta^2}{\sum_{j=1}^N (2n)^2}
\right)
\le \xi.\qedhere
\]
\end{proof}

Our next result records the existence and uniqueness of the centering tilt. Here and throughout, we use the notation
\begin{equation}\label{eq:m_expectedBalance_def}
    m(t):=\E_{\mu_{\lambda,t}}[B(I)].
\end{equation}

\begin{lemma}\label{lem:unique_centering_tilt}
For every bipartite graph
$G=(L\sqcup R,E)$ with maximum degree at most $\Delta$ and with
\(L,R\neq\varnothing\), and every $\lambda>0$, there exists a unique \(t^*\in\mathbb R\) such that
\[
m(t^*)=0.
\]
\end{lemma}

\begin{proof}
By direct differentiation, we have $m'(t) = \operatorname{Var}_{\mu_{\lambda,t}}(B(I)) =: \sigma_t^2$.
Since $B$ is not constant when $L, R \ne  \emptyset$, \(\sigma_t^2>0\). As \(t\to-\infty\), we have
\(m(t)\to-|R|<0\), while as \(t\to\infty\), we have
\(m(t)\to|L|>0\). Thus, continuity and strict monotonicity give a unique
\(t^*\) such that \(m(t^*)=0\).
\end{proof}

\subsection{Marginal occupation probabilities via the SAW tree}
\label{sec:sampling_from_tilted}

We begin by showing that there exists a poly-time algorithm that can approximate the marginal occupation probabilities of $\mu_{\lambda,t}$ within additive error $\pm \frac{\epsilon}{n}$. By standard self-reducibility arguments (see e.g.~Weitz~\cite[Section 5]{weitz2006counting}; cf.~Algorithm~\ref{alg:tilted-sampler} and the discussion following it), this will yield an $\epsilon$-approximate poly-time sampler. Our algorithm for approximating the marginals will be via the self-avoiding walk (SAW) tree~\cite{weitz2006counting}.

In what follows, we write $\boldsymbol\lambda=(\lambda_u)_{u\in V}$ for a general vector of vertex activities and $\mu_{G,\boldsymbol\lambda}$ for the corresponding hard-core measure. The tilted measure $\mu_{\lambda,t}$ on bipartite $G = (L \sqcup R, E)$ is the special case obtained by taking $\lambda_u=\lambda_L(t)$ on $L$ and $\lambda_u=\lambda_R(t)$ on $R$.

\begin{definition}[Depth-$L$ truncated SAW tree]
Fix a graph $G=(V,E)$, a root $v\in V$, vertex activities
$(\lambda_u)_{u\in V}$, and an ordering of the neighbors of each vertex.
The self-avoiding walk tree $T_{\mathrm{SAW}}(G,v)$ is the tree of all paths originating at root $v = v_0$, where each path is built recursively as follows. Given a self-avoiding path $(v_0,\dots,v_k)$, consider the neighbors $w$ of $v_{k}$ except the predecessor $v_{k-1}$. If there are no neighbors, the path terminates and $v_k$ is an ordinary leaf. If $w \notin \inbraces{v_0,\dots,v_k}$, extend the path to $(v_0,\dots,v_k,w)$. If instead $w=v_j$ for some $j<k$, the child is a
(terminal) boundary leaf. This boundary leaf is fixed to be occupied if, in the
ordering at $w$, the closing edge $\{v_k,w\}$ is larger than the edge
$\{w,v_{j+1}\}$ that starts the cycle, and is fixed to be unoccupied
otherwise. Each non-boundary copy of a vertex $u$ in $T_{\mathrm{SAW}}(G,v)$ has activity $\lambda_u$.

For integer $L\ge 0$, the depth-$L$ truncation $T_{\mathrm{SAW}}^{(L)}(G,v)$ is
obtained from $T_{\mathrm{SAW}}(G,v)$ by keeping only vertices at distance
at most $L$ from the root, together with the induced occupied/unoccupied boundary
conditions on any retained boundary leaves.
\end{definition}

A boundary condition $\tau$ is an assignment of occupied or unoccupied
states to a subset of vertices of $T_{\mathrm{SAW}}(G,v)$ not containing
the root.

Let
$\mu_{T_{\mathrm{SAW}}(G,v),\boldsymbol\lambda}^{\tau}$ denote the
hard-core measure on $T_{\mathrm{SAW}}(G,v)$ conditioned on $\tau$,
together with the occupied/unoccupied boundary conditions on the
cycle-closing leaves specified in the construction of the SAW tree.
Define the root occupation ratio by
\[
R_v^{\tau}
:=
\frac{
\mu_{T_{\mathrm{SAW}}(G,v),\boldsymbol\lambda}^{\tau}(\sigma_v=1)
}{
\mu_{T_{\mathrm{SAW}}(G,v),\boldsymbol\lambda}^{\tau}(\sigma_v=0)
}.
\]
When no additional boundary condition $\tau$ is imposed, we write $R_v$.
Thus,
\[
\mu_{T_{\mathrm{SAW}}(G,v),\boldsymbol\lambda}^{\tau}(\sigma_v=1)
=
\frac{R_v^\tau}{1+R_v^\tau}.
\]

By the SAW-tree identity of
Weitz~\cite[Section 3]{weitz2006counting}, the occupation ratio $R_v$ at the root of
$T_{\mathrm{SAW}}(G,v)$ equals the occupation ratio at $v$
in $G$, which is 
$
\mu_{G,\boldsymbol\lambda}(\sigma_v=1)/
\mu_{G,\boldsymbol\lambda}(\sigma_v=0).
$

The truncated SAW tree is used to approximate marginal probabilities in Algorithm \ref{alg:saw-trunc}.
\begin{algorithm}
\caption{Truncated SAW-tree marginal oracle}\label{alg:saw-trunc}
\begin{algorithmic}[1]
\Require A graph $G=(V,E)$, activities $(\lambda_u)_{u\in V}$,
a vertex $v\in V$, an accuracy parameter $\epsilon>0$.
\Ensure An estimate $\widetilde p_v$ for $\mu_{G,\boldsymbol\lambda}(\sigma_v=1)$

\State Choose a truncation depth
$
L\ge \frac1c \log\!\left(\frac{Cn}{\epsilon}\right)
$ for suitable constants $C,c > 0$ and construct the depth-$L$ truncated SAW tree $T_{\mathrm{SAW}}^{(L)}(G,v)$.

\State Fix every non-boundary vertex at depth exactly \(L\) to be unoccupied.

\State Compute occupation ratios bottom-up on
$T_{\mathrm{SAW}}^{(L)}(G,v)$ using
$
\displaystyle
R_x=\lambda_x\prod_{y\in \mathrm{children}(x)}\frac{1}{1+R_y},
$
with $R_x=0$ for leaves fixed unoccupied and $R_x=\infty$ for leaves fixed occupied.

\State Let $\widetilde R_v$ be the ratio computed at the root and \Return
$
\widetilde p_v=\frac{\widetilde R_v}{1+\widetilde R_v}.
$
\end{algorithmic}
\end{algorithm}

It is clear that Algorithm~\ref{alg:saw-trunc} is computationally efficient. Our next step is to show that it is applicable. We will accomplish this by showing that replacing
the full SAW tree by its depth-$L$ truncation has a negligible effect on the root ratio. This will follow from \emph{strong spatial
mixing} (SSM), which we define next.

\begin{definition}\label{def:ssm}
Fix a vertex $v\in V(G)$. We say that the model satisfies strong spatial mixing if there exist constants $C, c > 0$ such that whenever two boundary conditions $\tau$ and $\tau'$ on $T_{\mathrm{SAW}}(G,v)$ first differ only at distance at least $\ell$ from the root, their induced root occupation ratios differ by at most
\begin{equation}\label{eq:SSM}
|R_v^\tau-R_v^{\tau'}|
\le
C e^{-c\ell}.
\end{equation}
\end{definition}

\begin{remark}\label{remark:SSM_to_marginalOracle}
Once SSM \eqref{eq:SSM} is established, Algorithm~\ref{alg:saw-trunc} can be used as an $\epsilon/n$ marginal oracle for the root occupation probabilities $\mu_{G,\boldsymbol\lambda}(\sigma_v=1)=\frac{R_v}{1+R_v}$. This is seen as follows.

Let \(\tau^{(L)}\) denote the boundary condition on the full SAW tree
that fixes every non-boundary vertex at distance exactly \(L\) from the
root to be unoccupied. Then the truncated computation produces the root ratio $R_v^{\tau^{(L)}}$, while the true marginal corresponds to the root ratio $R_v$ on the full SAW tree. Since these two boundary conditions first differ only at distance at least $L$ from the root, \eqref{eq:SSM} applies and we have
\[
|R_v^{\tau^{(L)}}-R_v|
\le
C e^{-cL}.
\]

Since the map $R\mapsto R/(1+R)$ is $1$-Lipschitz on $[0,\infty)$, the same bound holds for the corresponding marginal:
$
    \left|
        \widetilde p_v-\mu_{G,\boldsymbol\lambda}(\sigma_v=1)
    \right|
    \le
    C e^{-cL}.
$
Therefore, choosing $
L\ge \frac{1}{c}\log\!\left(\frac{C n}{\epsilon}\right)$
ensures additive error at most $\epsilon/n$. Finally, because the maximum degree is at most $\Delta$, the depth-$L$ truncated SAW tree has at most
$
1+\Delta\sum_{j=0}^{L-1}(\Delta-1)^j
$
vertices, and therefore has size $O(\Delta^L)$. Since $\Delta$ is fixed and $L=O(\log(n/\epsilon))$, Algorithm~\ref{alg:saw-trunc} runs in time polynomial in $n/\epsilon$.
\end{remark}

We now specialize the preceding discussion to the tilted hard-core model \eqref{eq:tiltedHardcorePartition_intro} and verify \eqref{eq:SSM}.

\begin{theorem} \label{thm:ssm}
Fix $\Delta \ge 3$, $\lambda<\lambda_c(\Delta)$, and $t \in \mathbb R$. Then there exist constants $C, c > 0$, depending only on $(\Delta,\lambda,t)$, such that for every bipartite graph \(G=(L\sqcup R,E)\) of maximum degree at most
\(\Delta\), every vertex $v \in V(G)$, if $\tau,\tau'$ are two boundary conditions on $T_{\mathrm{SAW}}(G,v)$ and $\mathcal D\subseteq V(T_{\mathrm{SAW}}(G,v))$ is their disagreement set, then
\[
|R_v^\tau-R_v^{\tau'}|
\le
C e^{-c\, d_{T_{\mathrm{SAW}}}(v,\mathcal D)},
\]
where \(d_{T_{\mathrm{SAW}}}(v,\mathcal D)\) denotes the minimum distance in $T_{\mathrm{SAW}}(G,v)$ from the root \(v\) to a vertex of \(\mathcal D\).

Moreover, for every compact interval \(K\subset\mathbb R\), the constants
\(C,c\) may be chosen depending only on \(\Delta,\lambda,K\), uniformly
for all \(t\in K\).
\end{theorem}

The rest of this section is devoted to proving Theorem \ref{thm:ssm}. We introduce some notation towards this end. On a rooted tree, suppose a vertex \(u\) is in \(\dagger \in \inbraces{L,R}\), has activity \(\lambda_\dagger\), and has $d$ children with occupation ratios \(R_1,\dots,R_d\), where $d \leq \Delta-1$. Then the hard-core recursion is
\begin{equation}\label{eq:tree_ratio_recursion_tilted}
R_u
=
\lambda_\dagger\prod_{i=1}^d\frac{1}{1+R_i}.
\end{equation}
We will study this recursive  map abstractly, so it is useful to introduce the one-level update maps
\[
F_{\dagger,d}(x_1,\dots,x_d)
:=
\lambda_\dagger\prod_{i=1}^d\frac{1}{1+x_i},
\qquad\text{and}\qquad
f_{\dagger,d}(x)
:=
\frac{\lambda_\dagger}{(1+x)^d},
\]
where \(f_{\dagger,d}\) is the specialization of \(F_{\dagger,d}\) when all child ratios are equal.

We will use the coordinate change
\[
\phi(x):=\operatorname{arcsinh}\sqrt{x},
\]
and show that the recursive maps are contractive in this new coordinate system.

The one-level update in \(\phi\)-coordinates for the equal-input setting is $g_{\dagger,d} = \phi \circ f_{\dagger,d} \circ \phi^{-1}$, i.e.
\[
g_{\dagger,d}(w)
=
\operatorname{arcsinh}\!\left(\sqrt{\lambda_\dagger}\,\operatorname{sech}^d w\right).
\]

Intuitively, it will be convenient to study contractivity of a self-map that goes from $\dagger \in \inbraces{L,R}$ to $\dagger^c$ and back to $\dagger$. On the tree, this corresponds to a two-level map going in a boundary-to-root direction. Suppose a particular vertex $v$ is in \(\dagger\), has \(a\) children, and each child has \(b\) children (so grandchildren of $v$). If all grandchild ratios are equal to \(x_1\), then the intermediate child ratio and the ratio at $v$ are respectively
\[
y_0:=f_{\dagger^c,b}(x_1)
=
\lambda_{\dagger^c}(1+x_1)^{-b},
\qquad
x_0:=f_{\dagger,a}(y_0)
=
\lambda_\dagger(1+y_0)^{-a}.
\]
In \(\phi\)-coordinates, define the two-level map
\[
h_{\dagger,a,b}
:=
g_{\dagger,a}\circ g_{\dagger^c,b}
=
\phi\circ f_{\dagger,a}\circ f_{\dagger^c,b}\circ\phi^{-1}.
\]
The derivative of $h_{\dagger,a,b}$ at \(w=\phi(x_1)\), in terms of the original coordinates, is
\begin{equation}\label{eq:SSM_hPrime_ab_derivative}
h'_{\dagger,a,b}(\phi(x_1))
=
ab \sqrt{\frac{x_1}{1+x_1}} \frac{y_0}{1+y_0}
\sqrt{\frac{x_0}{1+x_0}}.
\end{equation}
Our analysis centers around this function \(h'_{\dagger,a,b}\), which measures how much a perturbation at the grandchildren can affect the root after passing through two recursive updates.

The proof of Theorem \ref{thm:ssm} has three main ingredients. Proposition \ref{prop:two-level-equal-input} starts with the general multivariate two-level update obtained from the maps \(F_{\dagger,a}\) and \(F_{\dagger^c,b_i}\), with arbitrary grandchild ratios, and bounds its Jacobian norm by the equal-input function $h_{\dagger,a,b}$. Proposition \ref{prop:phi_contraction} then proves that the derivatives of $h_{\dagger,a,b}$ are uniformly smaller than \(1\) whenever \(\lambda<\lambda_c(\Delta)\) over the tilt parameter $t$. Finally, Theorem \ref{thm:ssm} is proved by iterating this two-level contraction to give exponential decay in the \(\phi\)-coordinates, whence Lemma \ref{lem:dif_coords_suffice} converts that decay back to the original coordinates.

We define analogous notation for the non-equal input recursions. Consider a rooted two-level bipartite tree fragment in which the root belongs to \(\dagger\), has $a$ children, and each child $i$ has $b_i$ children, with corresponding grandchild ratios \(x_{1,i,1},\dots,x_{1,i,b_i}\ge 0\). Define the child ratios $y_1,\dots,y_a$ and root ratio $x_0$ by
\begin{align}\label{eq:SSM_unequal_yi_x0_def}
    y_i=\lambda_{\dagger^c}\prod_{j=1}^{b_i}(1+x_{1,i,j})^{-1},
    \qquad
    x_0=\lambda_\dagger\prod_{i=1}^a(1+y_i)^{-1}.
\end{align}
In the corresponding \(\phi\)-coordinates write the grandchild, child, and root ratios as
\begin{align}\label{eq:SSM_unequal_qr_phiCoords_def}
  q_{1,i,j}=\phi(x_{1,i,j}),\qquad r_i=\phi(y_i),\qquad\text{and}\qquad q_0=\phi(x_0).
\end{align}
Given an arbitrary input of unequal number of grandchildren and unequal grandchild ratios $\inbraces{x_{1,i,j}}_{1 \leq i \leq a, 1 \leq j \leq b_i}$, we will now outline a construction of a corresponding equal-input effective grandchild ratio $x_1$. We also set \(b=\max_{1\le i\le a} b_i\), and set \(b=0\) if \(a=0\). In Proposition \ref{prop:two-level-equal-input} the influence of the grandchild ratios on the root ratio for arbitrary unequal inputs will be shown to be bounded by that of the equal input setting with grandchild ratio $x_1$, i.e.~by the magnitude of $h'_{\dagger,a,b}(\phi(x_1))$. In the case where \(a,b>0\) and both activities are positive, set
\begin{align}\label{eq:SSM_unequalRatio_to_equalRatio_x1_def}
    x_1 = \inparen{\frac{\lambda_{\dagger^c}}{y_0}}^{1/b} - 1, \qquad\text{where}\qquad y_0 = \left(\prod_{i=1}^a(1+y_i)\right)^{1/a} - 1.
\end{align}
This $y_0$ is chosen to preserve the root ratio so that $x_0 = \lambda_\dagger(1+y_0)^{-a}$, and the equal-input grandchild ratio $x_1$ is chosen to preserve this child ratio $y_0$ so that $y_0 = \lambda_{\dagger^c} (1+x_1)^{-b}$. In the degenerate cases \(a=0\), \(b=0\), or one of the two activities is \(0\), we take \(x_1=0\).

\begin{proposition}\label{prop:two-level-equal-input}
On a rooted two-level bipartite tree fragment with root in $\dagger$ having $a$ children, and each child $i$ having $b_i$ children, suppose the grandchild ratios are $\inbraces{x_{1,i,j}}_{1 \leq i \leq a, 1 \leq j \leq b_i}$. Then with $x_1$ defined in \eqref{eq:SSM_unequalRatio_to_equalRatio_x1_def}, we have
\[
\sum_{i=1}^a\sum_{j=1}^{b_i}
\left|
\frac{\partial q_0}{\partial q_{1,i,j}}
\right|
\le
h'_{\dagger,a,b}(\phi(x_1)).
\]
\end{proposition}

\begin{proposition}\label{prop:phi_contraction}
For any $\lambda < \lambda_c(\Delta)$, we have
\[
\max_{\dagger\in\{L,R\}}\max_{0\le a,b\le \Delta-1}\ \sup_{t \in \RR} \sup_{x_1\ge 0}h'_{\dagger,a,b}(\phi(x_1))<1.
\]
\end{proposition}

\begin{lemma} \label{lem:dif_coords_suffice}
Suppose there exist constants \(C_0,c>0\) such that for any two boundary conditions \(\tau,\tau'\),
\[
\bigl|\phi(R_v^\tau)-\phi(R_v^{\tau'})\bigr|
\le
C_0 e^{-c\, d(v,\mathcal D)},
\]
where \(\mathcal D\) is the set on which the boundary conditions differ. Suppose also that there is a uniform bound $
0\le R_v^\tau,\; R_v^{\tau'} \le R_{\max}$.
Then there exists a constant \(C' = C'(R_{\max},C_0) > 0\) such that
\[
|R_v^\tau-R_v^{\tau'}|
\le
C' e^{-c\, d(v,\mathcal D)}.
\]
\end{lemma}

\begin{proof}
Since \(R_v^\tau,R_v^{\tau'}\in[0,R_{\max}]\) and the inverse map \(\phi^{-1}(s)=\sinh^2s\) is smooth on the compact interval \([0,\phi(R_{\max})]\), we have
$
L_{R_{\max}}:=\sup_{0\le s\le \phi(R_{\max})}\left|\left(\phi^{-1}\right)'(s)\right|<\infty.
$
The mean value theorem gives
\[
\bigl|R_v^\tau-R_v^{\tau'}\bigr|
\le
L_{R_{\max}}\,\bigl|\phi(R_v^\tau)-\phi(R_v^{\tau'})\bigr|
\le
L_{R_{\max}} C_0 e^{-c\,d(v,\mathcal D)}. \qedhere
\]
\end{proof}

\begin{proof}[Proof of Theorem \ref{thm:ssm}]
In this proof, abbreviate \(d:=\Delta-1\) and
\(\Lambda:=\max\{\lambda_L(t),\lambda_R(t)\}\).

Fix $v\in V(G)$ and let \(T_{\mathrm{SAW}}=T_{\mathrm{SAW}}(G,v)\). We first prove exponential decay in the \(\phi\)-coordinates on rooted bipartite trees of forward degree at most \(d\), and then apply this estimate to the descendant subtrees of \(T_{\mathrm{SAW}}\).

We use the recursion two levels at a time. Set
$
M:=\max_{\dagger}\max_{0\le a,b\le d}
\sup_{s\in\mathbb R}\sup_{x_1\ge0}
h'_{\dagger,a,b}(\phi(x_1)),
$
where the fugacities in \(h_{\dagger,a,b}\) are evaluated at tilt \(s\).
By Proposition~\ref{prop:phi_contraction}, \(M<1\), and
Proposition~\ref{prop:two-level-equal-input} implies that any two-level
update whose root and children have at most \(d\) children satisfies
\begin{equation}\label{eq:two_level_phi_jacobian_bound}
\sum_{i,j}
\left|
\frac{\partial q_0}{\partial q_{1,i,j}}
\right|
\le M.
\end{equation}

Let \(T\) be any rooted bipartite tree of forward degree at most \(d\).
Although a vertex fixed to be occupied by a boundary condition has ratio
\(+\infty\), every vertex \(u\) not fixed by the boundary condition
satisfies \(0\le R_u\le\lambda_u\le\Lambda\). For \(u\in T\) and
\(h\ge1\), define
\[
D_h(u):=
\sup_{\eta,\eta'}
\bigl|\phi(R_u^\eta)-\phi(R_u^{\eta'})\bigr|,
\]
where the supremum is over all pairs of boundary conditions
\(\eta,\eta'\) supported on descendants of \(u\), and which agree
on every vertex within graph distance strictly less than \(h\) from
\(u\). Thus \(D_h(u)\) measures the effect of boundary conditions below
a vertex \(u\) which is not itself fixed.

We claim that for every \(u\) and every \(h\ge1\),
\begin{equation}\label{eq:Dh_recursion}
D_{h+2}(u)\le M\max_{w:\,\operatorname{dist}_T(u,w)=2}D_h(w),
\end{equation}
with the right-hand side interpreted as \(0\) if \(u\) has no grandchildren.

To prove the claim, fix a pair \(\eta,\eta'\) appearing in the
supremum defining \(D_{h+2}(u)\). These boundary conditions agree at
every child and grandchild of \(u\). If a child of \(u\) is fixed to be
occupied, then \(R_u^\eta=R_u^{\eta'}=0\), and there is nothing to
prove. A child fixed to be unoccupied has ratio \(0\) under both
boundary conditions, and its factor in the recursion at \(u\) is
\((1+0)^{-1}=1\), so that child branch may be omitted.

Now let \(w\) be a grandchild of \(u\) whose parent is not fixed. Since
\(w\) is at distance \(2\) from \(u\), the restrictions of
\(\eta,\eta'\) to the descendant subtree rooted at \(w\) agree on every
vertex within distance strictly less than \(h\) from \(w\). If \(w\) is
not fixed, these restrictions are supported on proper descendants of
\(w\), and hence form a pair included in the supremum defining
\(D_h(w)\). Otherwise, \(w\) is fixed to be occupied by both boundary
conditions or fixed to be unoccupied by both.

In the notation of \eqref{eq:SSM_unequal_yi_x0_def}, if \(w\) is fixed
to be unoccupied, then \(x_{1,i,j}=0\), so its factor
\((1+x_{1,i,j})^{-1}=1\) in the recursion for \(y_i\) may be omitted. If
\(w\) is fixed to be occupied, then \(x_{1,i,j}=+\infty\), so this
factor is \(0\) and hence \(y_i=0\); consequently,
\((1+y_i)^{-1}=1\) in the recursion for \(x_0\), and the entire child
branch may be omitted. Thus every remaining grandchild ratio is finite,
and its two \(\phi\)-coordinates differ by at most \(D_h(w)\).

If \(u\) has no grandchildren there is nothing to show; otherwise write \(w_{i,j}\) for the grandchild corresponding to the ratio \(q_{1,i,j}\). Let \(\mathbf q_1\) and \(\mathbf q_1'\) be the two grandchild (relative to $u$) \(\phi\)-coordinate vectors induced by \(\eta\) and \(\eta'\), and write \(q_0(\mathbf q_1)\) for the root \(\phi\)-coordinate obtained from the two-level recursion. The fundamental theorem of calculus along \(\gamma(t):=\mathbf q_1'+t(\mathbf q_1-\mathbf q_1')\), \(0\le t\le1\), together with \eqref{eq:two_level_phi_jacobian_bound}, gives
\begin{align*}
\bigl|\phi(R_u^\eta)-\phi(R_u^{\eta'})\bigr|&=|q_0(\mathbf q_1)-q_0(\mathbf q_1')|
= \abs{\int_0^1 \sum_{i,j}\frac{\partial q_0}{\partial q_{1,i,j}}(\gamma(t))(q_{1,i,j}-q'_{1,i,j})\,dt}
\le M\max_{i,j}D_h(w_{i,j}).
\end{align*}
Taking the supremum over admissible \(\eta,\eta'\) proves \eqref{eq:Dh_recursion}.

For \(h\ge1\), define \(A_h:=\sup_{T,u}D_h(u)\), where the supremum
runs over all rooted bipartite trees of forward degree at most \(d\),
with the given activities, and all vertices \(u\). Then
\eqref{eq:Dh_recursion} gives
\[
A_{h+2}\le M A_h,\qquad h\ge1.
\]
Since the boundary conditions do not fix \(u\), we have
\(0\le R_u^\eta\le\Lambda\), and hence \(A_1\le\phi(\Lambda)\).
Since \(A_h\) is nonincreasing in \(h\), for every \(h\ge1\),
\begin{equation}\label{eq:Ah_decay}
A_h\le A_1M^{\lfloor(h-1)/2\rfloor}.
\end{equation}
Since \(M<1\) by Proposition~\ref{prop:phi_contraction}, the right-hand
side decays exponentially.

It remains to return to the actual SAW tree. The root \(v\) of \(T_{\mathrm{SAW}}\) may have as many as \(\Delta\) children, but each child subtree has forward degree at most \(d=\Delta-1\). Let \(\ell:=d_{T_{\mathrm{SAW}}}(v,\mathcal D)\). If \(\mathcal D=\varnothing\), there is nothing to prove. Otherwise,
since neither boundary condition fixes \(v\), we have \(\ell\ge1\).
The case \(\ell=1\) is covered by the bound
\(|\phi(R_v^\tau)-\phi(R_v^{\tau'})|\le\phi(\Lambda)\), after increasing
the prefactor. If \(\ell\ge2\), we first apply a one-level mean value
bound at the root (this nuisance arises because the SAW tree root may
have \(\Delta\) children, instead of \(\Delta-1\) children.) Since
\(\ell\ge2\), any child fixed by either boundary condition is fixed in
the same way by both. A child fixed to be occupied under both conditions
forces both root ratios to be \(0\), while a child fixed to be
unoccupied under both may be omitted from the recursion. Thus we may
assume that the remaining child roots are not fixed, and so their ratios
lie in \([0,\Lambda]\). A direct one-level derivative calculation shows
that the \(\ell_\infty\) Lipschitz constant of the root update in
\(\phi\)-coordinates is at most
$
L_1:=\Delta\,\frac{\Lambda}{1+\Lambda}.
$
The boundary conditions induced on each remaining child subtree first differ at distance at least \(\ell-1\) from that child. Therefore, using \eqref{eq:Ah_decay},
\[
\bigl|\phi(R_v^\tau)-\phi(R_v^{\tau'})\bigr|
\le
L_1 A_{\ell-1}
\le
C_0 e^{-c\,\ell}
\]
for constants \(C_0<\infty\) and \(c>0\) depending only on \((\Delta,\lambda,t)\). Lemma \ref{lem:dif_coords_suffice}, applied with \(R_{\max}=\Lambda\), converts this estimate from the \(\phi\)-coordinates back to original ratios, yielding the desired
$
|R_v^\tau-R_v^{\tau'}|
\le
C e^{-c\,d_{T_{\mathrm{SAW}}}(v,\mathcal D)}.
$
Finally, if \(t\) ranges over a fixed compact interval \(K\), then
\(\Lambda\) is uniformly bounded, while the contraction factor \(M<1\)
is independent of \(t\) by Proposition~\ref{prop:phi_contraction}.
Hence all the constants above may be chosen uniformly for \(t\in K\).
\end{proof}

\paragraph{Proof of Proposition~\ref{prop:two-level-equal-input}.}
\begin{proof}[Proof of Proposition~\ref{prop:two-level-equal-input}]
If \(a=0\), \(b=0\), or one of the two activities is \(0\), there is nothing to prove. Pad each child \(i\) with dummy grandchildren \(x_{1,i,j}=0\) for \(j>b_i\). This does not change \(y_i\), \(x_0\), or the Jacobian sum. Thus we may assume without loss every child has exactly \(b\) grandchildren.

A direct chain-rule computation gives
\begin{equation} \label{eq:l1_norm_jacobian}
\sum_{i=1}^a \sum_{j=1}^b
\left|
\frac{\partial q_0}{\partial q_{1,i,j}}
\right|
=
\sqrt{\frac{x_0}{1+x_0}}
\sum_{i=1}^a
\frac{y_i}{1+y_i}
\sum_{j=1}^b
\sqrt{\frac{x_{1,i,j}}{1+x_{1,i,j}}}.
\end{equation}
The argument to pass to equal inputs is essentially two applications of Jensen's inequality.

First fix \(i\). We will use a change of variables \(u_j:=\log(1+x_{1,i,j})\). Define $F(u) = \sqrt{1-e^{-u}}$ and note that $F$ is concave on \(u\ge0\). Thus by Jensen's inequality, we have
\[
\sum_{j=1}^b \sqrt{\frac{x_{1,i,j}}{1+x_{1,i,j}}}
= \sum_{j=1}^{b} F(u_j) \le
b \cdot F\!\inparen{  \frac{1}{b} \sum_{j=1}^{b} u_j }
= b\sqrt{\frac{x_{1,i}}{1+x_{1,i}}},
\]
where we define
\begin{align*}
    x_{1,i}:=\left(\prod_{j=1}^b(1+x_{1,i,j})\right)^{1/b} - 1.
\end{align*}
We may understand $x_{1,i}$ as the equalized grandchild ratios of the $i$-th child, chosen so that the ratio $y_i = \lambda_{\dagger^c}(1+x_{1,i})^{-b}$ is preserved. Thus \eqref{eq:l1_norm_jacobian} becomes
\begin{equation} \label{l1_norm_jacobian_one_sub}
\sum_{i=1}^a\sum_{j=1}^b
\left|
\frac{\partial q_0}{\partial q_{1,i,j}}
\right|
\le
b\sqrt{\frac{x_0}{1+x_0}}
\sum_{i=1}^a
\frac{y_i}{1+y_i}\sqrt{\frac{x_{1,i}}{1+x_{1,i}}}.
\end{equation}

Next set \(v_i:=\log(1+y_i)\). The $i$-th summand in \eqref{l1_norm_jacobian_one_sub} is
\[
K(v_i):=(1-e^{-v_i})
\sqrt{1-\left(\frac{e^{v_i}-1}{\lambda_{\dagger^c}}\right)^{1/b}}.
\]
We claim that $K$ is concave. This is seen as follows. Define
$
\alpha=\left(\frac{e^v-1}{\lambda_{\dagger^c}}\right)^{1/b}\in(0,1),
$
then one has
\[
K''(v)
=
-\frac{N(\alpha)}
{4\lambda_{\dagger^c} b^2 \alpha^b (1-\alpha)^{3/2}(1+\lambda_{\dagger^c}\alpha^b)},
\]
where
\[
\begin{aligned}
N(\alpha)
&=
\lambda_{\dagger^c}^2 \alpha^{2b+1}(2-\alpha)
+4\lambda_{\dagger^c} b^2 \alpha^b(1-\alpha)^2
+2\lambda_{\dagger^c} b\, \alpha^{b+1}(1-\alpha) \\
&\qquad\qquad\qquad
+2\lambda_{\dagger^c} \alpha^{b+1}(2-\alpha)
+\alpha\bigl(2b(1-\alpha)+(2-\alpha)\bigr)>0.
\end{aligned}
\]
Thus by Jensen's inequality again,
\[
\sum_{i=1}^a
\frac{y_i}{1+y_i}\sqrt{\frac{x_{1,i}}{1+x_{1,i}}}
\le
a\,\frac{y_0}{1+y_0}\sqrt{\frac{x_1}{1+x_1}},
\]
where $y_0$ and $x_1$ are defined as in \eqref{eq:SSM_unequalRatio_to_equalRatio_x1_def}.
Note that \(x_0=\lambda_\dagger(1+y_0)^{-a}\). Substituting into \eqref{l1_norm_jacobian_one_sub}, we see from \eqref{eq:SSM_hPrime_ab_derivative} that this yields the desired
\[
\sum_{i=1}^a\sum_{j=1}^b
\left|
\frac{\partial q_0}{\partial q_{1,i,j}}
\right|
\le
ab \sqrt{\frac{x_1}{1+x_1}} \frac{y_0}{1+y_0}
\sqrt{\frac{x_0}{1+x_0}}
=h'_{\dagger,a,b}(\phi(x_1)). \qedhere
\]
\end{proof}

\paragraph{Proof of Proposition~\ref{prop:phi_contraction}.}

In light of Proposition \ref{prop:two-level-equal-input}, to show contraction of the two-level map for arbitrary grandchild ratios, we can focus our attention on the equal-input scalar derivative \(h'_{\dagger,a,b}\) in \eqref{eq:SSM_hPrime_ab_derivative}. It is equivalent and slightly cleaner to bound the square of this derivative, \(\bigl(h'_{\dagger,a,b}(\phi(x_1))\bigr)^2\), which can be written as
\begin{equation} \label{eq:theta_defn}
\Theta_{a,b,\lambda}(y_0,x_1)
:=
a^2b^2\,
\frac{\lambda^2 y_0^2 x_1}
{(1+y_0)^2(1+x_1)\bigl(\lambda^2+y_0(1+y_0)^a(1+x_1)^b\bigr)}.
\end{equation}

For fixed \(t\) and \(x_1\), the value of \(y_0\) is determined. The point of the next lemma is that, after taking the supremum over all tilts, we may use \(y_0\) itself as the free parameter.

\begin{lemma} \label{lem:h_dagger_theta_reduction}
Let $\lambda>0$ and $\dagger \in \inbraces{L,R}$ be fixed. Then,
\[
\sup_{t\in\RR}\ \sup_{x_1\ge 0}\bigl(h'_{\dagger,a,b}(\phi(x_1))\bigr)^2
=
\sup_{y_0>0,\ x_1\ge 0}\Theta_{a,b,\lambda}(y_0,x_1).
\]
\end{lemma}

\begin{proof}
Fix \(x_1\ge0\). For a given tilt \(t\), the intermediate ratio in the equal-input two-level recursion is
\[
y_0=\lambda_{\dagger^c}(1+x_1)^{-b}.
\]
As \(t\) ranges over \(\RR\), the activity \(\lambda_{\dagger^c}\) ranges over all positive real numbers, and therefore, for this fixed \(x_1\), the corresponding \(y_0\)'s range over all of \((0,\infty)\).

Conversely, given any \(y_0>0\) and \(x_1\ge0\), set
$
\lambda_{\dagger^c}:=y_0(1+x_1)^b$
and $
\lambda_{\dagger}:=\lambda^2/\lambda_{\dagger^c}
$
which is realized by a unique tilt \(t\). This produces exactly the prescribed pair \((y_0,x_1)\). Hence maximizing over \(t\) and \(x_1\) is equivalent to maximizing over \((y_0,x_1)\).
\end{proof}

We first characterize the interior critical points of $\Theta_{a,b,\lambda}$. A key simplification is that the value of \(\Theta_{a,b,\lambda}\) at a critical point can be written as a one-variable function.

\begin{lemma}
Suppose $(y_0,x_1)$ is an interior critical point of $\Theta_{a,b,\lambda}(y_0,x_1)$, with $y_0>0$ and $x_1>0$. Let \(\omega:=(a+1)y_0\). Then \(\omega>1\), and
\[
\Theta_{a,b,\lambda}(y_0,x_1)
=
\frac{a^2b^2\,\omega^2(\omega-1)}
{(a+1+\omega)^2(2b+\omega+1)}.
\]
\end{lemma}

\begin{proof}
At an interior critical point, we may set the logarithmic derivatives of \(\Theta_{a,b,\lambda}\) to zero. A direct differentiation gives the two equations
\begin{equation}\label{eq:SSM_Theta_ab_criticalEquationsForLambda}
\lambda^2=T(bx_1-1),
\qquad
2\lambda^2=T\bigl((a+1)y_0-1\bigr),
\end{equation}
where $T:=y_0(1+y_0)^a(1+x_1)^b$. Solving these equations yields \(y_0=\omega/(a+1)\) and \(x_1=(\omega+1)/(2b)\). Evaluating $\Theta_{a,b,\lambda}$ at these values gives the critical objective value.
Since $T>0$, the identity
$
2\lambda^2=T(\omega-1)
$
implies \(\omega>1\).
\end{proof}

We now use this reduction to one variable to show that every interior critical value stays below \(1\) throughout the subcritical regime.

\begin{lemma} \label{lem:theta_contraction_interior fixed_points}
Let \(d:=\Delta-1\), let \(1\le a,b\le d\), and let \(0<\lambda < \lambda_c(\Delta)\). If \((y_0,x_1)\) is any interior critical point of \(\Theta_{a,b,\lambda}\), then
\[
\Theta_{a,b,\lambda}(y_0,x_1)<1.
\]
\end{lemma}

\begin{proof}
For \(\omega>1\), define
\begin{align}
\Xi_{a,b}(\omega)
&:=
\frac{a^2b^2\,\omega^2(\omega-1)}{(a+1+\omega)^2(2b+\omega+1)} \\
\Lambda_{a,b}(\omega)
&:=
\frac{\omega(\omega-1)}{2(a+1)}
\left(\frac{a+1+\omega}{a+1}\right)^a
\left(\frac{2b+\omega+1}{2b}\right)^b.
\end{align}
At an interior critical point, the previous lemma gives
$
\Theta_{a,b,\lambda}(y_0,x_1)=\Xi_{a,b}(\omega)$
with $ \omega=(a+1)y_0
$, 
and \eqref{eq:SSM_Theta_ab_criticalEquationsForLambda} gives \(\lambda^2=\Lambda_{a,b}(\omega)\).
The function \(\Xi_{a,b}\) is strictly increasing on \((1,\infty)\), since
\[
\frac{d}{d\omega}\log \Xi_{a,b}(\omega)
=
\frac2\omega+\frac1{\omega-1}-\frac2{a+1+\omega}-\frac1{2b+\omega+1} > 0
\]
term-by-term. 
Likewise, \(\Lambda_{a,b}\) is strictly increasing on \((1,\infty)\), because
\[
\frac{d}{d\omega}\log \Lambda_{a,b}(\omega)
=
\frac1\omega+\frac1{\omega-1}+\frac{a}{a+1+\omega}+\frac{b}{2b+\omega+1}>0.
\]

Since $\Xi_{a,b}(1^+)=0$ and $\Xi_{a,b}(\omega)\to\infty$ as $\omega\to\infty$, there is a unique \(\omega_{a,b}>1\) with $\Xi_{a,b}(\omega_{a,b})=1.$ Define
\[
B(a,b):=\Lambda_{a,b}(\omega_{a,b}).
\]
We claim that \(B(a,b)\) is decreasing in both coordinates. Temporarily extend \(a,b\) to positive reals and set
\[
F(a,b,\omega):=\log \Lambda_{a,b}(\omega),
\qquad
G(a,b,\omega):=\log \Xi_{a,b}(\omega).
\]
Along \(G(a,b,\omega)=0\), write \(\omega=\omega(a,b)=\omega_{a,b}\). Then
$
\log B(a,b)=F(a,b,\omega(a,b)).
$
Since \(\partial_\omega G>0\), the implicit function theorem applies. Differentiating with respect to \(b\), we obtain
\[
\frac{\partial}{\partial b}\log B(a,b)
=
\partial_b F-\partial_\omega F\,\frac{\partial_b G}{\partial_\omega G}.
\]
A direct simplification gives
$
\frac{\partial}{\partial b}\log B(a,b)
=
\log\!\left(\frac{2b+\omega+1}{2b}\right)-\frac{\omega+1}{b}
<0,
$
where we used \(\log(1+z)<z<2z\). Similarly,
\[
\frac{\partial}{\partial a}\log B(a,b)
=
\partial_a F-\partial_\omega F\,\frac{\partial_a G}{\partial_\omega G}
=
\log\!\left(\frac{a+\omega+1}{a+1}\right)-\frac{\omega+1}{a} < 0.
\]
Thus \(B(a,b)\) is decreasing in both coordinates, and for $1 \leq a, b \leq d$,
\[
B(a,b)\ge B(d,d).
\]
On the other hand, for \(\omega^\star=(d+1)/(d-1)\), direct substitution gives
\(\Xi_{d,d}(\omega^\star)=1\), hence \(\omega_{d,d}=\omega^\star\) and
\[
B(d,d)=\Lambda_{d,d}(\omega^\star)
=\left(\frac{d^d}{(d-1)^{d+1}}\right)^2
=\lambda_c(\Delta)^2.
\]
Suppose, for contradiction, that \(\Xi_{a,b}(\omega)\ge1\). Then
\(\omega\ge\omega_{a,b}\), and monotonicity of \(\Lambda_{a,b}\) gives
$
\lambda^2=\Lambda_{a,b}(\omega)
\ge \Lambda_{a,b}(\omega_{a,b})
=B(a,b)
\ge B(d,d)=\lambda_c(\Delta)^2,
$
contradicting \(\lambda<\lambda_c(\Delta)\). Thus
\(\Xi_{a,b}(\omega)<1\) which completes the proof.
\end{proof}

\begin{proof}[Proof of Proposition~\ref{prop:phi_contraction}]
Write \(d:=\Delta-1\). The cases $a=0$ or $b=0$ are trivial, since then \(h'_{\dagger,a,b}\equiv 0\).

So fix \(1\le a,b\le d\). By Lemma \ref{lem:h_dagger_theta_reduction}, it is enough to show that
$
\sup_{y_0>0,\ x_1\ge0}\Theta_{a,b,\lambda}(y_0,x_1)
$
is strictly upper bounded by 1. The boundary cases cause no difficulty since from the formula for $\Theta_{a,b,\lambda}$, we have
$
\Theta_{a,b,\lambda}(y_0,x_1)\to 0
$
as $y_0\to 0$, as $x_1\to 0$, as $y_0\to\infty$, or as $x_1\to\infty$. Therefore the global supremum is attained at an interior critical point. By Lemma \ref{lem:theta_contraction_interior fixed_points} every such critical value is strictly less than \(1\). Thus $\sup_{t\in\RR}\sup_{x_1\ge 0}\bigl(h'_{\dagger,a,b}(\phi(x_1))\bigr)^2<1$. The proof finishes by taking a maximum over the finitely many choices of \(\dagger,a,b\).
\end{proof}

\subsection{Tilted sampling via the SAW tree}

Algorithm~\ref{alg:saw-trunc} estimates a single conditional marginal. We use the standard self-reduction for weighted independent sets~\cite[Section~5]{weitz2006counting}.

\begin{algorithm}[H]
\caption{Approximate sampler from the tilted hard-core measure}\label{alg:tilted-sampler}
\begin{algorithmic}[1]
\Require A bipartite graph $G=(L\sqcup R,E)$, an activity $\lambda<\lambda_c(\Delta)$, a fixed tilt $t\in\mathbb R$, and an accuracy parameter $\epsilon>0$
\Ensure A sample from $\mu_{\lambda,t}$ up to total variation error $\epsilon$

\State Fix an ordering $v_1,\dots,v_n$ of $V(G)$; set $G_1\gets G$ and $I\gets\varnothing$
\For{$i=1,\dots,n$}
    \State If $v_i\notin V(G_i)$, set $G_{i+1}\gets G_i$; otherwise use Algorithm~\ref{alg:saw-trunc} on $G_i$ with its inherited activities to obtain $\widetilde p_i$ with $\abs{\widetilde p_i-\mu_{G_i,\lambda,t}(\sigma_{v_i}=1)}\le\epsilon/n$, and draw $X_i\sim\operatorname{Bernoulli}(\widetilde p_i)$.
    \State In the latter case, if $X_i=1$, set $I\gets I\cup\{v_i\}$ and $G_{i+1}\gets G_i-N_{G_i}[v_i]$; if $X_i=0$, set $G_{i+1}\gets G_i-v_i$.
\EndFor
\State \Return $I$
\end{algorithmic}
\end{algorithm}

We claim that Algorithm~\ref{alg:tilted-sampler} is an
\(\epsilon\)-approximate sampler for \(\mu_{\lambda,t}\) running in
polynomial time. To see this, note that if the exact conditional marginals
\(p_i:=\mu_{G_i,\lambda,t}(\sigma_{v_i}=1)\) were used instead of
\(\widetilde p_i\), the procedure would output an exact sample
from \(\mu_{\lambda,t}\). Run this exact procedure and
Algorithm~\ref{alg:tilted-sampler} in parallel, using the same vertex ordering,
and denote their outputs by \(I^{\mathrm{ex}}\) and \(I^{\mathrm{ap}}\),
respectively.
Let \(X_i\) and \(\widetilde X_i\) be their decisions at step \(i\), and define
\[
A_{i-1}:=\{X_j=\widetilde X_j\text{ for every }j<i\}.
\]
Conditional on \(A_{i-1}\) and a common occupation/rejection history \(H_{i-1}=h\) for the vertices $v_1,\dots,v_{i-1}$, the residual
graphs agree. Let \(p_i(h)\) and \(\widetilde p_i(h)\) be the corresponding
exact and approximate marginals. Let
\(U_i\sim\operatorname{Unif}[0,1]\) and set
\(X_i=\mathbf 1\{U_i\le p_i(h)\}\) and
\(\widetilde X_i=\mathbf 1\{U_i\le\widetilde p_i(h)\}\). By
Theorem~\ref{thm:ssm} and
Remark~\ref{remark:SSM_to_marginalOracle}, $\Pr(X_i\ne\widetilde X_i\mid A_{i-1},H_{i-1}=h)
=|p_i(h)-\widetilde p_i(h)|\le\epsilon/n$. Since this holds uniformly over $h$, we obtain by law of total probability
\begin{equation}\label{eq:tilted-sampler-coupling}
\Pr(X_i\ne\widetilde X_i\mid A_{i-1})\le\epsilon/n.
\end{equation}
Let \(F_i:=A_{i-1}\cap\{X_i\ne\widetilde X_i\}\), the event that the first
disagreement occurs at step \(i\). Using the fact that the \(F_i\) are disjoint, and \eqref{eq:tilted-sampler-coupling}, we have
\begin{align*}
\left\|\text{Law}(I^{\mathrm{ap}})-\mu_{\lambda,t}\right\|_{\mathrm{TV}}
&\le \Pr(I^{\mathrm{ap}}\ne I^{\mathrm{ex}})
\le \Pr\bigl((X_1,\dots,X_n)\ne(\widetilde X_1,\dots,\widetilde X_n)\bigr) \\
&=\sum_{i=1}^n\Pr(F_i)
=\sum_{i=1}^n\Pr(A_{i-1})
  \Pr(X_i\ne\widetilde X_i\mid A_{i-1})
\le\sum_{i=1}^n\frac{\epsilon}{n}\le\epsilon.
\end{align*}
Finally, Theorem~\ref{thm:ssm} and
Remark~\ref{remark:SSM_to_marginalOracle} compute every
\(\widetilde p_i\) in polynomial time.

\subsection{Zero-freeness of the tilted partition function}
\label{sec:zeroFreenessTiltedPartitionFunction}

This section establishes a zero-freeness result for a complex version of $Z_G(\lambda;t)$. This will be used to establish a local CLT for the balance $B(I)$ under the tilted measure.

For \(\zeta\in\mathbb C\), define
\begin{equation}\label{eq:tiltedComplexHardcorePartitionFunction}
Z_G(\lambda;t,\zeta):=\sum_{I\in\cI_G}\lambda^{|I|}e^{tB(I)}\zeta^{B(I)},
\end{equation}
where \(B(I)\) is defined in \eqref{eq:balance_definition}. Thus
\(Z_G(\lambda;t)=Z_G(\lambda;t,1)\).
Equivalently, we consider left and right complex fugacities
$
\lambda_L(\zeta)=\lambda e^t\zeta$ and
$
\lambda_R(\zeta)=\lambda e^{-t}\zeta^{-1}
$ in \eqref{eq:tiltedHardcorePartition_intro}. The main result is as follows.

\begin{proposition}
\label{prop:zero_free_tilted_partition}
Fix \(\Delta\ge3\), \(\delta>0\), and \(t_0>0\). Let \(G\) be a bipartite graph with maximum degree \(\Delta\), let \(\lambda \in (0,\lambda_c(\Delta)  - \delta) \), and let $t \in [-t_0,t_0]$. Then there exists $\rho = \rho(\Delta,\delta,t_0) \in (0,1)$ so that
\[
Z_G(\lambda;t,\zeta) \neq 0 \quad \text{for all } \abs{\zeta-1} < \rho.
\]
\end{proposition}

We introduce some notation. Fix $\epsilon>0$ and $t_0>0$, and write $d:=\Delta-1$. Set
\[
\lambda_{\max}:=(1-\epsilon)\lambda_c(\Delta)e^{t_0},
\qquad\text{and}\qquad
W:=\phi(\lambda_{\max}).
\]
For $\alpha>0$, write
\[
D_\alpha:=\inbraces{q\in\C:\operatorname{dist}(q,[0,W])<\alpha}.
\]

Consider a rooted bipartite tree fragment whose root is on side \(\dagger\in\{L,R\}\), has \(a\) children, and whose \(i\)-th child has \(b_i\) children, where $0 \leq a , b_i \leq d$.
Let $\mathbf b=(b_1,\dots,b_a)$. Suppose the grandchild ratios in $\phi$-coordinates are
$
\mathbf q_1=(q_{1,i,j})_{1\le i\le a,\ 1\le j\le b_i}
$.
In original coordinates, the corresponding child ratios would be
\[
y_i:=
\lambda_{\dagger^c}(\zeta)
\prod_{j=1}^{b_i}\operatorname{sech}^2 q_{1,i,j}.
\]
Thus the root ratio $q_0$ in $\phi$-coordinates is $\mathcal H_{\dagger,a,\mathbf b}^{\lambda,t,\zeta}(\mathbf q_1)$ given by
\begin{equation}\label{eq:zeroFreeness_twoLevel_update}
\mathcal H_{\dagger,a,\mathbf b}^{\lambda,t,\zeta}(\mathbf q_1)
:=
\operatorname{arcsinh}\left(
\sqrt{\lambda_\dagger(\zeta)}
\prod_{i=1}^a (1+y_i)^{-1/2}
\right),
\end{equation}
where we take the square-root branches which agree with the positive real branches at $\zeta=1$ and real nonnegative inputs.

A key step in the proof of Proposition \ref{prop:zero_free_tilted_partition} is extending the real-variable results from Section \ref{sec:sampling_from_tilted} into the complex plane. In particular, at the real \(\zeta=1\), Propositions \ref{prop:two-level-equal-input} and \ref{prop:phi_contraction} showed that the two-level hard-core recursion is uniformly contracting in $\phi$-coordinates whenever \(\lambda\) is bounded away from \(\lambda_c(\Delta)\). We will use continuity to show that, for complex $\zeta$ around $1$, the same two-level recursion maps a small complex neighborhood of the real interval \([0,W]\) back into itself. This complex neighborhood is chosen so that, after converting back by \(\phi^{-1}\), the ratios in original coordinates stay away from \(-1\).

\begin{lemma}\label{lem:zeroFreeness_twoLevel_complex_thickening}
There exist $\alpha,\rho>0$ such that
\begin{equation}\label{eq:zeroFreeness_phi_inverse_halfplane}
\phi^{-1}(\overline{D_\alpha})\subseteq \{x\in\C:\operatorname{Re}x>-1/2\},
\end{equation}
and the following holds. Whenever \(0\le \lambda\le (1-\epsilon)\lambda_c(\Delta)\), \(t\in[-t_0,t_0]\), \(|\zeta-1|<\rho\), \(\dagger\in\{L,R\}\), \(0\le a\le d\), and \(\mathbf b=(b_1,\dots,b_a)\) with \(0\le b_i\le d\), the map $\mathcal H_{\dagger,a,\mathbf b}^{\lambda,t,\zeta}$ is well defined and
\[
\mathcal H_{\dagger,a,\mathbf b}^{\lambda,t,\zeta}
\bigl(D_\alpha^{\,b_1+\cdots+b_a}\bigr)
\subseteq
D_\alpha.
\]
\end{lemma}

\begin{proof}
We first state the real contraction result. At \(\zeta=1\), every real
occupation ratio is bounded by the activity at its vertex:
\[
R_u=\lambda_{\dagger(u)}\prod_{v\in N(u)}(1+R_v)^{-1}
\le \lambda_{\dagger(u)}
\le \lambda e^{t_0}
\le \lambda_{\max}.
\]
Thus the ratios in \(\phi\)-coordinate lie in
\([0,W]\). On this real domain, Propositions
\ref{prop:two-level-equal-input} and \ref{prop:phi_contraction}, together
with the monotonicity of $\Theta_{a,b,\lambda}$ in $\lambda$, imply that
there is a number $\eta_0=\eta_0(\Delta,\epsilon)>0$ such that, for every
real choice of parameters
\[
0\le \lambda\le (1-\epsilon)\lambda_c(\Delta),
\qquad
t\in\RR,
\qquad
0\le a\le d,
\qquad
0\le b_i\le d,
\]
and every real input vector $\mathbf q_1^0\in[0,W]^{b_1+\cdots+b_a}$,
\begin{equation}\label{eq:zeroFreeness_real_twoLevel_jacobian_gap}
\sum_{i=1}^a\sum_{j=1}^{b_i}
\left|
\frac{\partial}{\partial q_{1,i,j}}
\mathcal H_{\dagger,a,\mathbf b}^{\lambda,t,1}(\mathbf q_1^0)
\right|
\le 1-\eta_0.
\end{equation}
Indeed, if $a=0$ or $b_i=0$ for every $i$, then the left-hand side is zero. Otherwise, Proposition \ref{prop:two-level-equal-input} bounds the left-hand side by \(h'_{\dagger,a,b}(\phi(x_1))\) with $b=\max_i b_i$ for some $x_1\ge0$, and Proposition \ref{prop:phi_contraction} gives a uniform gap when $\lambda\le(1-\epsilon)\lambda_c(\Delta)$.

Choose $\alpha>0$ small enough so that all the branch choices in
\eqref{eq:zeroFreeness_twoLevel_update} are valid on \(D_\alpha\) for
\(\zeta\) near \(1\), and so that
\eqref{eq:zeroFreeness_phi_inverse_halfplane} holds. This is possible
because $\phi^{-1}$ maps $[0,W]$ onto $[0,\lambda_{\max}]$. We also take
\(\alpha\) small enough that the real derivative gap
\eqref{eq:zeroFreeness_real_twoLevel_jacobian_gap} persists on the
\(\alpha\)-thickening at \(\zeta=1\). Indeed, for each arity pattern
\[
(a,\mathbf b)=(a,b_1,\ldots,b_a),
\qquad 0\le a\le d,\qquad 0\le b_i\le d,
\]
there are only finitely many choices, and the corresponding parameter set
\[
0\le \lambda\le (1-\epsilon)\lambda_c(\Delta),
\qquad
t\in[-t_0,t_0],
\qquad
\mathbf q_1^0\in[0,W]^{b_1+\cdots+b_a}
\]
is compact. Every point of \(D_\alpha^{\,b_1+\cdots+b_a}\) is within
\(\alpha\), coordinatewise, of this real cube. The formula
\eqref{eq:zeroFreeness_twoLevel_update} shows that, for each \(\lambda>0\),
the maps are analytic in the variables \(q_{1,i,j}\) and in \(\zeta\) in a
fixed neighborhood of this compact set; their coordinate derivatives extend
continuously to \(\lambda=0\), where the map is constant. Hence uniform
continuity, together with the strict gap in
\eqref{eq:zeroFreeness_real_twoLevel_jacobian_gap}, lets us choose this
\(\alpha\) so that the derivative sum is still at most \(1-3\eta_0/4\) for
\(\zeta=1\) and all \(q_{1,i,j}\in D_\alpha\).

Now fix this \(\alpha\). Applying the same compactness and uniform
continuity argument in the \(\zeta\)-variable, and using again that there are
only finitely many arity patterns, we may choose \(\rho>0\) small enough so
that, whenever \(|\zeta-1|<\rho\) and each \(q_{1,i,j}\in D_\alpha\),
\begin{equation}\label{eq:zeroFreeness_complex_twoLevel_jacobian_gap}
\sum_{i=1}^a\sum_{j=1}^{b_i}
\left|
\frac{\partial}{\partial q_{1,i,j}}
\mathcal H_{\dagger,a,\mathbf b}^{\lambda,t,\zeta}(\mathbf q_1)
\right|
\le 1-\frac{\eta_0}{2},
\end{equation}
and, for every real $\mathbf q_1^0\in[0,W]^{b_1+\cdots+b_a}$,
\begin{equation}\label{eq:zeroFreeness_parameter_error}
\left|
\mathcal H_{\dagger,a,\mathbf b}^{\lambda,t,\zeta}(\mathbf q_1^0)
-
\mathcal H_{\dagger,a,\mathbf b}^{\lambda,t,1}(\mathbf q_1^0)
\right|
\le
\frac{\eta_0\alpha}{2}.
\end{equation}

Now fix $\mathbf q_1\in D_\alpha^{b_1+\cdots+b_a}$, and choose $\mathbf q_1^0\in[0,W]^{b_1+\cdots+b_a}$ such that
\[
\max_{i,j}|q_{1,i,j}-q^0_{1,i,j}|<\alpha.
\]
The set $D_\alpha$ is convex, so the line segment from $\mathbf q_1^0$ to $\mathbf q_1$ stays inside the domain on which \eqref{eq:zeroFreeness_complex_twoLevel_jacobian_gap} holds. Hence, by the fundamental theorem of calculus,
\[
\left|
\mathcal H_{\dagger,a,\mathbf b}^{\lambda,t,\zeta}(\mathbf q_1)
-
\mathcal H_{\dagger,a,\mathbf b}^{\lambda,t,\zeta}(\mathbf q_1^0)
\right|
\le
\left(1-\frac{\eta_0}{2}\right)\alpha.
\]
For real parameters and real inputs, the output ratio lies in
$[0,\lambda_{\max}]$, so
$\mathcal H_{\dagger,a,\mathbf b}^{\lambda,t,1}(\mathbf q_1^0)\in[0,W]$. The triangle inequality and
\eqref{eq:zeroFreeness_parameter_error} give
\begin{align*}
\operatorname{dist}\left(
\mathcal H_{\dagger,a,\mathbf b}^{\lambda,t,\zeta}(\mathbf q_1),
[0,W]
\right)
&\le
\left|
\mathcal H_{\dagger,a,\mathbf b}^{\lambda,t,\zeta}(\mathbf q_1)
-
\mathcal H_{\dagger,a,\mathbf b}^{\lambda,t,1}(\mathbf q_1^0)
\right| \\
&\le
\left|
\mathcal H_{\dagger,a,\mathbf b}^{\lambda,t,\zeta}(\mathbf q_1)
-
\mathcal H_{\dagger,a,\mathbf b}^{\lambda,t,\zeta}(\mathbf q_1^0)
\right|
+
\left|
\mathcal H_{\dagger,a,\mathbf b}^{\lambda,t,\zeta}(\mathbf q_1^0)
-
\mathcal H_{\dagger,a,\mathbf b}^{\lambda,t,1}(\mathbf q_1^0)
\right| \\
&<
\left(1-\frac{\eta_0}{2}\right)\alpha+\frac{\eta_0\alpha}{2}
=\alpha.
\end{align*}
This finishes the proof.
\end{proof}

\begin{proof}[Proof of Proposition \ref{prop:zero_free_tilted_partition}]
Write
\[
\epsilon:=\frac{\delta}{\lambda_c(\Delta)}.
\]
If \(\delta\ge \lambda_c(\Delta)\), then the interval for \(\lambda\) is empty, so we may assume \(0<\epsilon<1\). Applying Lemma \ref{lem:zeroFreeness_twoLevel_complex_thickening} with this \(\epsilon\) and \(t_0\), we obtain numbers \(\alpha,\rho>0\). Shrink \(\rho\) if necessary so that \(\rho<1\), and set
\[
\Omega:=\phi^{-1}(D_\alpha):=\{\phi^{-1}(w):w\in D_\alpha\}.
\]
By \eqref{eq:zeroFreeness_phi_inverse_halfplane}, we have
\begin{equation}\label{eq:zeroFreeness_ratio_domain_right_halfplane}
\Omega\subseteq\{x\in\C:\operatorname{Re}x>-1/2\}.
\end{equation}

Fix a complex number \(\zeta\) with \(|\zeta-1|<\rho\). Since the partition function is multiplicative over connected components of \(G\), it suffices to prove the lemma when \(G\) is connected.

Let \(H\) be an induced subgraph of \(G\), and let \(u\in V(H)\) have at most \(d=\Delta-1\) neighbors in \(H\). Whenever \(Z_{H-u}(\lambda;t,\zeta)\neq0\), write
\[
R_{H,u}
:=
\frac{\lambda_{\dagger(u)}(\zeta)\,
Z_{H\setminus N_H[u]}(\lambda;t,\zeta)}
{Z_{H-u}(\lambda;t,\zeta)},
\]
where \(\dagger(u)\in\{L,R\}\) denotes the side containing \(u\).
List the neighbors of \(u\) in \(H\) as \(v_1,\dots,v_a\), define \(H_0:=H-u\) and \(H_i:=H_{i-1}-v_i\), and then list the neighbors of \(v_i\) in \(H_{i-1}\) as \(w_{i,1},\dots,w_{i,b_i}\). Since \(u\) has been deleted from \(H_{i-1}\), we have \(b_i\le d\). Let
\[
K_{i,0}:=H_{i-1}-v_i,
\qquad\text{and}\qquad
K_{i,j}:=K_{i,j-1}-w_{i,j}\quad \text{for } 1\le j\le b_i.
\]
Suppose that the ratios \(R_{K_{i,j-1},w_{i,j}}\) are well defined and lie in \(\Omega\). Choose \(\omega_{i,j}\in D_\alpha\) with
\[
R_{K_{i,j-1},w_{i,j}}=\phi^{-1}(\omega_{i,j}).
\]
Substituting the one-level ratio recurrence twice gives
\begin{equation}\label{eq:zeroFreeness_boundary_ratio_update}
R_{H,u}
=
\phi^{-1}\left(
\mathcal H_{\dagger(u),a,\mathbf b}^{\lambda,t,\zeta}
\bigl((\omega_{i,j})_{1\le i\le a,\ 1\le j\le b_i}\bigr)
\right),
\qquad \mathbf b=(b_1,\dots,b_a).
\end{equation}
Lemma \ref{lem:zeroFreeness_twoLevel_complex_thickening} puts the point inside \(\phi^{-1}(\cdot)\) in \(D_\alpha\), so \eqref{eq:zeroFreeness_boundary_ratio_update} gives \(R_{H,u}\in\Omega\). In particular \(R_{H,u}\neq-1\) by \eqref{eq:zeroFreeness_ratio_domain_right_halfplane}.

Fix \(v_0\in V(G)\). We claim that, for every \(U\subseteq V(G)\setminus\{v_0\}\), with \(H=G[U]\), we have \(Z_H(\lambda;t,\zeta)\neq0\). Moreover, if \(u\in U\) has a neighbor in \(V(G)\setminus U\), then \(R_{H,u}\) is well defined and lies in \(\Omega\).

We prove the claim by induction on \(|U|\). The case \(U=\emptyset\) is immediate. Assume \(U\neq\emptyset\), and that the claim holds for all proper subsets of \(U\). Since \(G\) is connected and \(v_0\notin U\), there is a vertex \(u_0\in U\) with a neighbor in \(V(G)\setminus U\). Then \(u_0\) has at most \(d\) neighbors in \(H\), and \(Z_{H-u_0}(\lambda;t,\zeta)\neq0\) by induction. In the setup of \eqref{eq:zeroFreeness_boundary_ratio_update} for \(H,u_0\), every graph \(K_{i,j-1}\) is a proper induced subgraph of \(H\), and \(w_{i,j}\) has the neighbor \(v_i\) outside \(K_{i,j-1}\). Hence the induction hypothesis puts each grandchild ratio \(R_{K_{i,j-1},w_{i,j}}\) in \(\Omega\), and \eqref{eq:zeroFreeness_boundary_ratio_update} gives \(R_{H,u_0}\in\Omega\). Therefore
\[
Z_H(\lambda;t,\zeta)
=
Z_{H-u_0}(\lambda;t,\zeta)\,(1+R_{H,u_0})
\neq0.
\]
Now let \(u\in U\) be any vertex with a neighbor outside \(U\). Using \eqref{eq:zeroFreeness_boundary_ratio_update} for \(H,u\), the needed grandchild ratios again live in proper induced subgraphs and are covered by induction. Thus \(R_{H,u}\in\Omega\). This completes the induction.

We now prove nonvanishing for \(G\). If \(G\) has no vertices there is nothing to show. Otherwise fix \(v_0\in V(G)\), and list its neighbors as \(v_1,\dots,v_m\), where \(m\le\Delta\). Define
\[
G_0:=G-v_0,
\qquad
G_i:=G_{i-1}-v_i\quad (1\le i\le m).
\]
Applying the claim to \(U=V(G)\setminus\{v_0\}\) gives \(Z_{G-v_0}(\lambda;t,\zeta)\neq0\). For each \(i\), applying the claim to
\[
U_i:=V(G)\setminus\{v_0,v_1,\dots,v_{i-1}\}
\]
gives \(R_{G_{i-1},v_i}\in\Omega\), since \(v_i\) has the outside neighbor \(v_0\).

\begin{itemize}
\item If \(m\le d\), use \eqref{eq:zeroFreeness_boundary_ratio_update} with \(H=G\) and \(u=v_0\). The required grandchild ratios lie in proper induced subgraphs of \(G-v_0\), with the deleted parent outside, so the claim puts them in \(\Omega\). Hence \(R_{G,v_0}\in\Omega\), and therefore \(R_{G,v_0}\neq-1\). It follows that
\[
Z_G(\lambda;t,\zeta)=Z_{G-v_0}(\lambda;t,\zeta)(1+R_{G,v_0})\neq0.
\]

\item If \(m=\Delta\), set
\[
Q:=
\frac{\lambda_{\dagger(v_0)}(\zeta)}
{\prod_{i=1}^{m-1}(1+R_{G_{i-1},v_i})}
=R_{G,v_0}(1+R_{G_{m-1},v_m}).
\]
Expanding only the first \(m-1=d\) neighbor ratios one further level, the claim supplies all grandchild inputs in \(\Omega\). Lemma \ref{lem:zeroFreeness_twoLevel_complex_thickening} therefore gives \(Q\in\Omega\). We also have \(R_{G_{m-1},v_m}\in\Omega\). If \(R_{G,v_0}=-1\), then
\[
-1=Q+R_{G_{m-1},v_m},
\]
which is impossible by \eqref{eq:zeroFreeness_ratio_domain_right_halfplane}, since both summands have real part greater than \(-1/2\). Thus \(R_{G,v_0}\neq-1\), and again \(Z_G(\lambda;t,\zeta)\neq0\).
\end{itemize}

Since \(\zeta\) was arbitrary with \(|\zeta-1|<\rho\), this proves the claimed nonvanishing of \(Z_G(\lambda;t,\zeta)\).
\end{proof}

\subsection{A local central limit theorem for the imbalance} \label{subsec:lclt}

Throughout this section, let $G=(L\sqcup R,E)$ be a bipartite graph of maximum degree at most $\Delta$, and write $n:=|V(G)|=|L|+|R|$.

We now prove a local central limit theorem for the imbalance \(B(I)\) (defined in \eqref{eq:balance_definition}) under the
tilted hard-core measure. This will be crucial for the rejection sampling
step in Algorithm \ref{alg:balanced_hardcore_sampler_template}. Throughout, let
\[
N(x):=\frac{e^{-x^2/2}}{\sqrt{2\pi}}
\]
denote the density of the standard normal distribution. The main result for this section is as follows.

\begin{theorem} \label{thm:main_local_clt}
Fix $\Delta \ge 3$, $\delta>0$, $t_0>0$, and $\lambda\in(0,\lambda_c(\Delta)-\delta)$.
Let $I\sim \mu_{\lambda,t}$ for some $t\in[-t_0,t_0]$. Let
\[
\mu:=\mathbb E_{\mu_{\lambda,t}}[B(I)],
\qquad
\sigma^2:=\Var_{\mu_{\lambda,t}}[B(I)].
\]
Then
\[
\sup_{b\in\mathbb Z}
\left|
\sigma^{-1}N\!\left(\frac{b-\mu}{\sigma}\right)
-
\mathbb P_{\mu_{\lambda,t}}[B(I)=b]
\right|
=
O_{\Delta,\delta,t_0,\lambda}\!\left(
\frac{(\log n)^{5/2}}{n}
\right).
\]
\end{theorem}

The proof proceeds by the following standard Fourier inversion estimate.

\begin{lemma}(see Lemma 2.5 of \cite{jain2021approximatecountingsamplinglocal} or Lemma 3 of \cite{berkowitz2017quantitativelocallimittheorem}) \label{lem:fourier-inversion}
Let $X$ be a random variable supported on the lattice
$L=\alpha+\sigma^{-1}\mathbb{Z}$ for some $\alpha\in\mathbb R$ and
$\sigma>0$, and let $Z\sim \mathcal{N}(0,1)$. Then
\[
\sup_{x\in L}\left|\sigma^{-1}N(x)-\mathbb{P}[X=x]\right|
\le
\frac{1}{\sigma} \int_{-\pi\sigma}^{\pi\sigma}
\left|
\mathbb{E}\!\left[e^{i\theta X}\right]-\mathbb{E}\!\left[e^{i\theta Z}\right]
\right|\,d\theta
+
e^{-\pi^{2}\sigma^{2}/2}.
\]
\end{lemma}

The proof follows the strategy of Section~3 of \cite{jain2021approximatecountingsamplinglocal}. Writing
\(X=(B(I)-\mu)/\sigma\), \(\phi_X(\theta):=\mathbb E[e^{i\theta X}]\), and
\(\phi_Z(\theta):=e^{-\theta^2/2}\), the Fourier integral in
Lemma~\ref{lem:fourier-inversion} is split at a cutoff
\(\theta_0\asymp\sqrt{\log n}\) as
\[
\int_{-\pi\sigma}^{\pi\sigma}
\bigl|\phi_X(\theta)-\phi_Z(\theta)\bigr|\,d\theta
=
\underbrace{\int_{|\theta|\le \theta_0}
\bigl|\phi_X(\theta)-\phi_Z(\theta)\bigr|\,d\theta}_{\text{low Fourier phases}}
+
\underbrace{\int_{\theta_0<|\theta|\le \pi\sigma}
\bigl|\phi_X(\theta)-\phi_Z(\theta)\bigr|\,d\theta}_{\text{high Fourier phases}}.
\]

The low frequency part is controlled by the zero-freeness result from the previous subsection, similar to the proof of Lemma~3.3 in \cite{jain2021approximatecountingsamplinglocal}. The high frequency part is bounded by conditioning on a collection of
well-separated vertices, adapting the arguments in \cite{jain2021approximatecountingsamplinglocal} Lemma~3.5. Note that \cite{jain2021approximatecountingsamplinglocal} gives a local CLT for $\abs{I}$ in a univariate hard-core model, which differs from  our variable of interest $B(I)$ and our tilted model with different fugacities in $L$ and $R$.

Following \cite{jain2021approximatecountingsamplinglocal}, we begin with a
variance bound.

\begin{lemma}\label{lem:balance_variance_bounds}
Let \(\lambda\in(0,\lambda_c(\Delta)-\delta)\) and \(t\in[-t_0,t_0]\).
Then there exist constants \(c_{\Delta,\delta,t_0},C_{\Delta,\delta,t_0}>0\) such that
\[
    c_{\Delta,\delta,t_0}\,\lambda n
    \;\le\;
    \Var_{\mu_{\lambda,t}}[B(I)]
    \;\le\;
    C_{\Delta,\delta,t_0}\,n.
\]
\end{lemma}

\begin{proof}
We first prove the upper bound. Since \(B(I)\) may be negative, we pass from
\(Z_G(\lambda;t,\zeta)\) to the polynomial
\[
    F_t(\zeta):=\zeta^{|R|}Z_G(\lambda;t,\zeta)
    =
    \sum_{I\in\cI_G}
    \lambda^{|I|}e^{\,tB(I)}\zeta^{B(I)+|R|}.
\]
Because $-|R|\le B(I)\le |L|,$
the exponent \(B(I)+|R|\) lies in \(\{0,1,\dots,n\}\), so \(F_t\) is a polynomial of degree at most \(n\).
By Proposition~\ref{prop:zero_free_tilted_partition}, there exists \(\rho=\rho(\Delta,\delta,t_0)\in(0,1)\) such that \(Z_G(\lambda;t,\zeta)\neq 0\) for all \(|\zeta-1|<\rho\).
Thus, \(F_t\) also has no zeros in \(\{|\zeta-1|<\rho\}\). Moreover, \(F_t(0) = \lambda^{|R|}e^{-t|R|} \neq 0\). Thus, we write
\[
    F_t(\zeta)=a_t\prod_{j=1}^N\left(1-\frac{\zeta}{\zeta_j}\right),
\]
where \(N\le n\) and \(\zeta_1,\dots,\zeta_N\) are the nonzero roots of \(F_t\), counted with multiplicity.
Then every \(\zeta_j\) satisfies $|\zeta_j-1|\ge \rho.$
Now, observe that (with some overloaded notation)
\[
    Z_G(\lambda;t+s)
    =
    \sum_{I\in\cI_G}
    \lambda^{|I|}e^{(t+s)B(I)}
    =
    Z_G(\lambda;t,e^s),
\]
and since \(F_t(e^s)=e^{s|R|}Z_G(\lambda;t,e^s)\),
we have \(\log F_t(e^s)=s|R|+\log Z_G(\lambda;t,e^s)\).
Therefore
\[
    \Var_{\mu_{\lambda,t}}[B(I)]
    =
    \left.\frac{d^2}{ds^2}\log Z_G(\lambda;t+s)\right|_{s=0}
    =
    \left.\frac{d^2}{ds^2}\log F_t(e^s)\right|_{s=0}.
\]
Using the factorization of \(F_t\), we have that $
\log F_t(e^s)=\log a_t+\sum_{j=1}^N \log\!\left(1-\frac{e^s}{\zeta_j}\right),$ so
\begin{align}
\frac{d}{ds}\log F_t(e^s)
=
-\sum_{j=1}^N \frac{e^s}{\zeta_j-e^s}
\qquad\text{and}\qquad
\frac{d^2}{ds^2}\log F_t(e^s)
=
-\sum_{j=1}^N \frac{e^s\zeta_j}{(\zeta_j-e^s)^2}.
\end{align}
Evaluating at \(s=0\), we obtain
\[
\Var_{\mu_{\lambda,t}}[B(I)]
=
-\sum_{j=1}^N \frac{\zeta_j}{(\zeta_j-1)^2}
\le
\sum_{j=1}^N \left|\frac{\zeta_j}{(\zeta_j-1)^2}\right|.
\]
Since \(|\zeta_j|\le|\zeta_j-1|+1\) and \(|\zeta_j-1|\ge\rho\), each summand is at most \(\rho^{-1}+\rho^{-2}=O_\rho(1)\). As \(N\le n\), this gives \(\Var_{\mu_{\lambda,t}}[B(I)]\le C_\rho n\), proving the upper bound.

We now prove the lower bound. Choose \(J\in\{L,R\}\) so that \(|J|\ge n/2\),
write \(\bar J:=V(G)\setminus J\), and set \(K:=I\cap\bar J\). Let
\(\lambda_J\) be the activity on \(J\), and set
\(p_J:=\lambda_J/(1+\lambda_J)\). Since \(G\) is bipartite, conditional on
\(K\), the vertices of \(J\setminus N(K)\) are mutually nonadjacent and
unblocked, while the vertices of \(N(K)\cap J\) are forced to be absent. Thus,
with \(U:=|J\setminus N(K)|\), we have
\(|I\cap J|\mid K\sim\operatorname{Bin}(U,p_J)\).  Moreover \(B(I)\) differs
from \(|I\cap J|\) only by a sign and an additive function of \(K\). Hence the
law of total variance gives
\[
    \Var[B(I)]
    \ge
    \mathbb E[\Var(B(I)\mid K)]
    =
    p_J(1-p_J)\,\mathbb E U.
\]

It remains to lower bound \(\mathbb E U\). Fix \(u\in J\), and let
\(v_1,\dots,v_d\) be its neighbors in \(\bar J\), where \(d\le\Delta\). Then
\[
\mathbb P(u\in J\setminus N(K))
=
\prod_{i=1}^d
\mathbb P\!\left(v_i\notin I\,\middle|\,v_1,\dots,v_{i-1}\notin I\right).
\]
Let \(\lambda_{\bar J}\) be the activity on \(\bar J\), and let \(A_i\) be the
event that \(v_1,\dots,v_{i-1}\notin I\). Deleting \(v_i\) from any independent
set in \(A_i\) that contains \(v_i\) gives an independent set in \(A_i\) that
does not contain \(v_i\), with the original weight larger by a factor
\(\lambda_{\bar J}\). Therefore
\[
\mathbb P(v_i\in I\mid A_i)
\le
\lambda_{\bar J}\,\mathbb P(v_i\notin I\mid A_i),
\qquad\text{so}\qquad
\mathbb P(v_i\notin I\mid A_i)\ge\frac{1}{1+\lambda_{\bar J}}.
\]
Since \(|t|\le t_0\), both \(\lambda_J\) and \(\lambda_{\bar J}\) lie in
\([\lambda e^{-t_0},\lambda e^{t_0}]\). Hence
\[
\mathbb P(u\in J\setminus N(K))
\ge
(1+\lambda_{\bar J})^{-d}
\ge
(1+\lambda e^{t_0})^{-\Delta},
\qquad
\mathbb E U
\ge
\frac n2(1+\lambda e^{t_0})^{-\Delta}.
\]
Also
\[
p_J(1-p_J)
=
\frac{\lambda_J}{(1+\lambda_J)^2}
\ge
\frac{\lambda e^{-t_0}}{(1+\lambda e^{t_0})^2}.
\]
Combining these estimates and using
\(\lambda\le\lambda_c(\Delta)-\delta\), we finally obtain
\[
    \Var_{\mu_{\lambda,t}}[B(I)]
    \ge
    \frac{e^{-t_0}}{2(1+(\lambda_c(\Delta)-\delta)e^{t_0})^{2+\Delta}}
    \,\lambda n.\qedhere
\]
\end{proof}

\paragraph{Low Fourier phases.}

Here we use the zero-freeness of the tilted hard-core
model from Proposition \ref{prop:zero_free_tilted_partition} as well as the following general result of
\cite{michelen2019centrallimittheoremsgeometry}.

\begin{theorem}[Theorem 1.2 of \cite{michelen2019centrallimittheoremsgeometry}] \label{thm:zero_freeness_converter}
Let $X$ be a random variable taking values in $\{0,1,\dots,n\}$ with mean $\mu$
and variance $\sigma^2$, and let $f_X(\zeta)=\sum_{k=0}^n \mathbb P[X=k]\zeta^k$ denote its probability generating function. Let
\(\xi_*:=\min\{|\xi-1|:f_X(\xi)=0\}\). Then
\[
\sup_{t\in\mathbb R}
\left|
\mathbb P\!\left[\frac{X-\mu}{\sigma}\le t\right]
-
\mathbb P[Z\le t]
\right|
=
O\!\left(\frac{\log n}{\xi_* \sigma}\right).
\]
\end{theorem}

In what follows, we will apply this with 
\begin{equation}\label{eq:lowFourier_notation}
    \mu:=\mathbb E_{\mu_{\lambda,t}}[B(I)],\qquad \sigma^2:=\operatorname{Var}_{\mu_{\lambda,t}}(B(I)), 
    \qquad\text{and}\qquad 
    X:=\frac{B(I)-\mu}{\sigma}.
\end{equation}

\begin{lemma} \label{lem:low_fourier_control}
Let $\lambda\in(0,\lambda_c(\Delta)-\delta)$ and $t\in[-t_0,t_0]$, and let \(\rho=\rho(\Delta,\delta,t_0)\in(0,1)\) be as in Proposition~\ref{prop:zero_free_tilted_partition}.
Then for every \(u\in\mathbb R\), for $Z \sim \cN(0,1)$,
\[
\left|
\mathbb E_{\mu_{\lambda,t}}\!\left[e^{iuX}\right]
-
\mathbb E\!\left[e^{iuZ}\right]
\right|
=
O_{\rho}\!\left(
\frac{|u|(\log n)^{3/2}+\log n}{\sigma}
\right).
\]
\end{lemma}

\begin{proof}
Set \(Y:=B(I)+|R|\). The shift by \(|R|\) converts the $B(I)$ into
an integer-valued variable in \(\{0,\ldots,n\}\), without changing its variance
or its standardized version. Its probability generating function is
\[
f_Y(\zeta):=\mathbb E_{\mu_{\lambda,t}}[\zeta^Y]
=\zeta^{|R|}\frac{Z_G(\lambda;t,\zeta)}{Z_G(\lambda;t,1)}.
\]
By Proposition~\ref{prop:zero_free_tilted_partition}, the \(Z_G\) factors have no
zero in \(|\zeta-1|<\rho\). The $\zeta^{\abs{R}}$ factor has zeros only at
\(\zeta=0\), which is outside this disk because \(\rho<1\). Hence the nearest zero
parameter \(\xi_*\) from Theorem~\ref{thm:zero_freeness_converter} satisfies
\(\xi_*\ge\rho\). Moreover,
\(\mathbb E[Y]=\mu+|R|\), \(\operatorname{Var}(Y)=\sigma^2\), and
$
\frac{Y-\mathbb E[Y]}{\sigma}
=\frac{B(I)-\mu}{\sigma}=X.
$
Applying Theorem~\ref{thm:zero_freeness_converter} to \(Y\) gives
\[
\sup_{x\in\mathbb R}
\left|\mathbb P[X\le x]-\mathbb P[Z\le x]\right|
=O_\rho\!\left(\frac{\log n}{\sigma}\right).
\]
To convert this tail difference estimate into a difference in characteristic functions, we follow the smoothing argument in \cite{jain2021approximatecountingsamplinglocal} proof of Lemma~3.3, specifically
equation~(3.1) and the displayed calculations  immediately following it. The correspondence is:
the variable \(Y=|I|\), mean \(\mu\), and standard
deviation \(\sigma\) in \cite{jain2021approximatecountingsamplinglocal} are replaced here by
\(Y=B(I)+|R|\), \(\mu+|R|\), and \(\sigma\) respectively. The required zero-freeness result is provided by
Proposition~\ref{prop:zero_free_tilted_partition}. Using the truncation cutoff \(\tau=\sqrt{8\log n}\), the calculations are identical and lead to the desired conclusion, so we omit the details.
\end{proof}

\paragraph{High Fourier phases.}

The following conditioning argument follows the proof of Lemma~3.5 of
\cite{jain2021approximatecountingsamplinglocal}; see also
\cite{DobrushinTirozzi1977}. We first require the following general lemma.

\begin{lemma}[Lemma 3.4 of \cite{jain2021approximatecountingsamplinglocal}] \label{lem:dist_4_set}
Let $G=(V,E)$ be a graph on $n$ vertices with maximum degree at most $\Delta$.
Then there exists a subset $S\subseteq V$ of size
$|S|=\Omega\!\left(\frac{n}{\Delta^3}\right)$
such that all vertices in $S$ are pairwise at distance at least $4$ in the graph metric.
Moreover, there is an algorithm to find such a subset $S$ in time $O_{\Delta}(n)$.
\end{lemma}

\begin{lemma} \label{lem:high_fourier_control}
Let $\lambda\in(0,\lambda_c(\Delta)-\delta)$, $t\in[-t_0,t_0]$, and $\mu$, $\sigma$, and $X$ be defined as in \eqref{eq:lowFourier_notation}.
Then for all $\theta\in[-\pi\sigma,\pi\sigma]$,
\[
\left|\mathbb E_{\mu_{\lambda,t}}\!\left[e^{-i\theta X}\right]\right|
\le
\exp\!\left(
-c(\Delta,t_0)\,
\lambda n\,\frac{\theta^2}{\sigma^2}
\right).
\]
\end{lemma}

\begin{proof}
We follow the conditioning calculation in the proof of Lemma~3.5 of
\cite{jain2021approximatecountingsamplinglocal}, recording the modifications needed
for the two activities in the present setting. It suffices to prove
\begin{equation}\label{eq:high_fourier_unstandardized}
\left|\mathbb E_{\mu_{\lambda,t}}\!\left[e^{-iuB(I)}\right]\right|
\le \exp(-c\lambda n u^2)
\qquad\text{for all }u\in[-\pi,\pi],
\end{equation}
for some \(c=c(\Delta,t_0)>0\). Indeed, substituting \(u=\theta/\sigma\)
and using
\[
\mathbb E_{\mu_{\lambda,t}}\!\left[e^{-i\theta X}\right]
=e^{i\theta\mu/\sigma}
\mathbb E_{\mu_{\lambda,t}}\!\left[e^{-i(\theta/\sigma)B(I)}\right]
\]
gives the stated bound whenever \(|\theta|\le\pi\sigma\).

Fix \(u\in[-\pi,\pi]\). The activities satisfy
\[
0<\lambda e^{-t_0}\le\lambda_L(t),\lambda_R(t)
\le\lambda_c(\Delta)e^{t_0}.
\]
Choose \(S\) as in Lemma~\ref{lem:dist_4_set}, so that
\(|S|=\Omega_\Delta(n)\) and distinct vertices of \(S\) have distance at
least \(4\). Let \(T:=\{x:\operatorname{dist}_G(x,S)\ge2\}\), sample
\(J:=I\cap T\), and condition on \(J\). For \(v\in S\), write
\[
H_v:=G[\{v\}\cup N(v)],\qquad
A_v(J):=\{w\in N(v):N(w)\cap J=\varnothing\},\qquad
m_v:=|A_v(J)|.
\]
Thus \(H_v\) is the graph consisting of \(v\) and its neighbours, and
\(A_v(J)\) is the subset of those neighbours that may still be occupied.
The graphs \(H_v\) are disjoint and have no edges between them. Hence the
random variables \(W_v:=B(I\cap V(H_v))\) are independent conditional on
\(J\), and
\[
B(I)\overset{d}{=}B(J)+\sum_{v\in S}W_v\qquad\text{conditional on }J.
\]

Let \(\varepsilon_v:=1\) for \(v\in L\) and \(\varepsilon_v:=-1\) for
\(v\in R\). If \(\lambda_v\) and \(\lambda_{\bar v}\) denote respectively
the activities on the side containing \(v\) and on the opposite side, then
\[
Z_v(J)=\lambda_v+(1+\lambda_{\bar v})^{m_v},\qquad
\mathbb P(W_v=0\mid J)=\frac1{Z_v(J)},\qquad
\mathbb P(W_v=\varepsilon_v\mid J)=\frac{\lambda_v}{Z_v(J)}.
\]
Since \(m_v\le\Delta\) and the activities obey the preceding uniform bounds,
there are \(a_0,a_1>0\), depending only on \(\Delta,t_0\), such that
\[
\mathbb P(W_v=0\mid J)\ge a_0,\qquad
\mathbb P(W_v=\varepsilon_v\mid J)\ge a_1\lambda. \tag{*}
\]

Let \(W_v'\) be an independent conditional copy of \(W_v\). Since
\(W_v-W_v'\) has a symmetric conditional distribution and
\(1-\cos u\ge2u^2/\pi^2\) for \(|u|\le\pi\), the two events in \((*)\) give
\[
\begin{aligned}
\left|\mathbb E[e^{-iuW_v}\mid J]\right|^2
&=\mathbb E[\cos(u(W_v-W_v'))\mid J]\\
&\le1-2\mathbb P(W_v-W_v'=\varepsilon_v\mid J)(1-\cos u)\\
&\le1-\frac{4a_0a_1}{\pi^2}\lambda u^2.
\end{aligned}
\]
Indeed, the probability in the middle line is at least
\(\mathbb P(W_v=\varepsilon_v\mid J)\mathbb P(W_v'=0\mid J)\).
Conditional independence, followed by \(|S|=\Omega_\Delta(n)\), now gives
\[
\left|\mathbb E[e^{-iuB(I)}\mid J]\right|
\le\prod_{v\in S}\left|\mathbb E[e^{-iuW_v}\mid J]\right|
\le\exp(-c\lambda n u^2).
\]
Averaging over \(J\) proves \eqref{eq:high_fourier_unstandardized}, and hence
the lemma.
\end{proof}

\paragraph{Completing the proof of Theorem \ref{thm:main_local_clt}.}

With Lemmas \ref{lem:balance_variance_bounds}, \ref{lem:low_fourier_control}, and \ref{lem:high_fourier_control} in hand, we are now able to complete the proof of Theorem \ref{thm:main_local_clt}.

\begin{proof}[Proof of Theorem \ref{thm:main_local_clt}]
Let
\[
X:=\frac{B(I)-\mu}{\sigma},
\qquad
\phi_X(\theta):=\mathbb E_{\mu_{\lambda,t}}\!\left[e^{i\theta X}\right],
\qquad
\phi_Z(\theta):=\mathbb E\!\left[e^{i\theta Z}\right]=e^{-\theta^2/2},
\]
where \(Z\sim N(0,1)\). Since \(B(I)\in \mathbb Z\), the random variable \(X\) is supported on the
lattice $-\mu/\sigma+\sigma^{-1}\mathbb Z$.
Therefore Lemma \ref{lem:fourier-inversion} gives
\[
\sup_{b\in\mathbb Z}
\left|
\sigma^{-1}N\!\left(\frac{b-\mu}{\sigma}\right)-\mathbb P_{\mu_{\lambda,t}}[B(I)=b]
\right|
\le
\frac1\sigma
\int_{-\pi\sigma}^{\pi\sigma}
\left|
\phi_X(\theta)-\phi_Z(\theta)
\right|\,d\theta
+
e^{-\pi^2\sigma^2/2}.
\]

Choose a constant \(A>0\), to be fixed later, and set $T:=A\sqrt{\log n}$. We split the Fourier integral at \(T\). Using Lemmas
\ref{lem:low_fourier_control}, \ref{lem:high_fourier_control}, and
\ref{lem:balance_variance_bounds}, we obtain
\begin{align*}
&\sup_{b\in\mathbb Z}
\left|
\sigma^{-1}N\!\left(\frac{b-\mu}{\sigma}\right)-\mathbb P_{\mu_{\lambda,t}}[B(I)=b]
\right| \\
&\le
\frac1\sigma
\int_{-\pi\sigma}^{\pi\sigma}
\left|
\phi_X(\theta)-\phi_Z(\theta)
\right|\,d\theta
+
e^{-\pi^2\sigma^2/2} \\
&\le
\frac1\sigma
\int_{|\theta|\le T}
\left|
\phi_X(\theta)-\phi_Z(\theta)
\right|\,d\theta
+
\frac1\sigma
\int_{T\le |\theta|\le \pi\sigma}
\left|
\phi_X(\theta)-\phi_Z(\theta)
\right|\,d\theta
+
e^{-\pi^2\sigma^2/2} \\
&\le
\frac{C_1}{\sigma^2}
\int_0^T
\bigl(\theta(\log n)^{3/2}+\log n\bigr)\,d\theta
+
\frac1\sigma
\int_{T\le |\theta|\le \pi\sigma}
\bigl(|\phi_X(\theta)|+|\phi_Z(\theta)|\bigr)\,d\theta
+
e^{-\pi^2\sigma^2/2} \\
&\le
\frac{C_2}{\sigma^2}
\bigl(T^2(\log n)^{3/2}+T\log n\bigr)
+
\frac1\sigma
\int_{T\le |\theta|\le \pi\sigma}
\exp\!\left(-c\,\lambda n\,\frac{\theta^2}{\sigma^2}\right)\,d\theta
+
\frac1\sigma
\int_T^\infty e^{-\theta^2/2}\,d\theta
+
e^{-\pi^2\sigma^2/2}.
\end{align*}
By Lemma \ref{lem:balance_variance_bounds}, we have \(\sigma^2\le C_3 n\), so $\lambda n\,\frac{\theta^2}{\sigma^2}
\ge
\frac{\lambda}{C_3}\theta^2.$ So, after changing constants,
\begin{align*}
&\sup_{b\in\mathbb Z}
\left|
\sigma^{-1}N\!\left(\frac{b-\mu}{\sigma}\right)-\mathbb P_{\mu_{\lambda,t}}[B(I)=b]
\right| \\
&\le
\frac{C_2}{\sigma^2}
\bigl(T^2(\log n)^{3/2}+T\log n\bigr)
+
\frac1\sigma
\int_{T\le |\theta|\le \pi\sigma}
e^{-c'\theta^2}\,d\theta
+
\frac1\sigma
\int_T^\infty e^{-\theta^2/2}\,d\theta
+
e^{-\pi^2\sigma^2/2} \\
&\le
\frac{C_5(\log n)^{5/2}}{\sigma^2}
+
2\pi e^{-c'T^2}
+
\frac{C_4}{\sigma T}e^{-T^2/2}
+
\frac{2}{\pi^2\sigma^2} \\
&=
\frac{C_5(\log n)^{5/2}}{\sigma^2}
+
2\pi n^{-c'A^2}
+
\frac{C_4}{\sigma\sqrt{\log n}}\,n^{-A^2/2}
+
\frac{2}{\pi^2\sigma^2} \\
&=
O\!\left(\frac{(\log n)^{5/2}}{\sigma^2}\right),
\end{align*}
where in the last step we used that \(T=A\sqrt{\log n}\), chose \(A\) sufficiently large, and used \(\sigma^2\le C_3 n\). Finally, the lower bound in Lemma~\ref{lem:balance_variance_bounds} gives
$
\sigma^2\ge c_{\Delta,\delta,t_0}\lambda n
$. This finishes the proof.
\end{proof}

\subsection{Acceptance probability estimate}

\begin{theorem}\label{thm:bounded_ratio_centered_tilt_acceptance}
Fix \(\Delta\ge3\), \(\lambda<\lambda_c(\Delta)\), and \(\gamma\ge1\), and suppose \eqref{eq:bounded_imbalance_regime} holds. Then there exists
\(t_0=t_0(\Delta,\lambda,\gamma)>0\) such that the following holds: Let \(t^*\) be the unique centering tilt from
Lemma~\ref{lem:unique_centering_tilt}, then we have
\begin{enumerate}[(i)]
    \item \label{item:bounded_ratio_compact_tstar}
    The centering tilt lies in the fixed compact interval
    \(t^*\in[-t_0,t_0]\).
    \item \label{item:bounded_ratio_acceptance}
    For every fixed \(A>0\), uniformly over
    \(t\in[-t_0,t_0]\) satisfying
    \(|\E_{\mu_{\lambda,t}}[B(I)]|\le A\),
    \[
    \Pr_{\mu_{\lambda,t}}\{B(I)=0\}
    =
    \Omega_{\Delta,\lambda,\gamma,A}\!\left(\frac1{\sqrt n}\right).
    \]
\end{enumerate}
\end{theorem}

\begin{proof}
Let \(\dagger^c\) denote the side opposite \(\dagger\). We first prove a
crude one-vertex marginal bound. Fix \(x\in\dagger\in\{L,R\}\). The occupation
probability of \(x\) is at most its activity ratio
\(\lambda_\dagger(t)/(1+\lambda_\dagger(t))\). For the lower bound, we first
force all neighbors of \(x\) to be unoccupied. Expose these neighbors one at a
time. At each step, the next neighbor has activity
\(\lambda_{\dagger^c}(t)\), so its conditional probability of being
unoccupied is at least \((1+\lambda_{\dagger^c}(t))^{-1}\).
Since \(x\) has at most \(\Delta\) neighbors,
\begin{equation}\label{eq:occupationProb_upperLowerBounds}
\frac{\lambda_\dagger(t)}{1+\lambda_\dagger(t)}
(1+\lambda_{\dagger^c}(t))^{-\Delta}
\le
\Pr_{\mu_{\lambda,t}}(x\in I)
\le
\frac{\lambda_\dagger(t)}{1+\lambda_\dagger(t)}.
\end{equation}
Writing \(\ell:=|L|\) and \(r:=|R|\), and using
\[
m(t)
=\sum_{u\in L}\Pr_{\mu_{\lambda,t}}(u\in I)
-\sum_{v\in R}\Pr_{\mu_{\lambda,t}}(v\in I),
\]
we get
\[
\ell\cdot \frac{\lambda_L(t)}{1+\lambda_L(t)}(1+\lambda_R(t))^{-\Delta}
-
r\cdot \frac{\lambda_R(t)}{1+\lambda_R(t)}
\le
m(t)
\le
\ell\cdot \frac{\lambda_L(t)}{1+\lambda_L(t)}
-
r\cdot \frac{\lambda_R(t)}{1+\lambda_R(t)}(1+\lambda_L(t))^{-\Delta}.
\]

Since \(1\le\ell/r\le\gamma\), these bounds imply that there exists
\(t_0=t_0(\Delta,\lambda,\gamma)>0\) such that
\[
m(-t_0)<0<m(t_0).
\]
Indeed, after dividing by \(r\), the lower bound tends to \(\ell/r\ge1\) as
\(t\to\infty\), while the upper bound tends to \(-1\) as \(t\to-\infty\),
uniformly over \(\ell/r\in[1,\gamma]\). Since \(m\) is strictly increasing,
the unique zero \(t^*\) lies in \([-t_0,t_0]\). This proves
\ref{item:bounded_ratio_compact_tstar}.

Set \(\delta:=(\lambda_c(\Delta)-\lambda)/2>0\). Having met the conditions
of Theorem~\ref{thm:main_local_clt}, it follows that, for every fixed \(A>0\),
uniformly over \(t\in[-t_0,t_0]\) satisfying \(|m(t)|\le A\),
\begin{align}
    \Pr_{\mu_{\lambda,t}}\{B(I)=0\}
    &=
    \sigma_t^{-1} N\!\left(\frac{-m(t)}{\sigma_t}\right)
    +
    O\!\left(\frac{(\log n)^{5/2}}{\sigma_t^2}\right) \nonumber\\
    &\ge
    \frac{c_3}{\sqrt n}
    -
    C_2\frac{(\log n)^{5/2}}{n} \nonumber \\
&=
\Omega\!\left(\frac{1}{\sqrt n}\right).
\end{align}
Indeed, the Gaussian density \(N(-m(t)/\sigma_t)\) is bounded below by a
positive constant, \(\sigma_t=\Theta(\sqrt n)\), and the error term is lower
order. This proves \ref{item:bounded_ratio_acceptance}.
\end{proof}

\subsection{Proof of Proposition~\ref{prop:bounded_imbalance_sampler}}

\begin{proof}[Proof of Proposition~\ref{prop:bounded_imbalance_sampler}]
Fix \(\epsilon\in(0,1)\).
Let \(t_0\) be as in
Theorem~\ref{thm:bounded_ratio_centered_tilt_acceptance}, and let \(t^*\) be
the unique centering tilt. Specialize
Step~\ref{line:balanced_sampler_template_regime_choice} by taking
\[
    \mathcal T=[-t_0,t_0],\qquad
    \mathsf{TiltedSampler}=\text{Algorithm~\ref{alg:tilted-sampler}},
    \qquad
    \theta=2.
\]
We implement the empirical search in Algorithm~\ref{alg:balanced_hardcore_sampler_template}
by bisection on \([a,b]=[-t_0,t_0]\), with
\[
    J:=\max\{1,\lceil\log_2(n^2(b-a)/\theta)\rceil\},\qquad
    N:=\lceil 8\theta^{-2}n^2\log(20J/\epsilon)\rceil,
\]
and we use
\[
    M:=\lceil C\sqrt n\log(10/\epsilon)\rceil,
    \qquad
    \tau:=\epsilon/(10(JN+M))
\]
for the final rejection stage and for the accuracy of each tilted-sampler call.
Choose \(C=C(\Delta,\lambda,\gamma)\) sufficiently large below. We first analyze an
idealized \emph{exact-proposal} version of Algorithm~\ref{alg:balanced_hardcore_sampler_template},
in which every call to \(\mathsf{TiltedSampler}\), both during the empirical
search and during the final rejection stage, returns an exact independent
sample from \(\mu_{\lambda,t}\).

Conclusion~\ref{item:bounded_ratio_compact_tstar} of
Theorem~\ref{thm:bounded_ratio_centered_tilt_acceptance} gives
\(t^*\in[-t_0,t_0]\). By Lemma~\ref{lem:empirical-balance-concentration} with
\(\eta=1\), a union bound gives an event \(\cE\) of probability at least
\(1-\epsilon/10\) on which every empirical mean queried during the search
is within \(1\) of its expectation. On \(\cE\), by Lemma~\ref{lem:balance_variance_bounds}
and the identity
\(m'(t)=\sigma_t^2\), there is a constant
\(C_3=C_3(\Delta,\lambda,\gamma)>0\) such that
\[
|m(t)-m(t^*)|
\le C_3n|t-t^*|
\qquad\text{for all }t\in[-t_0,t_0].
\]
An early stop gives \(|m(\widetilde t)|\le3\). Otherwise, at every iteration in
which the search continues, \(|\widehat m|>2\), so \(\widehat m\) and \(m(t)\)
have the same sign.
Thus \(t^*\) remains in the retained interval. Its final width is at most
\(2/n^2\), so \(|m(\widetilde t)|\le2C_3\). Hence, with
\(A:=\max\{3,2C_3\}\), Theorem ~\ref{thm:bounded_ratio_centered_tilt_acceptance} \ref{item:bounded_ratio_acceptance} gives
\[
\Pr_{\mu_{\lambda,\widetilde t}}\{B(I)=0\}
\ge \frac{c}{\sqrt n}
\]
for some \(c=c(\Delta,\lambda,\gamma)>0\). Choose \(C\) so that
$
\exp(-cM/\sqrt n)\le\epsilon/10.
$

Then, conditional on \(\cE\), the \(M\) rejection trials fail to produce a
balanced sample with probability at most \(\epsilon/10\); on this event the
algorithm outputs the empty independent set. Every successful output has law
\(\mu_{G,\lambda}^{\mathrm{bal}}\), since
\(e^{\widetilde t B(I)}=1\) on \(B(I)=0\). Thus the exact-proposal version is within total variation distance
\(\epsilon/5\) of \(\mu_{G,\lambda}^{\mathrm{bal}}\).

The actual implementation makes at most \(JN+M\)
\(\mathsf{TiltedSampler}\) calls. By the same coupling argument as in
\eqref{eq:tilted-sampler-coupling} and a union bound, its output law is within
\((JN+M)\tau\le\epsilon/10\) in total variation distance of that of the
exact-proposal version.

The total error is at most \( \epsilon/5 + \epsilon/10 =  3\epsilon/10\le\epsilon\). All searched
tilts lie in \([-t_0,t_0]\), where Algorithm~\ref{alg:tilted-sampler} runs in
polynomial time, and all parameters are polynomial in \(n\) and
\(1/\epsilon\).
\end{proof}

\section{FPTAS for the balanced partition function}
\label{sec:fptas_balanced_partition}

This section establishes the FPTAS assertion below the uniqueness threshold in
Theorem~\ref{thm:comp_thresh}. In fact, we prove the stronger assertion for
the class \(\widetilde{\cG}_{\Delta,\Delta}^\gamma\) defined in
\eqref{eq:tilde_graph_class}.

\begin{proposition}
\label{prop:compact_tilt_fptas_balanced_partition}
Fix \(\Delta\ge3\), \(0<\lambda<\lambda_c(\Delta)\), and \(\gamma\ge1\). Then there is an
FPTAS for \(Z_G^{\mathrm{bal}}(\lambda)\) on inputs
\(G\in\widetilde{\cG}_{\Delta,\Delta}^\gamma\).
\end{proposition}

Throughout this section, fix \(\gamma\ge1\) and
\(G=(L\sqcup R,E)\in\widetilde{\cG}_{\Delta,\Delta}^\gamma\) and write
\(n:=|V(G)|=|L|+|R|\). We also fix
\(0<\lambda<\lambda_c(\Delta)\). As in
Section~\ref{sec:below_uniqueness}, after possibly exchanging the two sides of
the bipartition, we assume that \(|L|\ge|R|\). Consequently,
\eqref{eq:bounded_imbalance_regime} holds.

Let \(t_0=t_0(\Delta,\lambda,\gamma)>0\) be the compact-tilt constant from
Theorem~\ref{thm:bounded_ratio_centered_tilt_acceptance}. By
\eqref{eq:bounded_imbalance_regime}, the centering tilt \(t^*\), characterized
by
\(\E_{\mu_{\lambda,t^*}}B(I)=0\), lies in the compact interval
\([-t_0,t_0]\). All tilts supplied to the subroutines below will lie in this
interval.

A proof overview is given next. The key basic identity is that, for every tilt \(t\),
\begin{equation}\label{eq:balanced_counting_identity}
    Z_G^{\mathrm{bal}}(\lambda)
    =
    Z_G(\lambda;t)\,
    \Pr_{\mu_{\lambda,t}}\{B(I)=0\}.
\end{equation}

Fix a constant \(A>0\), depending only on \(\Delta\) and \(\lambda\).
Based on \eqref{eq:balanced_counting_identity}, an FPTAS for $Z_G^{\mathrm{bal}}(\lambda)$ requires the following three tasks, which we want to implement in time polynomial in
\(n\) and \(1/\epsilon\) and with constants depending only on
\(\Delta,\lambda,\gamma,A\). Recall the notation
$
    m(t):=\E_{\mu_{\lambda,t}}B(I).
$
\begin{enumerate}[leftmargin=*]
    \item Find a deterministic tilt \(\widetilde t\in[-t_0,t_0]\) such that
    \(|m(\widetilde t)|\le A\).
    \item Given \(\widetilde t\), approximate \(Z_G(\lambda;\widetilde t)\) to relative error
    \(1\pm O(\epsilon)\).
    \item Given \(\widetilde t\), approximate
    \(\Pr_{\mu_{\lambda,\widetilde t}}\{B(I)=0\}\) to relative error
    \(1\pm O(\epsilon)\).
\end{enumerate}
Splitting the error budget among the last two multiplicative approximations
and using \eqref{eq:balanced_counting_identity} then gives an
\(\epsilon\)-relative approximation to \(Z_G^{\mathrm{bal}}(\lambda)\).
The first two steps are standard consequences of the compact-tilt SAW-tree
approximation algorithms. The last step involves "algorithmizing" the local CLT
proof, in the spirit of \cite{jain2021approximatecountingsamplinglocal}.

\subsection{\texorpdfstring{Finding a deterministic tilt and approximating \(Z_G\)}{Finding a deterministic tilt and approximating ZG}}

We first show that a deterministic tilt search can replace
the empirical search used in the sampling algorithm.

\begin{lemma}
\label{lem:deterministic_compact_tilt_centering}
Suppose \eqref{eq:bounded_imbalance_regime} holds, and let
\(A>0\). There is a deterministic algorithm which returns a tilt
\(\widetilde t\in[-t_0,t_0]\) satisfying
\[
    |m(\widetilde t)|\le A
\]
in time polynomial in \(n\), possibly depending on the constants
\(\Delta,\lambda,\gamma,A\).
\end{lemma}

\begin{proof}
Set \(\eta:=A/10\). At each queried tilt \(r\in[-t_0,t_0]\), use
Algorithm~\ref{alg:saw-trunc} to approximate every occupation marginal to
additive error \(\eta/n\). Equation~\eqref{eq:m_expectedBalance_def} then gives
an estimate \(\widehat m(r)\) with \(|\widehat m(r)-m(r)|\le\eta\).

Run bisection on \([-t_0,t_0]\) as in the proof of
Proposition~\ref{prop:bounded_imbalance_sampler}. If
\(|\widehat m(r)|<2\eta\), return \(r\), since then
\(|m(r)|<3\eta\le A\). Otherwise \(\widehat m(r)\) and \(m(r)\) have the same
sign, so the retained interval contains \(t^*\). If no early stop occurs,
continue until its width is at most \(A/(4C_{\mathrm{var}}n)\), where
\(\Var_{\mu_{\lambda,s}}(B(I))\le C_{\mathrm{var}}n\) on
\([-t_0,t_0]\) by Lemma~\ref{lem:balance_variance_bounds}. Its midpoint
\(\widetilde t\) then satisfies, by \(m'(s)=\Var_{\mu_{\lambda,s}}(B(I))\),
\[
|m(\widetilde t)|\le C_{\mathrm{var}}n|\widetilde t-t^*|
\le A/8\le A.
\]
The search makes \(O(\log n)\) bisection steps, each using \(n\)
marginal oracle calls, and is therefore polynomial-time.
\end{proof}

\begin{lemma}
\label{lem:real_compact_tilt_partition_oracle}
Suppose \eqref{eq:bounded_imbalance_regime} holds. There is a
deterministic algorithm which, given a tilt \(t\in[-t_0,t_0]\) and
\(\epsilon\in(0,1)\), outputs an
\(\epsilon\)-relative approximation to \(Z_G(\lambda;t)\) in time
\(\poly(n,1/\epsilon)\), with constants depending on
\(\Delta,\lambda,\gamma\).
\end{lemma}

\begin{proof}
This is the standard hard-core self-reduction
\cite[Section~5]{weitz2006counting}, also used for
Algorithm~\ref{alg:tilted-sampler}. For an ordering \(v_1,\ldots,v_n\), let
\(p_i\) be the occupation probability of \(v_i\) in the residual graph
\(G-\{v_1,\ldots,v_{i-1}\}\), with its inherited activities. Then
\[
    Z_G(\lambda;t)=\prod_i(1-p_i)^{-1}.
\]
These activities remain in the compact tilt window, so the SAW-tree oracle
approximates each \(p_i\) to additive accuracy \(O(\epsilon/n)\), while
\(1-p_i\) is uniformly bounded away from zero. The product above then gives
the claimed relative approximation.
\end{proof}

\subsection{Algorithmic local CLT for compact tilts}

The subroutine for estimating the factor $\Pr_{\mu_{\lambda,t}}\{B(I)=0\}$ in \eqref{eq:balanced_counting_identity} is based on an algorithmic version of the compact-tilt
local CLT, Theorem~\ref{thm:main_local_clt}. 

\begin{proposition}
\label{prop:algorithmic_lclt_compact_tilts}
Suppose \eqref{eq:bounded_imbalance_regime} holds, and let
\(A>0\). There is a deterministic algorithm which, given a tilt
\(t\in[-t_0,t_0]\) satisfying \(|m(t)|\le A\), and
\(\epsilon\in(0,1)\), outputs an \(\epsilon\)-relative approximation to
\[
    \Pr_{\mu_{\lambda,t}}\{B(I)=0\}
\]
in time \(\poly(n,1/\epsilon)\), with constants depending on
\(\Delta,\lambda,\gamma,A\).
\end{proposition}

Our proof strategy for Proposition \ref{prop:algorithmic_lclt_compact_tilts} is given next. Recall the tilted complex partition function \eqref{eq:tiltedComplexHardcorePartitionFunction}: for \(\zeta\in\CC\)
\[
    Z_G(\lambda;t,\zeta)
    =
    \sum_{I\in\cI_G}
    \lambda^{|I|}e^{tB(I)}\zeta^{B(I)}.
\]
Fourier inversion on the lattice \(\mathbb Z\), followed by the change of
variables \(\theta=\sqrt n\,u\), gives
\begin{equation}\label{eq:balance_zero_fourier_ratio}
    \Pr_{\mu_{\lambda,t}}\{B(I)=0\}
    =
    \frac{1}{2\pi\sqrt n}
    \int_{-\pi\sqrt n}^{\pi\sqrt n}
    \frac{Z_G(\lambda;t,e^{i\theta/\sqrt n})}
         {Z_G(\lambda;t)}
    \,d\theta.
\end{equation}
The algorithmic tasks are:
\begin{itemize}[leftmargin=*]
    \item Approximate the partition function ratio
    \(Z_G(\lambda;t,\zeta)/Z_G(\lambda;t,1)\) for \(\zeta\) near \(1\). This will be done using the complex SAW-tree which gives approximations to the
    occupation ratios.
    \item Approximate the Fourier integral separately for low and high frequencies. Let
    \(T=O(\sqrt{\log(1/\epsilon)})\). By
    taking the \(\sqrt n\)-scale, the quantity \(e^{i\theta/\sqrt n}\) is close to \(1\)
    on \(|\theta|\le T\), so the integrand is in the range of the ratio
    approximation from the first bullet. We approximate this part of the integral by
    evaluating the ratio at finitely many evenly spaced points. For \(|\theta|>T\),
    we reuse the compact-tilt high-frequency estimate from the local CLT proof,
    Lemma~\ref{lem:high_fourier_control}, to show that this part of the
    integral is negligible.
\end{itemize}
This approach is inspired by that in \cite{jain2021approximatecountingsamplinglocal}. A difference is that here, we estimate the partition function ratio in \eqref{eq:balance_zero_fourier_ratio} directly. The route taken in \cite{jain2021approximatecountingsamplinglocal} via
zero-freeness and Barvinok~\cite{barvinok2016combinatorics} and Patel--Regts \cite{patel2017deterministic} to approximate the characteristic function of $B(I)$ would likely also work. However, our approach leverages the complex SAW-tree machinery already developed in Section \ref{sec:zeroFreenessTiltedPartitionFunction}.

Our next result gives the complex ratio oracle needed for the low-frequency
partition function approximation. The proof builds on the zero-freeness argument in
Section~\ref{sec:zeroFreenessTiltedPartitionFunction}, especially
Lemma~\ref{lem:zeroFreeness_twoLevel_complex_thickening} and the proof of
Proposition~\ref{prop:zero_free_tilted_partition}.

Throughout the next two lemmas, for fixed \(t\) and \(\zeta\), an induced
subgraph \(H\subseteq G\), and a vertex \(u\in V(H)\), write
\[
    R_{H,u}
    :=
    \frac{\lambda_{\dagger(u)}(\zeta)\,
    Z_{H\setminus N_H[u]}(\lambda;t,\zeta)}
    {Z_{H-u}(\lambda;t,\zeta)}.
\]
This is defined whenever the denominator is nonzero. To emphasize
the dependence on \(\zeta\), we may write \(R_{H,u}(\zeta)\). An elementary
decomposition of independent sets according to whether \(u\) is absent or
present gives
\begin{equation}\label{eq:complex_self_reduction_one_step}
    Z_H(\lambda;t,\zeta)
    =
    Z_{H-u}(\lambda;t,\zeta)\bigl(1+R_{H,u}\bigr).
\end{equation}

\begin{lemma}
\label{lem:complex_ssm_ratio_oracle}
Suppose \eqref{eq:bounded_imbalance_regime} holds. Let
\(\delta_\lambda:=(\lambda_c(\Delta)-\lambda)/2\), and let \(\rho\) be
the radius from Proposition~\ref{prop:zero_free_tilted_partition}, applied with
slack \(\delta_\lambda\) and compact tilt interval \([-t_0,t_0]\). Then for every
\(\rho'<\rho\), there is a constant
\(\kappa>0\), depending only on \(\Delta,\lambda,\gamma,\rho'\), such that the
following holds.

Let \(t\in[-t_0,t_0]\), let \(\zeta\in\CC\) satisfy
\(|\zeta-1|\le\rho'\), let \(H\subseteq G\) be an induced subgraph, and let
\(u\in V(H)\). Then:
\begin{enumerate}[label=(\roman*),leftmargin=*]
    \item The ratio \(R_{H,u}\) is well defined and satisfies
    $
        |1+R_{H,u}|\ge \kappa .
    $
    \item There is a deterministic algorithm which, given
    \(\eta\in(0,1)\), outputs \(\widehat R_{H,u}\) satisfying
    \[
        |\widehat R_{H,u}-R_{H,u}|\le \eta
    \]
    in time \(\poly(|V(H)|,1/\eta)\), with constants depending only on
    \(\Delta,\lambda,\gamma,\rho'\).
\end{enumerate}
\end{lemma}

\begin{proof}
This is the algorithmic version of the zero-freeness argument in
Proposition~\ref{prop:zero_free_tilted_partition}.  We use the same complex
domain \(D_\alpha\) and \(\Omega=\phi^{-1}(D_\alpha)\), chosen with slack
\(\delta_\lambda\), and work on the compact subdisk \(|\zeta-1|\le\rho'<\rho\).

The derivative estimate
\eqref{eq:zeroFreeness_complex_twoLevel_jacobian_gap} gives a uniform two-level
contraction for the complex SAW-tree recursion on \(D_\alpha\).  Iterating this
contraction as in the proof of Theorem~\ref{thm:ssm}, and treating the root
separately because it may have \(\Delta\) children, shows that changing boundary
data below depth \(L\) changes the root ratio by at most \(Ce^{-cL}\), uniformly
for \(t\in[-t_0,t_0]\) and \(|\zeta-1|\le\rho'\).

For positive real activities, Weitz's SAW-tree identity identifies \(R_{H,u}\)
with the root ratio of \(T_{\mathrm{SAW}}(H,u)\), with the prescribed
cycle-closing leaves. Both sides are rational functions of the activities, and
Proposition~\ref{prop:zero_free_tilted_partition} gives the nonvanishing needed
for these rational functions to be defined for
\(t\in[-t_0,t_0]\) and \(|\zeta-1|\le\rho'\).
Hence the identity extends to these complex activities.  Truncating the SAW tree
at depth \(L=O_{\Delta,\lambda,\gamma,\rho'}(\log(1/\eta))\), while keeping the
prescribed cycle-closing leaves exact, gives an additive \(\eta\)-approximation
to \(R_{H,u}\) in time \(\poly(|V(H)|,1/\eta)\).

It remains only to note that \(1+R_{H,u}\) is uniformly bounded away from zero.
This is exactly the final separation estimate in the proof of
Proposition~\ref{prop:zero_free_tilted_partition}: if \(u\) has at most
\(\Delta-1\) neighbors in \(H\), then \(R_{H,u}\in\Omega\); if \(u\) has
\(\Delta\) neighbors, the same decomposition used there writes
\(Q=R_{H,u}(1+S)\) with \(Q,S\in\Omega\). Since
\(\rho'<\rho\), the compactness of the subdomain and the containment
\(\phi^{-1}(\overline{D_\alpha})\subset\{z:\operatorname{Re}z>-1/2\}\) give a
constant \(\kappa=\kappa(\Delta,\lambda,\gamma,\rho')>0\) such that
\(|1+R_{H,u}|\ge\kappa\) in both cases.
\end{proof}

\begin{definition}
\label{def:complex_relative_approximation}
For \(z_1,z_2\in\CC\), we say that \(z_1\) is a
\(\delta\)-additive, \(\epsilon\)-relative approximation to \(z_2\) if
\[
    z_1=re^{i\theta}z_2+z_3
    \quad\text{for some}\quad
    e^{-\epsilon}\le r\le e^\epsilon,\quad
    |\theta|\le\epsilon,\quad
    |z_3|\le\delta .
\]
When \(\delta=0\), we simply call \(z_1\) an
\(\epsilon\)-relative approximation to \(z_2\).
\end{definition}

\begin{lemma}
\label{lem:low_frequency_complex_partition_oracle}
Suppose \eqref{eq:bounded_imbalance_regime} holds. Let
\(\delta_\lambda:=(\lambda_c(\Delta)-\lambda)/2\), and let \(\rho\) be
the zero-free radius from Proposition~\ref{prop:zero_free_tilted_partition}, applied
with slack \(\delta_\lambda\) and compact tilt interval \([-t_0,t_0]\). For every
\(\rho'<\rho\), there is a deterministic algorithm which, given
\(t\in[-t_0,t_0]\), \(\zeta\in\CC\) with
\(|\zeta-1|\le\rho'\), and \(\epsilon\in(0,1)\), outputs an
\(\epsilon\)-relative approximation to the ratio
\[
    \frac{Z_G(\lambda;t,\zeta)}{Z_G(\lambda;t,1)}
\]
in time \(\poly(n,1/\epsilon)\), with constants depending on
\(\Delta,\lambda,\gamma,\rho'\).
\end{lemma}

\begin{proof}
Since \(|\zeta-1|\le\rho'<\rho<1\), we have \(\zeta\neq0\). Fix an ordering
\(v_1,\dots,v_n\) of \(V(G)\), and let
\(G_i:=G-\{v_1,\dots,v_{i-1}\}\). Iterating the self-reduction
\eqref{eq:complex_self_reduction_one_step}, first at \(\zeta\) and then at
\(\zeta=1\), gives
\begin{equation}\label{eq:complex_partition_product_quotients}
    \frac{Z_G(\lambda;t,\zeta)}{Z_G(\lambda;t,1)}
    =
    \prod_{i=1}^n
    \frac{1+R_{G_i,v_i}(\zeta)}{1+R_{G_i,v_i}(1)}.
\end{equation}

Let \(\kappa>0\) be the constant from
Lemma~\ref{lem:complex_ssm_ratio_oracle}, chosen small enough so that
\(\kappa\le1\), and set \(\eta:=\epsilon/(20n)\). Use that lemma to
approximate each \(R_{G_i,v_i}(\zeta)\) and \(R_{G_i,v_i}(1)\) to additive
accuracy \(\kappa\eta\), and denote the approximations by
\(\widehat R_i(\zeta)\) and \(\widehat R_i(1)\). The lower bound
\[
|1+R_{G_i,v_i}(\zeta)|,\ |1+R_{G_i,v_i}(1)|\ge\kappa
\]
implies that, for some \(\delta_i,\delta_i'\in\CC\) with
\(|\delta_i|,|\delta_i'|\le\eta\),
\[
1+\widehat R_i(\zeta)
=(1+R_{G_i,v_i}(\zeta))(1+\delta_i),
\qquad
1+\widehat R_i(1)
=(1+R_{G_i,v_i}(1))(1+\delta_i').
\]
The algorithm's output is taken as
\[
    \widehat Q
    :=
    \prod_{i=1}^n
    \frac{1+\widehat R_i(\zeta)}{1+\widehat R_i(1)}.
\]
Combining the above estimates gives
\[
    \widehat{Q} \bigg/ \frac{Z_G(\lambda;t,\zeta)}{Z_G(\lambda;t,1)}
    =
    \prod_{i=1}^n \frac{1+\delta_i}{1+\delta_i'}
    =
    \exp(S),
    \qquad\text{where}\qquad 
    S:=\sum_{i=1}^n
    \operatorname{Log}\frac{1+\delta_i}{1+\delta_i'}.
\]
Here \(\operatorname{Log}\) denotes the principal logarithm. Since
\(|\operatorname{Log}(1+z)|\le2|z|\) whenever \(|z|\le1/2\), we have
\[
|S|\le4n\eta=\epsilon/5.
\]
Writing \(S=a+ib\), both \(|a|\) and \(|b|\) are at most
\(\epsilon/5\).
Thus \(e^S=e^ae^{ib}\) satisfies Definition~\ref{def:complex_relative_approximation},
so \(\widehat Q\) is an \(\epsilon\)-relative approximation. There are
\(2n\) calls to Lemma~\ref{lem:complex_ssm_ratio_oracle}, each with requested
accuracy \(\kappa\epsilon/(20n)\), so the total running time is polynomial
in \(n\) and \(1/\epsilon\).
\end{proof}

\begin{proof}[Proof of Proposition~\ref{prop:algorithmic_lclt_compact_tilts}]
Write
\[
    g(\theta)
    :=
    \frac{Z_G(\lambda;t,e^{i\theta/\sqrt n})}{Z_G(\lambda;t,1)}.
\]
By \eqref{eq:balance_zero_fourier_ratio},
\[
    \Pr_{\mu_{\lambda,t}}\{B(I)=0\}
    =
    \frac{1}{2\pi\sqrt n}
    \int_{-\pi\sqrt n}^{\pi\sqrt n}g(\theta)\,d\theta .
\]

By Theorem~\ref{thm:bounded_ratio_centered_tilt_acceptance} and the assumption
\(|m(t)|\le A\), there is a constant \(c_A=c_A(\Delta,\lambda,\gamma,A)>0\) such that
\begin{equation}\label{eq:algorithmic-lclt-probability-lower-bound}
\Pr_{\mu_{\lambda,t}}\{B(I)=0\}\ge\frac{c_A}{\sqrt n}.
\end{equation}
Therefore, to prove the relative error statement, it suffices to show that our approximation $\widehat{p}$ satisfies 
\begin{equation}\label{eq:algorithmic-lclt-probability_additive-goal}
|\widehat p-\Pr_{\mu_{\lambda,t}}\{B(I)=0\}|
\le
\frac{\epsilon c_A}{\sqrt n}.
\end{equation}
Set
\[
T:=C\sqrt{\log(2/\epsilon)}
\qquad\text{and}\qquad
\delta:=\frac{\pi\epsilon c_A}{2}.
\]
By the definition of \(\mu_{\lambda,t}\), we can also express $g$ as 
$
g(\theta)=\E_{\mu_{\lambda,t}}e^{i\theta B(I)/\sqrt n}.
$
Lemma~\ref{lem:high_fourier_control}, in the form
\eqref{eq:high_fourier_unstandardized}, gives
\(|g(\theta)|\le e^{-c\theta^2}\) for every
\(\theta\in[-\pi\sqrt n,\pi\sqrt n]\). Choosing
\(C=C(\Delta,\lambda,\gamma,A)\) sufficiently large therefore gives
\begin{equation}\label{eq:algorithmic-lclt-tail-bound}
\int_{T<|\theta|\le\pi\sqrt n}|g(\theta)|\,d\theta\le\delta.
\end{equation}

Fix \(0<\rho'<\rho\), where \(\rho\) is the zero-free radius from
Proposition~\ref{prop:zero_free_tilted_partition}. Choose fixed constants
\(c_0=c_0(\Delta,\lambda,\gamma,A)>0\) and
\(n_0=n_0(\Delta,\lambda,\gamma,A)\in\mathbb N\) such that
\[
C\sqrt{c_0+\frac{\log 2}{n_0}}\le\rho'.
\]
If \(n\le n_0\), direct enumeration examines at most \(2^{n_0}\) subsets and
therefore runs in polynomial time. If \(\epsilon<e^{-c_0n}\), direct
enumeration also runs in time polynomial in \(1/\epsilon\), since
$
2^n\le (1/\epsilon)^{(\log 2)/c_0}.
$
We may therefore assume that \(n>n_0\) and
\(\epsilon\ge e^{-c_0n}\). For \(|\theta|\le T\), these assumptions give
\[
|e^{i\theta/\sqrt n}-1|
\le\frac{|\theta|}{\sqrt n}
\le\frac{T}{\sqrt n}
\le C\sqrt{c_0+\frac{\log 2}{n}}
\le\rho'
\]
Thus Lemma~\ref{lem:low_frequency_complex_partition_oracle} applies throughout
\([-T,T]\). To replace the integral over this interval by a finite sum, note
that
\begin{equation}\label{eq:algorithmic-lclt-lipschitz-bound}
|g(\theta)|\le1,
\qquad\text{and}\qquad
|g'(\theta)|
\le\frac{\E_{\mu_{\lambda,t}}|B(I)|}{\sqrt n}
\le\sqrt n.
\end{equation}

Let
\[
M:=\left\lceil\frac{8T^2\sqrt n}{\delta}\right\rceil,
\qquad
h:=\frac{2T}{M},
\qquad
q:=\min\left\{\frac12,\frac{\delta}{16T}\right\},
\]
and partition \([-T,T]\) into the intervals
\([\theta_j,\theta_j+h]\), where \(\theta_j:=-T+jh\) for
\(j=0,\dots,M-1\). At each left endpoint \(\theta_j\), use
Lemma~\ref{lem:low_frequency_complex_partition_oracle} with accuracy \(q\) to
obtain \(\widehat g_j\), and define
\[
\widehat I:=h\sum_{j=0}^{M-1}\widehat g_j.
\]
By Definition~\ref{def:complex_relative_approximation},
\(\widehat g_j=r_je^{i\varphi_j}g(\theta_j)\), where
\(e^{-q}\le r_j\le e^q\) and \(|\varphi_j|\le q\). Since \(q\le1/2\) and
\(|g(\theta_j)|\le1\), this gives
\(|\widehat g_j-g(\theta_j)|\le4q\). Hence
\begin{align}
\left|\widehat I-\int_{-T}^Tg(\theta)\,d\theta\right|
&\le
\sum_{j=0}^{M-1}\int_{\theta_j}^{\theta_j+h}
|g(\theta)-g(\theta_j)|\,d\theta
+h\sum_{j=0}^{M-1}|\widehat g_j-g(\theta_j)| \nonumber\\
&\le T\sqrt n\,h+8Tq 
\le \frac{\delta}{4}+\frac{\delta}{2}
<\delta.
\label{eq:algorithmic-lclt-finite-sum-error}
\end{align}
We take as our final output
\[
    \widehat p:=\operatorname{Re}\left(\frac{\widehat I}{2\pi\sqrt n}\right).
\]
Then by \eqref{eq:algorithmic-lclt-tail-bound} and
\eqref{eq:algorithmic-lclt-finite-sum-error},
\(
|\widehat p-\Pr_{\mu_{\lambda,t}}\{B(I)=0\}|
\le 2\delta/(2\pi\sqrt n)
=\epsilon c_A/(2\sqrt n),
\)
which proves \eqref{eq:algorithmic-lclt-probability_additive-goal}.
Finally,
\[
M=O\!\left(\frac{T^2\sqrt n}{\epsilon}\right),
\qquad
\frac1q=O\!\left(1+\frac{T}{\epsilon}\right),
\]
where the constants may depend on \(\Delta,\lambda,\gamma,A\). Thus the number of
evaluations and the accuracy requested from
Lemma~\ref{lem:low_frequency_complex_partition_oracle} are polynomial in
\(n\) and \(1/\epsilon\), so the algorithm runs in polynomial time.
\end{proof}

\subsection{Assembling the FPTAS}

\begin{proof}[Proof of Proposition~\ref{prop:compact_tilt_fptas_balanced_partition}]
Fix \(\epsilon\in(0,1)\).
By Theorem~\ref{thm:bounded_ratio_centered_tilt_acceptance}, the centering tilt
lies in a compact interval \([-t_0,t_0]\). Use
Lemma~\ref{lem:deterministic_compact_tilt_centering} to find
\(\widetilde t\) with \(|m(\widetilde t)|\le 1\). Then use
Lemma~\ref{lem:real_compact_tilt_partition_oracle} and
Proposition~\ref{prop:algorithmic_lclt_compact_tilts}, each with accuracy
\(\epsilon/3\), to obtain
\(\widehat Z=(1+\eta_Z)Z_G(\lambda;\widetilde t)\) and
\(\widehat p=(1+\eta_p)\Pr_{\mu_{\lambda,\widetilde t}}\{B(I)=0\}\), where
\(|\eta_Z|,|\eta_p|\le\epsilon/3\). By
\eqref{eq:balanced_counting_identity}, we output \(\widehat Z\widehat p\); its
relative error is at most
\(|\eta_Z|+|\eta_p|+|\eta_Z\eta_p|
\le2\epsilon/3+\epsilon^2/9<\epsilon\).
\end{proof}

\section{\texorpdfstring{Hardness when \(\lambda > \lambda_c\)}{lambda above uniqueness}} 
\label{sec:hardness}

We prove the hardness part of Theorem~\ref{thm:comp_thresh} by adapting the phase-coexistence gadget framework introduced by Sly~\cite{Sly10} for the hard-core model above the tree uniqueness threshold and subsequently used in several hardness reductions for two-spin systems~\cite{sly2012computational,GalanisSV15,galanis2012improvedinapproximabilityresultscounting,carlson2022computational}. The basic building block is a random bipartite $\Delta$-regular \emph{gadget}, modified by deleting a small number of matching edges and attaching finite $(\Delta-1)$-ary trees in order to create terminal vertices. In the non-uniqueness regime $\lambda>\lambda_c(\Delta)$, the hard-core model on such a gadget has two dominant phases: in the (+)-phase the left side of the core is more heavily occupied, while in the ($-$)-phase the right side is more heavily occupied. Conditioned on either phase, the terminal occupations are approximately independent with phase-dependent marginals. 

In our proof, the role of the global balance constraint is similar to that of the fixed-magnetization in the Ising model in \cite{carlson2022computational}. 
Given an input graph $H$, we replace each vertex $x\in V(H)$ by a copy $G_x$ of the gadget. 
A phase vector $Y\in \{+,-\}^{V(H)}$ therefore assigns a phase to each gadget copy, and hence plays the role of a spin configuration on $H$. 
We then join terminal vertices of different gadget copies according to the edges of $H$, and also add isolated vertices to provide an independent source of fluctuation for the global balance variable. The first part of the proof analyzes the disjoint union $\widehat H_s^G$, before the inter-gadget edges are added. 
In this decoupled graph, the gadget copies are independent, and we prove 
concentration and point probability estimates for the global imbalance 
$B(I) = |I \cap L(\widehat H_s^G)| - |I \cap R(\widehat H_s^G)|.$ 
These estimates show that, after conditioning on exact balance $B(I) = 0$, phase vectors with unequal numbers of $(+)$ and $(-)$ gadgets are exponentially suppressed, while phase vectors with equal numbers of $(+)$ and $(-)$ gadgets contribute at scale at least $e^{-O(h\log n)}(nh)^{-1/2}$. This weaker lower bound is still sufficient because it creates only an $O(h\log n)$ error after taking logarithms.

The second part of the proof reintroduces the inter-gadget edges. Since the terminal spins are approximately independent under the phase-conditioned measures, the effect of the inter-gadget edges can be computed explicitly in terms of the phase vector $Y$. With our choice of terminal matchings, an edge of $H$ whose endpoints have equal phases receives a larger compatibility weight than an edge whose endpoints have opposite phases. Consequently, among the balanced phase vectors, the dominant contribution to the balanced partition function comes from those minimizing the number of cut edges in $H$. Thus the balanced hard-core partition function on $H_s^G$ encodes the minimum bisection value of $H$. Comparing $Z^{\mathrm{bal}}_{H_s^G}$ with the partition function $Z_{\widehat H_s^G}$ then allows us to recover the minimum bisection value from sufficiently accurate multiplicative approximations, giving the desired hardness reduction.

The remaining task is to estimate the partition function of the decoupled graph. We do this without using existing algorithms for random bipartite graphs, whose available guarantees in the present regime require $\Delta$ sufficiently large. Here, instead, the expected partition function of a single gadget can be computed explicitly, and small subgraph conditioning shows that a typical gadget differs from this expectation by at most a polynomial factor. This loss is small enough for the reduction.

We will use the following standard hardness fact.

\begin{fact}[\cite{garey1974some}]\label{fact:hardness_MINBisection}
The exact \textbf{MIN-BISECTION} problem is $\textbf{NP-hard}$: given a graph $H$ on an even number of vertices, it is $\textbf{NP-hard}$ to compute the minimum number of edges crossing a bisection of $V(H)$.
\end{fact}

The main technical result of this section is the following reduction. Theorem~\ref{thm:comp_thresh}, in the non-uniqueness regime, follows immediately from Theorem~\ref{thm:N_eta_approx_reduction} and Fact~\ref{fact:hardness_MINBisection}.

\begin{theorem}\label{thm:N_eta_approx_reduction}
Fix $\Delta \ge 3$ and $\lambda > \lambda_c(\Delta)$. There exists $\zeta > 0$ such that the following holds. If there is a randomized polynomial-time algorithm which, on every $N$-vertex input graph $F \in \mathcal{G}_{\Delta,\Delta}$, gives an $e^{N^\zeta}$-factor approximation to $Z^{\text{bal}}_F(\lambda)$, meaning that its output $\widehat Z$ satisfies
\[
    e^{-N^\zeta}Z^{\mathrm{bal}}_F(\lambda)
    \le \widehat Z\le
    e^{N^\zeta}Z^{\mathrm{bal}}_F(\lambda)
\]
with probability at least $2/3$, then \textbf{MIN-BISECTION} can be solved in randomized polynomial time. Moreover, if there is an efficient sampling scheme for $\mu^{\text{bal}}_{F,\lambda}$ on all $F \in \mathcal{G}_{\Delta,\Delta}$, then \textbf{MIN-BISECTION} can be solved in randomized polynomial time.
\end{theorem}

\subsection{The gadget graph and its properties}

We begin with the gadget construction used throughout the reduction. The gadget $G=G(\Delta,n,\theta,\psi)$, where $\theta,\psi$ are constants chosen later, is the random bipartite hard-core gadget introduced in~\cite{Sly10} and analyzed in its growing-terminal form in~\cite{galanis2016inapproximability}, with an asymptotically equivalent rounding of its size parameters; see also~\cite{carlson2022computational}. It is a balanced bipartite graph on $(2+o(1))n$ vertices, with all but $o(n)$ vertices of degree $\Delta$, and with $m=O(n^\theta)$ terminal vertices of degree $\Delta-1$ on each side. We now describe the construction and then record the one-gadget properties needed later in the reduction.

Fix $\Delta \ge 3$ and a fugacity $\lambda > \lambda_c(\Delta)$ in the non-uniqueness region of the
hard-core model on the infinite $\Delta$-regular tree. Fix constants
\[
\theta,\psi \in (0,1/8),
\qquad
\theta+\psi<1/4.
\]
For each $n$ define
\[
    d_n := 2\bigl\lfloor \tfrac{\psi}{2} \log_{\Delta-1} n \bigr\rfloor,
\qquad
L_n := (\Delta-1)^{d_n},
\qquad
m := \lfloor n^\theta \rfloor,
\qquad
m' := mL_n.
\]
Then $m'=\BigO(n^{\theta+\psi})=o(n)$. To construct the gadget $G=G(\Delta,n,\theta,\psi)$, let $G' = G'(\Delta,n,\theta,\psi)$
be a random bipartite graph with $n+m'$ vertices on each side, obtained by choosing $\Delta$
perfect matchings between the two sides uniformly at random and then deleting a uniformly random set of $m'$ edges from the
last matching. The probabilistic analysis uses this unconditioned matching
multigraph. Before it is used as an algorithmic instance, parallel copies of
an edge are suppressed. This preserves every independent set and its weight
and can only decrease degrees; the labels below refer to multigraph degrees.

Let $U$ be the set of vertices of degree $\Delta$ in $G'$, and let $W$ be the set of vertices
of degree $\Delta-1$ in $G'$. Then $|U|=2n,$ and $|W|=2m'$ and we can write the two sides of the bipartition of $G'$ as $ U^L \sqcup W^L$ and $U^R \sqcup W^R$,
so that $|U^L|=|U^R|=n$ and $|W^L|=|W^R|=m'$.

To form $G$ from $G'$, on each side partition the $m'$ vertices of degree $\Delta-1$ into $m$
disjoint groups of size $L_n$, and attach to each group the leaves of a $(\Delta-1)$-ary tree of depth
$d_n$. The roots of these added trees are the only vertices of degree $\Delta-1$ in the final
graph $G$; denote their union by $T$. Thus $T$ is the terminal set through which different
copies of the gadget will later be joined, and write $T=T^L\sqcup T^R$
according to the bipartition. (The even choice of $d_n$ places each root on
the same side as its leaves.) By construction, $|T|=2m=\BigO\left(n^\theta\right),$
and the total number of vertices of $G$ outside the $2n+2m'$ vertices of $G'$ is $\BigO(mL_n)=\BigO(n^{\theta+\psi})=o(n).$

We next fix notation for a single canonical gadget. When copies of the gadget are used later, the same notation will be decorated by the copy index. 
For a configuration $\sigma \in \{0,1\}^{V(G)}$,  set
\[
    X_L = \sum_{v \in U^L} \sigma_v, \quad
    X_R = \sum_{v \in U^R} \sigma_v,
\]
define the single-gadget imbalance on $G$ by
\begin{equation}\label{eq:def_gadget_balance}
     B_{\mathrm{gad}}(\sigma):= X_{L}(\sigma) - X_{R}(\sigma).
\end{equation}
 Let $\Sigma = \{+,-\}$ be the set of phases. Independently of the
configuration, let $C_G$ be a uniform random sign in $\Sigma$. We define the
phase $Y_G(\sigma,C_G)\in\Sigma$ by
\begin{equation} \label{eq:phase_of_gadget}
    Y_{G}(\sigma,C_G)=
    \begin{cases}
        +,& B_{\mathrm{gad}}(\sigma)>0, \quad \text{or}\quad B_{\mathrm{gad}}(\sigma)=0\ \text{and }C_G=+, \\
        -,& B_{\mathrm{gad}}(\sigma)<0, \quad \text{or}\quad B_{\mathrm{gad}}(\sigma)=0\ \text{and }C_G=-.
    \end{cases}
\end{equation}
Thus tied configurations are split equally between the two phases. Equivalently,
the phase-restricted partition functions are
\[
    Z_{G,\varsigma}:=\frac12\sum_{c\in\Sigma}
    \sum_{I\in\cI_G:\,Y_G(I,c)=\varsigma}\lambda^{|I|},
    \qquad \varsigma\in\Sigma,
\]
so $Z_G=Z_{G,+}+Z_{G,-}$. In every phase-conditioned statement, we augment
$\mu_G$ by the independent fair sign $C_G$ and then condition on $Y_G$; the
marginal on independent sets remains the ordinary hard-core measure. We write
the induced measures as $\mu_{G,+}$ and $\mu_{G,-}$.

Finally, for $\alpha\in\{0,\frac1n,\dots,1\}$ and
$\varsigma\in\Sigma$, define the phase-restricted slice partition function by
\[
    Z_{G,\alpha,\varsigma}:=\frac12\sum_{c\in\Sigma}
    \sum_{\substack{I\in\cI_G:\,X_L(I)=\alpha n,\\
    Y_G(I,c)=\varsigma}}\lambda^{|I|}.
\]

Let \(\lambda>\lambda_c(\Delta)\). Let \(\mu_+\) and \(\mu_-\) denote the two semi-translation-invariant Gibbs measures of the hard-core model on the infinite \(\Delta\)-regular tree \(T_\Delta\), obtained as weak limits with occupied boundary conditions on even and odd levels, respectively. Then $\alpha_{\pm}$ denote the marginal occupation probability of the root under $\mu_{\pm}$.
Conditioned on the (+)-phase, the empirical occupation densities on the two sides of the gadget concentrate around \((\alpha_+,\alpha_-)\); conditioned on the $(-)$-phase, they concentrate around \((\alpha_-,\alpha_+)\).

The occupation marginals of the terminal vertices are described by a different pair of parameters. Let \(q_+\) and \(q_-\) denote the occupation
probabilities at the root of the rooted \((\Delta-1)\)-ary tree as the same weak limit measure above. They satisfy
\[
    q_+ = \frac{\alpha_+}{1-\alpha_-},
\qquad
    q_- = \frac{\alpha_-}{1-\alpha_+}.
\]
We define the product measures $Q^{+}_{T}$ (respectively $Q^-_{T}$) on configurations $\sigma_{T}$ on $T$ so that the spins are i.i.d.\ Bernoulli with probability $q_+$ (resp.\ $q_-$) on $T^L$ and $q_-$ (resp.\ $q_+$) on $T^R$. 
That is,
\[
    Q_T^\pm(\sigma_T)
    :=
    (q_\pm)^{\sum_{v\in T^L}\sigma_v}
    (1-q_\pm)^{m-\sum_{v\in T^L}\sigma_v}
    \cdot
    (q_\mp)^{\sum_{v\in T^R}\sigma_v}
    (1-q_\mp)^{m-\sum_{v\in T^R}\sigma_v}.
\]

The following lemma collects previous results about the quantities defined above on the gadget, most notably the near independence of the terminal spins conditioned on a phase.

\begin{lemma}\label{lem:all_prelim}
Let $G$ be the random graph described above with parameters $\theta, \psi \in (0, 1/8)$.
There exist constants $c,C>0$ such that the following holds for all sufficiently large $n$.
\[
\frac{\mathbb{E}Z_{G,\alpha,+}}{\mathbb{E}Z_{G,+}}
\le
\frac{C}{\sqrt n}\exp\!\bigl(-c n(\alpha-\alpha_+)^2\bigr),
\qquad
\alpha\in\Bigl\{0,\frac1n,\dots,1\Bigr\},
\tag{1}
\]
and
\[
\frac{\mathbb{E}Z_{G,\alpha,-}}{\mathbb{E}Z_{G,-}}
\le
\frac{C}{\sqrt n}\exp\!\bigl(-c n(\alpha-\alpha_-)^2\bigr),
\qquad
\alpha\in\Bigl\{0,\frac1n,\dots,1\Bigr\}.
\tag{2}
\]

 Moreover, there exist choices of constants $\theta, \psi \in (0, 1/8)$ and $C'>0$ such that all of the following hold simultaneously for both signs $\varsigma\in\{+,-\}$ with probability at least $99/100$ over the choice of $G$:

\begin{enumerate}[label=\textup{(\roman*)}]
    \item For every $\tau\in\{0,1\}^{T}$,
    \[
        \left|
        \frac{\mu_{G,\varsigma}(\sigma_{T}=\tau)}{Q^\varsigma_{T}(\tau)}-1
        \right|
        \le n^{-2\theta}.
        \tag{3}
    \]

    \item There exists a set $B_\varsigma\subseteq \{0,1\}^{W}$ such that
    \[
        \mu_{G,\varsigma}(\sigma_{W}\in B_\varsigma)\le \exp(-n^{2\theta}),
        \tag{4}
    \]
    and for every $\eta\in \{0,1\}^{W}\setminus B_\varsigma$ and every $\tau\in\{0,1\}^{T}$,
    \[
        \left|
        \frac{\mu_{G,\varsigma}(\sigma_{T}=\tau\mid \sigma_{W}=\eta)}{Q^\varsigma_{T}(\tau)}-1
        \right|
        \le n^{-3\theta}.
        \tag{5}
\]

    \item We have the bound  \[
        n^{-C'}\,\mathbb{E}Z_{G,\varsigma}
        \le
        Z_{G,\varsigma}
        \le
        n^{C'}\,\mathbb{E}Z_{G,\varsigma}.
        \tag{6}
        \]
\end{enumerate}
Moreover,
\[
    \mathbb{E}Z_{G,+}=\mathbb{E}Z_{G,-}.
    \tag{7}
\]
Consequently, on the event in \textup{(6)},
\[
    n^{-2C'} \le \frac{Z_{G,+}}{Z_{G,-}}
    \le n^{2C'}.
    \tag{8}
\]
\end{lemma}

\begin{proof}
The phase restricted Gaussian slice estimates \textup{(1)}--\textup{(2)}
follow from the strict maximality and Laplace calculation in
\cite[Proof of Lemma B.3]{GalanisSV15}.  The first moment
formula with growing boundary in \cite[Lemma 20 and Section 7.2.1]{galanis2016inapproximability}
adds only $O(m'|\alpha-\alpha_\varsigma|+(m')^2/n)$ to the exponent, which is
absorbed by completing the square since $m'=o(n^{1/4})$.  Summing the resulting
uniform boundary conditioned estimate against the nonnegative tree weights
gives \textup{(1)}--\textup{(2)}.

Parts \textup{(i)}--\textup{(iii)} are the growing-tree gadget conclusions
proved in \cite[Lemmas 19 and 23]{galanis2016inapproximability};
in particular, that proof already transfers the small-subgraph estimate to
the tree-augmented gadget by summing against the appended tree weights.  Its
displayed lower and upper bounds are polynomial in $n$, which gives
\textup{(6)} after enlarging $C'$.  Applicability throughout the hard-core
non-uniqueness region is supplied by
\cite[Theorem 1.4, Lemma 3.2, and the discussion following Theorem 1.5]{GalanisSV15}.
The cited proof uses a fixed tie convention.  The strict gap between the two
dominant phases and the diagonal sector makes the tie contribution
$e^{-\Omega(n)}$ relative to either phase, so splitting ties equally does not
affect these conclusions.

For \textup{(7)}, let $\widetilde G$ be the reflected version of the random
gadget obtained by swapping the left and right sides. Couple the auxiliary
signs by $C_{\widetilde G}=\overline{C_G}$. Reflection negates
$B_{\mathrm{gad}}$, and the coupled sign swaps the phase on a tie. Thus
$\widetilde G$ has the same law as $G$, while the weight-preserving bijection
$(I,C_G)\mapsto(\widetilde I,\overline{C_G})$ gives
\[
Z_{\widetilde G,+}=Z_{G,-}.
\]
Taking expectations gives
\[
\mathbb{E}Z_{G,+}=\mathbb{E}Z_{\widetilde G,+}=\mathbb{E}Z_{G,-}.
\]
Finally, increase $n_0$ if necessary to obtain overall probability at least $99/100$. The ratio bound \textup{(8)} follows immediately from \textup{(6)} and \textup{(7)}.
\end{proof}

Now, let $\mu_{G,\varsigma}$ denote the hard-core measure on $G$ conditioned on $\set{Y_G=\varsigma}$ and let $\langle\cdot\rangle_{G,\varsigma}$ be the corresponding expectation.
Similarly let $\mu_{G,\varsigma,\tau}$ and $\langle\cdot\rangle_{G,\varsigma,\tau}$ denote the hard-core measure and the expectation on $G$ conditioned on $\set{Y_G=\varsigma,\ \sigma_{T}=\tau}.$

Our next step is to upgrade these one-gadget structural facts to quantitative probabilistic
control that is uniform in the terminal condition \(\tau\). In particular, we will need
bounds on moments and exponential moments of the occupation variable under
\(\mu_{G,\varsigma,\tau}\). 

\begin{lemma}\label{lem:moment_mgf_bounds_X}
    There exists a constant $C_0<\infty$ with the following property. If $t=t(n)>0$ satisfies
    \begin{equation}
        |T|=o(t^2 n),
        \label{eq:t-condition}
    \end{equation}
    Then for all sufficiently large $n$, with probability at least $19/20$ over the choice of the random gadget $G$, the following statements hold simultaneously for every
    $\tau\in\{0,1\}^{T}$ and both signs $\varsigma \in \{+,-\}$.

    Then,  writing $X = X_L$
    \begin{align}
        \left\langle \abs{X-\alpha_\varsigma n}\right\rangle_{G,\varsigma,\tau}
        &=\BigO(\sqrt{n\log n}), \label{eq:lem12-first}\\
        \left\langle \abs{X-\alpha_\varsigma n}^{2}\right\rangle_{G,\varsigma,\tau}
        &=\BigO(n\log n), \label{eq:lem12-second}\\
        \left\langle \abs{X-\alpha_\varsigma n}^{3}\right\rangle_{G,\varsigma,\tau}
        &=\BigO((n\log n)^{3/2}), \label{eq:lem12-third}
    \end{align}
    and with the same constant $C_0$
    \begin{align}
        \left\langle e^{t(X-\alpha_\varsigma n)}\right\rangle_{G,\varsigma,\tau}
        &\le \exp(C_0 t^2 n), \label{eq:lem12-mgf-plus}\\
        \left\langle e^{t(\alpha_\varsigma n-X)}\right\rangle_{G,\varsigma,\tau}
        &\le \exp(C_0 t^2 n). \label{eq:lem12-mgf-minus}
    \end{align}
\end{lemma}

\begin{proof}
Fix $t=t(n)>0$ satisfying \eqref{eq:t-condition}. We first prove the bounds under
$\mu_{G,\varsigma}$ and then transfer them to $\mu_{G,\varsigma,\tau}$. 
By Markov's inequality and a union bound over both signs and all $n+1$
density slices, with probability $1-o(1)$,
\[
    Z_{G,\alpha,\varsigma}
    \le n^3\mathbb E Z_{G,\alpha,\varsigma}
    \qquad\text{for every $\alpha$ and $\varsigma$.}
\]
Intersecting this event with the event from Lemma~\ref{lem:all_prelim} has
probability greater than $19/20$ for all sufficiently large $n$. Fix a graph
$G$ in this event.
By \textup{(1)}, \textup{(2)}, and \textup{(6)}, there are constants
$K,C,c>0$ such that, for either sign and
$\ell=\alpha n-\alpha_\varsigma n$,
\[
    \mu_{G,\varsigma}(X=\alpha n)
    \le Cn^K e^{-c\ell^2/n}.
\]
Splitting at $|\ell|=A\sqrt{n\log n}$, with $A$ sufficiently large, and
summing the Gaussian tail gives
\[
    \left\langle |X-\alpha_\varsigma n|^k\right\rangle_{G,\varsigma}
    =O((n\log n)^{k/2}),
    \qquad k=1,2,3.
\]
Completing the square in the same bound gives
\begin{align*}
\left\langle e^{t(X-\alpha_\varsigma n)}\right\rangle_{G,\varsigma}
&\le Cn^K\sum_{\ell\in\mathbb Z}e^{t\ell-c\ell^2/n}\\
&\le Cn^{K+1/2}\exp\!\left(\frac{t^2n}{4c}\right)
\le \exp\!\bigl(\BigO(t^2n)\bigr),
\end{align*}
where the last inequality uses $\log n=o(t^2n)$, which follows from
$|T|=\Theta(n^\theta)=o(t^2n)$.
Thus
\[
\left\langle e^{t(X-\alpha_\varsigma n)}\right\rangle_{G,\varsigma}
\le \exp(C_1 t^2n)
\]
for some constant $C_1<\infty$ independent of $n$, $t$, and $\varsigma$. The same argument with
$\xi(\ell)=e^{-t\ell}$ gives
\[
\left\langle e^{t(\alpha_\varsigma n-X)}\right\rangle_{G,\varsigma}
\le \exp(C_1 t^2n).
\]

We will now transfer to terminal conditioning for the exponential moments. Fix $\varsigma\in\{+,-\}$ and $\tau\in\{0,1\}^{T}$. By Lemma
\ref{lem:all_prelim},
\[
    \mu_{G,\varsigma}(\sigma_{T}=\tau)
=
    Q^\varsigma_{T}(\tau)\bigl(1+\BigO(n^{-2\theta})\bigr).
\]
Since $Q^\varsigma_{T}$ is a product measure with one-site marginals bounded away from $0$ and $1$, there exists $b<\infty$ such that for all $\varsigma,\tau$, $ Q^\varsigma_{T}(\tau)\ge e^{-b|T|} $.
Hence for all sufficiently large $n$, $
\mu_{G,\varsigma}(\sigma_{T}=\tau)^{-1}\le 2e^{b|T|}.$
Therefore
\begin{align*}
    \left\langle e^{t(X-\alpha_\varsigma n)}\right\rangle_{G,\varsigma,\tau}
    &=
    \frac{\left\langle e^{t(X-\alpha_\varsigma n)}\1_{\{\sigma_{T}=\tau\}}\right\rangle_{G,\varsigma}}{\mu_{G,\varsigma}(\sigma_{T}=\tau)}
    \le
    2e^{b|T|} \left\langle e^{t(X-\alpha_\varsigma n)}\right\rangle_{G,\varsigma}\\
    &\le
    \exp\bigl(C_1 t^2n+b|T|+O(1)\bigr) 
    \le
    \exp((C_1+1)t^2n).
\end{align*}

The last inequality holds since $|T|=o(t^2n)$ and this implies that \eqref{eq:lem12-mgf-plus}--\eqref{eq:lem12-mgf-minus} hold with
$C_0:=C_1+1.$

Lastly, we transfer to terminal conditioning for the polynomial moments. Fix $k\in\{1,2,3\}$ and write $B_\varsigma$ for the bad set from Lemma~\ref{lem:all_prelim}\textnormal{(ii)}.
Decompose
\[
    \left\langle
    \abs{X-\alpha_\varsigma n}^k \1_{\{\sigma_{T}=\tau\}}
    \right\rangle_{G,\varsigma}
=
    A_{\mathrm{good}}+A_{\mathrm{bad}},
\]
according to whether $\sigma_{W}\notin B_\varsigma$ ($A_{\mathrm{good}}$) or $\sigma_{W}\in B_\varsigma$ ($A_{\mathrm{bad}}$).

Since $0\le X\le n$,
\[
    A_{\mathrm{bad}}
    \le
    n^k\,\mu_{G,\varsigma}(\sigma_{W}\in B_\varsigma)
    \le
    n^k e^{-n^{2\theta}}.
\]
After division by $\mu_{G,\varsigma}(\sigma_{T}=\tau)\ge \tfrac12 e^{-b|T|}$ this remains negligible.

For the $A_{\mathrm{good}}$, we condition on $\sigma_{W}=\eta$ with $\eta\notin B_\varsigma$. Once
$\sigma_{W}$ is fixed, the hard-core Gibbs weight factorizes between the degree-$\Delta$ core
and the attached trees, because these subgraphs intersect only through $W$ and the spins on
$W$ have been frozen. The random variable $X$ and the event $\{Y=\varsigma\}$ are measurable
with respect to the core configuration and $C_G$, whereas the event $\{\sigma_{T}=\tau\}$ is measurable
with respect to the tree configurations. Hence under the conditioning $ \{Y=\varsigma,\sigma_{W}=\eta\}$, the variable $X$ is independent of $\{\sigma_{T}=\tau\}$. 
Therefore
\begin{align*}
    A_{\mathrm{good}}
    &=
    \sum_{\eta\notin B_\varsigma}
    \mu_{G,\varsigma}(\sigma_{W}=\eta)\,
    \mu_{G,\varsigma}(\sigma_{T}=\tau\mid \sigma_{W}=\eta)\,
    \E\!\left[
    \abs{X-\alpha_\varsigma n}^{k}
    \,\middle|\,
    Y=\varsigma,\ \sigma_{W}=\eta
    \right] \nonumber \\
    &= 
    \sum_{\eta\notin B_\varsigma}
    \mu_{G,\varsigma}(\sigma_{W}=\eta)\,
    Q^\varsigma_{T}(\tau)\bigl(1+\BigO(n^{-3\theta})\bigr)\,
    \E\!\left[
    \abs{X-\alpha_\varsigma n}^{k}
    \,\middle|\,
    Y=\varsigma,\ \sigma_{W}=\eta
    \right] \nonumber \\
    &= 
    Q^\varsigma_{T}(\tau)\bigl(1+\BigO(n^{-3\theta})\bigr)
    \left\langle
    \abs{X-\alpha_\varsigma n}^{k}
    \1_{\{\sigma_{W}\notin B_\varsigma\}}
    \right\rangle_{G,\varsigma}.
\end{align*}
where in the second equality, we used Lemma \ref{lem:all_prelim}(ii).
Dividing by
\[
    \mu_{G,\varsigma}(\sigma_{T}=\tau)
    =
    Q^\varsigma_{T}(\tau)\bigl(1+\BigO(n^{-2\theta})\bigr)
\]
and using the unconditional bounds gives
\[
\left\langle
\abs{X-\alpha_\varsigma n}^{k}
\right\rangle_{G,\varsigma,\tau}
=
\BigO((n\log n)^{k/2}),
\qquad k=1,2,3.
\]
This proves \eqref{eq:lem12-first}--\eqref{eq:lem12-third}.
\end{proof}

We establish the following corollary which says that the imbalance parameter $B_{\mathrm{gad}}(I)$ as moment and exponential moment bounds analogous to those for $X$. 
This will be useful in the point-probability estimates for the global imbalance later.

\begin{corollary}\label{cor:moment_mgf_bounds_B}

    There exists a constant $C_0<\infty$ with the following property. If $t=t(n)>0$ satisfies
    $
        |T|=o(t^2 n),
    $
    then for all sufficiently large $n$, with probability at least $9/10$ over the choice of the random gadget $G$, the following statements hold simultaneously for every $\tau\in\{0,1\}^{T}$ and both signs $\varsigma \in \Sigma$.
    \begin{align}
        \left\langle \abs{B_{\mathrm{gad}} - (\alpha_\varsigma - \alpha_{\Bar{\varsigma}}) n} \right\rangle_{G,\varsigma,\tau}
        &= 
        \BigO(\sqrt{n\log n}), \label{eq:cor3.4-first} \\
        \left\langle \abs{B_{\mathrm{gad}} - (\alpha_\varsigma -  \alpha_{\Bar{\varsigma}}) n}^{2} \right\rangle_{G,\varsigma,\tau}
        &= 
        \BigO(n\log n), \label{eq:cor3.4-second} 
        \\
        \left\langle \abs{B_{\mathrm{gad}} - (\alpha_\varsigma -  \alpha_{\Bar{\varsigma}}) n}^{3} \right\rangle_{G,\varsigma,\tau}
        &= 
        \BigO((n\log n)^{3/2}), \label{eq:cor3.4-third}
\end{align}
and
\begin{align}
    \left\langle e^{t(B_{\mathrm{gad}} - (\alpha_\varsigma -  \alpha_{\Bar{\varsigma}}) n)} \right\rangle_{G,\varsigma,\tau}
    &\le \exp(C_0 t^2 n), \label{eq:cor3.4-mgf-plus} 
    \\
    \left\langle e^{t((\alpha_\varsigma -  \alpha_{\Bar{\varsigma}}) n - B_{\mathrm{gad}})} \right\rangle_{G,\varsigma,\tau}
    &\le \exp(C_0 t^2 n). \label{eq:cor3.4-mgf-minus}
\end{align}
\end{corollary}

\begin{proof}
Recall that $B_{\mathrm{gad}} = X_L - X_R$ where $X_L:=|I\cap U^L|$ and $X_R:=|I\cap U^R|$.
By Lemma~\ref{lem:moment_mgf_bounds_X}, with probability at least $19/20$ over the choice of $G$, (\ref{eq:lem12-first})-(\ref{eq:lem12-mgf-minus}) hold for $X_L$ and $\alpha_{\varsigma}n$.

To pass from the corresponding bounds for $X_L, \alpha_{\varsigma}n$ to those for $X_R, \alpha_{\Bar{\varsigma}}n$, we use the left--right symmetry of the gadget. Indeed, interchanging the two bipartition classes sends the random gadget to one with the same distribution, while exchanging the roles of the two phases (including flipping the auxiliary sign on a tie) and replacing the left-core count by the right-core count. Accordingly, the statement of the previous lemma for the left side immediately yields the analogous statement for the right side, with the center $\alpha_\varsigma n$ replaced by $\alpha_{\Bar{\varsigma}}n$. We may therefore assume that the same moment and exponential moment bounds hold simultaneously for both $X_L$ and $X_R$, uniformly over all $\tau\in\{0,1\}^T$ and $\varsigma\in\{+,-\}$ with probability at least 9/10.

Now for $k = 1,2,3$, since for any real numbers $a,b$
\[
    |a+b|^k \le (|a|+|b|)^k \le 2^{k-1}|a|^k+2^{k-1}|b|^k,
\]
it follows that
\[
    \left\langle |B_{\mathrm{gad}}-(\alpha_{\varsigma}-\alpha_{\Bar{\varsigma}})n|^k \right\rangle_{G,\varsigma,\tau}
    \le
    2^{k-1}\left\langle \left|X_L-\alpha_\varsigma n\right|^k \right\rangle_{G,\varsigma,\tau}
    +
    2^{k-1}\left\langle |\alpha_{\Bar{\varsigma}}n-X_R|^k \right\rangle_{G,\varsigma,\tau}
    =
    \BigO((n\log n)^{k/2}).
\]

We now prove the exponential moment bounds. Since $|T|=o(t^2n)$, we also have $|T| = o\!\left((2t)^2n\right)$,
so the previous lemma may be applied with parameter $2t$.
By Cauchy--Schwarz,
\[
    \left\langle e^{t(B_{\mathrm{gad}}-(\alpha_{\varsigma}-\alpha_{\Bar{\varsigma}})n)} \right\rangle_{G,\varsigma,\tau}
    \le
    \left\langle e^{2t(X_L-\alpha_\varsigma n)} \right\rangle_{G,\varsigma,\tau}^{1/2}
    \left\langle e^{2t(\alpha_{\Bar{\varsigma}}n-X_R)} \right\rangle_{G,\varsigma,\tau}^{1/2}
    \le \exp(4Ct^2n),
\]
the other bound follows similarly. 
\end{proof}

\subsection{Characterizing the global phases}

We now give the reduction from \textbf{MIN-BISECTION} to the problem of approximating the balanced
hard-core partition function. Let $H$ be an input graph on an even number $h$ of vertices; the odd case may be reduced to this one by adding one isolated vertex.

We choose the gadget size $n = n(h)$ polynomially large in $h$, with the 
constant in the polynomial chosen so that
$$h = \Theta\left(n^{\theta/4}\right), \quad kh \le m, \quad \frac{h\log n}{k} = o(1),$$
where
$$k := \left\lfloor n^{3\theta/4} \right\rfloor, \quad m := \left\lfloor n^\theta \right\rfloor.$$

Given the gadget graph $G$ as above and an even integer $s = \Theta(nh) \geq 0$, we construct a graph $H^G_s$ of maximum degree $\Delta$ 
on $N := h|G| + s$ vertices as follows:

\begin{itemize}
    \item For each vertex $x \in V(H)$, include a copy $G_x$ of $G$. We write $U^L_x, U^R_x, W^L_x, W^R_x, T^L_x, T^R_x$ for the corresponding subsets of this copy, and $T_x := T^L_x \sqcup T^R_x$.
    
    \item Add $s/2$ isolated vertices to the left side and $s/2$ isolated vertices to the right side.
    
    \item For each edge $xy \in E(H)$, add a matching of size $k$ between $T^L_x$ and $T^R_y$, and another matching of size $k$ between $T^R_x$ and $T^L_y$. The matchings are chosen so that each terminal is used at most once. This is possible because each vertex of $H$ is incident to at most $h$ edges and $kh \le m$.
\end{itemize}

Let $\widehat{H}^G_s$ denote the graph obtained from $H^G_s$ by deleting  all inter-gadget edges; thus $\widehat{H}^G_s$ is the disjoint union of  the $h$ gadget copies and the isolated vertices. Let $\mathcal{E}$ be the set of  inter-gadget edges. The isolated vertices are included to give an explicit  independent variance source for the global imbalance. More precisely, under the hard-core measure on $\widehat{H}^G_s$, even after conditioning on the gadget phases and terminal occupations, the isolated vertices remain independent Bernoulli random variables with parameter $\lambda/(1 + \lambda)$.

Let $L_{\mathrm{iso}}$ and $R_{\mathrm{iso}}$ denote the isolated vertices on the left and right side of $\widehat H_s^G$.
For a configuration $\sigma$ on $\widehat H_s^G$, define the imbalance of those isolated vertices as
\[
    B_{\mathrm{iso}}(\sigma)
    :=
    \sum_{u\in L_{\mathrm{iso}}}\sigma_u-\sum_{v\in R_{\mathrm{iso}}}\sigma_v.
\]
For each gadget copy $G_x$, recall that
\[
    B_x(\sigma)
    :=
    \sum_{v\in U_x^L}\sigma_v
    -
    \sum_{v\in U_x^R}\sigma_v
\]
is its core imbalance. Define the imbalance contributed by the remaining, non-core vertices of all gadget copies by
\[
    B_{\mathrm{rem}}(\sigma)
    :=
    \sum_{x\in V(H)}
    \left(
        \sum_{v\in L(G_x)\setminus U_x^L}\sigma_v
        -
        \sum_{v\in R(G_x)\setminus U_x^R}\sigma_v
    \right).
\]
We reserve $B$ for the true full imbalance of the entire decoupled graph:
\[
    B(\sigma)
    :=
    \sum_{v\in L(\widehat H_s^G)}\sigma_v
    -
    \sum_{v\in R(\widehat H_s^G)}\sigma_v.
\]
Thus
\[
    B
    =
    \sum_{x\in V(H)}B_x
    +B_{\mathrm{iso}}
    +B_{\mathrm{rem}}.
\]

Since each gadget contains $O(n^{\theta+\psi})$ vertices outside $U_x^L\sqcup U_x^R$, there is a constant $C<\infty$ such that, deterministically,
$ |B_{\mathrm{rem}}|\le Chn^{\theta+\psi}=o(n),$
and also $hn^{\theta+\psi}=o(h\sqrt n).$
The first estimate will be used to separate phase vectors with $D(Y)\ne0$, while the second allows the non-core contribution to be absorbed into the $O(h\sqrt{n\log n})$ window used in the balance point-probability lower bound.

For each gadget copy $G_x$, independently draw a fair auxiliary sign $C_x$,
and define the phase vector by
$Y(\sigma,C)=\left(Y_{G_x}(\sigma|_{G_x},C_x)\right)_{x\in V(H)}$.
Equivalently, all phase-restricted partition functions below sum over the
auxiliary signs with weight $2^{-h}$. This augmentation leaves the hard-core
law of $\sigma$ and every unconditioned partition function unchanged.
For $Y \in \{+,-\}^{V(H)}$, define its phase imbalance by
$$D(Y) := \left|\{x \in V(H) : Y_x = +\}\right| - \left|\{x \in V(H) : Y_x = -\}\right|.$$
Thus $D(Y) = 0$ means that the gadget phases form a bisection of $V(H)$. 
We also write
$$\mathrm{cut}(Y) := \left|\{xy \in E(H) : Y_x \neq Y_y\}\right|.$$
Let $T := \bigcup_{x \in V(H)} T_x$. 
For $Y \in \{+,-\}^{V(H)}$ and $\tau \in \{0,1\}^{T}$, let $\mu_{\widehat{H}^G_s, Y, \tau}$ and $\langle \cdot \rangle_{\widehat{H}^G_s, Y, \tau}$ denote the hard-core measure and expectation operator on $\widehat{H}^G_s$ conditioned on the phase vector $Y$ and terminal 
pattern $\tau$. We use analogous notation when only the phase vector is conditioned on.

The first estimate controls the core-plus-isolated imbalance $B-B_{\mathrm{rem}}$. Combined with the deterministic bound $|B_{\mathrm{rem}}|=o(n)$, it will imply that if the phase imbalance is nonzero, then exact global balance $B=0$ is exponentially unlikely.

\begin{lemma} \label{lem:hardcore13}

For $Y \in \{+,-\}^{V(H)}$, set
$$\nu_0(Y) := n \sum_{x \in V(H)} (\alpha_{Y_x} - \alpha_{\bar{Y}_x}) = n(\alpha_+ - \alpha_-)D(Y).$$
For every fixed $\delta > 0$, there are
$t_\delta=\Theta_\delta(1/h)$ and $c=c(\delta)>0$ such that, whenever
$G$ satisfies the conclusions of Corollary~\ref{cor:moment_mgf_bounds_B}
at $t=t_\delta$, the following holds for all sufficiently large $n$,
uniformly over $Y \in \{+,-\}^{V(H)}$ and $\tau \in \{0,1\}^{T}$:
$$\mu_{\widehat{H}^G_s, Y, \tau} \left( |(B-B_{\mathrm{rem}}) - \nu_0(Y)| \ge \delta n \right) \le \exp(-cn/h).$$
\end{lemma}
\begin{proof}
    We prove the upper tail; the proof for the lower tail is identical. We first handle $B_{\mathrm{iso}}$. As discussed above we have $B_{\mathrm{iso}} = J^{L}_{s/2} - J^{R}_{s/2}$, where $J^{L}_{s/2}, J^{R}_{s/2} \sim \mathrm{Bin}\left(\frac{s}{2}, \frac{\lambda}{1+\lambda}\right)$.
    Then for some constant $c_1(\lambda)>0$ we have
    \begin{align*}
        \left\langle e^{tB_{\mathrm{iso}}} \right\rangle_{\widehat{H}^G_s, Y, \tau}
        = 
        \left(\frac{(1+\lambda e^{t})(1+\lambda e^{-t})}{(1+\lambda)^2}\right)^{\frac{s}{2}}
        \leq e^{c_1st^2}
    \end{align*}
    uniformly for all $|t|\le t_0(\lambda)$.
    
    Since $\widehat{H}^G_s$ is a disjoint union after the phase and terminal pattern are fixed, the moment generating function factors over the gadget copies and the isolated vertices. Let $\nu_x := n(\alpha_{Y_x} - \alpha_{\bar{Y}_x})$. Combining Corollary~\ref{cor:moment_mgf_bounds_B} at $t=\Theta(1/h)$ with $s=\Theta(nh)$ gives a constant $C_*>0$ such that
    \begin{align*}
        \left\langle e^{t((B-B_{\mathrm{rem}})-\nu_0(Y))} \right\rangle_{\widehat{H}^G_s, Y, \tau} 
        = \left\langle e^{tB_{\mathrm{iso}}} \right\rangle_{\widehat{H}^G_s, Y, \tau}
        \prod_{x\in V(H)}  \left\langle e^{t(B_x-\nu_x)} \right\rangle_{G_x, Y_x, \tau_x}
        \leq 
        e^{c_1st^2}e^{C_0nht^2}
        \leq
        e^{C_*t^2nh}.
    \end{align*}
    Here $\tau_x$ denotes the restriction of $\tau$ to $T_x$, and the use of
    the corollary is legitimate because
    $|T_x|=O(n^\theta)=o(nt^2)=o(n/h^2)$. Choose
    $t_\delta=\delta/(2C_*h)$. Chernoff's bound gives
    \[
        \mu_{\widehat{H}^G_s, Y, \tau}(B-B_{\mathrm{rem}} \geq \nu_0(Y) + \delta n)
        \leq e^{-t \delta n} \Bigl\langle e^{t ((B-B_{\mathrm{rem}}) - \nu_0(Y))} \Bigr\rangle_{\widehat{H}^G_s, Y, \tau}
        \leq e^{-t \delta n + C_*t^2nh}
        = e^{ -\frac{\delta^2 n}{4C_*h}}.
    \]
    The same estimate holds for the lower tail. A union bound contributes a
    factor of $2$, which is absorbed by decreasing $c(\delta)>0$ for all
    sufficiently large $n$.
\end{proof}

The next estimates treat the complementary case in which the phase vector is balanced. We
no longer need a full local central limit theorem for the total imbalance. Instead, it is enough to
use the isolated vertices as an explicit smoothing variable. The first lemma records the elementary
point-probability bounds for the isolated contribution, and the second lemma transfers these bounds
to the full decoupled graph $\widehat H_s^G$.

\begin{lemma}
\label{lem:isolated-imbalance}
Fix $p \in (0, 1)$. There exist constants $\delta, c, C > 0$, depending only on $p$, such that the following holds. Let $s$ be a positive even integer, and let $X_L, X_R$ be independent $\text{Bin}(s/2, p)$ random variables. Then, for every integer $a$,
$$\mathbb{P}(X_L - X_R = a) \le \frac{C}{\sqrt{s}}.$$
Moreover, for every integer $a$ with $|a| \le \delta s$,
$$\mathbb{P}(X_L - X_R = a) \ge \frac{c}{\sqrt{s}} \exp \left( -C \frac{a^2}{s} \right).$$
\end{lemma}

\begin{proof}
Write $M = s/2$. The upper bound follows from the standard binomial anti-concentration estimate $\sup_j \mathbb{P}(\text{Bin}(M, p) = j) \le C_0/\sqrt{M}$. Indeed,
$$\mathbb{P}(X_L - X_R = a) = \sum_{j} \mathbb{P}(X_R = j)\mathbb{P}(X_L = j + a) \le \sup_{i} \mathbb{P}(X_L = i) \le \frac{C}{\sqrt{s}}.$$

For the lower bound, we use a standard consequence of Stirling's formula. There are constants $\delta_0, c_0, C_0 > 0$, depending only on $p$, such that, whenever $|r - Mp| \le \delta_0 M$,
$$\mathbb{P}(\text{Bin}(M, p) = r) \ge \frac{c_0}{\sqrt{M}} \exp \left( -C_0 \frac{(r - Mp)^2}{M} \right).$$

Choose $\delta > 0$ sufficiently small in terms of $p$. For a fixed integer $a$ with $|a| \le \delta s$, let $I_a$ be the set of integers $j$ satisfying
$$\left| j - \left( Mp - \frac{a}{2} \right) \right| \le c_1 \sqrt{M},$$
where $c_1 > 0$ is a small constant depending only on $p$. For all $j \in I_a$, both $j$ and $j + a$ lie in the range where the preceding Stirling estimate applies, for all sufficiently large $M$. Hence
$$\mathbb{P}(X_R = j)\mathbb{P}(X_L = j + a) \ge \frac{c_2}{M} \exp \left( -C_2 \frac{a^2}{M} \right), \quad j \in I_a.$$

Since $|I_a| \ge c_3 \sqrt{M}$, summing over $j \in I_a$ gives
\begin{align*}
\mathbb{P}(X_L - X_R = a) &\ge \sum_{j \in I_a} \mathbb{P}(X_R = j)\mathbb{P}(X_L = j + a) \\
&\ge \frac{c}{\sqrt{M}} \exp \left( -C \frac{a^2}{M} \right).
\end{align*}
Since $M = s/2$, this is the desired lower bound, after adjusting constants. The finitely many small values of $s$ are absorbed by decreasing $c$.
\end{proof}

\begin{lemma}
\label{lem:decoupled-imbalance}
Assume the gadget $G$ satisfies the conclusions of Corollary~\ref{cor:moment_mgf_bounds_B}. Suppose $s = \Theta(nh)$, with $s/2$ isolated vertices added to each side of the bipartition. 
Then there is a constant $C < \infty$ such that, uniformly over all phase vectors $Y \in \{+,-\}^{V(H)}$, all terminal configurations $\tau \in \{0,1\}^{T}$, and all integers $r$,
$$\mu_{\widehat{H}^G_s, Y, \tau} (B = r) \le \frac{C}{\sqrt{nh}}.$$
Moreover, uniformly over all $Y \in \{+,-\}^{V(H)}$ with $D(Y) = 0$ and all $\tau \in \{0,1\}^{T}$,
$$
    \mu_{\widehat{H}^G_s, Y, \tau} (B = 0) \ge \frac{e^{-Ch\log n}}{\sqrt{nh}}.
$$
\end{lemma}

\begin{proof}
Write
\[
    B=S_{\mathrm{full}}+B_{\mathrm{iso}},
    \quad
    S_{\mathrm{full}}:=S+B_{\mathrm{rem}}, \quad 
    S:=\sum_{x\in V(H)} B_x.
\]
The random variable $S_{\mathrm{full}}$ contains the entire imbalance contribution of the gadget copies, including all non-core vertices. 
Since $\widehat H_s^G$ is a disjoint union and the conditioning on $Y,\tau$ concerns only the gadget copies, $S_{\mathrm{full}}$ and $B_{\mathrm{iso}}$ are independent under $\mu_{\widehat H_s^G,Y,\tau}$.

Let $p = \lambda / (1 + \lambda)$. The isolated contribution has the form $B_{\text{iso}} = X_L - X_R,$ where $X_L$ and $X_R$ are independent $\text{Bin}(s/2, p)$ random variables. Therefore, by conditioning on $S_{\mathrm{full}}$ and applying Lemma~\ref{lem:isolated-imbalance},
\begin{align*}
    \mu_{\widehat{H}^G_s, Y, \tau} (B = r) &= \mathbb{E}\left[ \mu_{\widehat{H}^G_s, Y, \tau} (B_{\text{iso}} = r - S_{\mathrm{full}} \mid S_{\mathrm{full}}) \right] \\
    &= \mathbb{E}\left[ \mu_{\widehat{H}^G_s, Y, \tau} (B_{\text{iso}} = r - S_{\mathrm{full}}) \right] \\
    &\le \frac{C}{\sqrt{s}} \le \frac{C}{\sqrt{nh}}.
\end{align*}
This proves the upper bound. It remains to prove the lower bound when $D(Y) = 0$.  By Corollary~\ref{cor:moment_mgf_bounds_B}, uniformly in the phase and terminal conditioning of a single gadget,
$$
    \left| \langle B_x \rangle_{G_x, Y_x, \tau_x} - n(\alpha_{Y_x} - \alpha_{\bar{Y}_x}) \right| \le C_1 \sqrt{n\log n}, \quad \text{Var}_{G_x, Y_x, \tau_x}(B_x) \le C_1 n\log n.
$$
Summing over the independent gadget copies gives
$$
    \langle S \rangle_{\widehat{H}^G_s, Y, \tau} = n (\alpha_+ - \alpha_-)D(Y) + O(h\sqrt{n\log n}),\quad
    \text{Var}_{\widehat{H}^G_s, Y, \tau}(S) \le C_2 nh\log n.
$$
Since $D(Y) = 0$, the leading phase contribution cancels. 
Thus $\left| \langle S \rangle_{\widehat{H}^G_s, Y, \tau} \right| \le C_3 h \sqrt{n\log n}.$
By Chebyshev's inequality, after increasing the constant $C_4$ if necessary,
$$
    \mu_{\widehat{H}^G_s, Y, \tau} \left( |S| \le C_4 h \sqrt{n\log n} \right) \ge c_4
$$
for some constant $c_4 > 0$.

On the event
$\mathcal A:=\{|S|\le C_4h\sqrt{n\log n}\}$, the bound
$|B_{\mathrm{rem}}|=o(h\sqrt n)$ implies that, for all sufficiently large $n$,
\[
    |S_{\mathrm{full}}|
    =
    |S+B_{\mathrm{rem}}|
    \le
    C_5h\sqrt{n\log n}.
\]

For all sufficiently large $n$, the inequality $|a| \le C_5 h \sqrt{n\log n}$ implies $|a| \le \delta s$, where $\delta$ is the constant from Lemma~\ref{lem:isolated-imbalance}; this uses $s = \Theta(nh)$. 
Hence Lemma~\ref{lem:isolated-imbalance} gives, uniformly for all integers $a$ with $|a| \le C_5 h \sqrt{n\log n}$,
$$
    \mu_{\widehat{H}^G_s, Y, \tau} (B_{\text{iso}} = a) \ge \frac{c}{\sqrt{s}} \exp \left( -C \frac{a^2}{s} \right) \ge \frac{e^{-C_6 h\log n}}{\sqrt{nh}},
$$
because $a^2/s = O(h\log n)$ throughout this window. Therefore,
\[
\begin{aligned}
    \mu_{\widehat H_s^G,Y,\tau}(B=0)
    &=
    \mathbb E_{\widehat H_s^G,Y,\tau}
    \left[
        \mu_{\widehat H_s^G,Y,\tau}
        \left(
            B_{\mathrm{iso}}=-S_{\mathrm{full}}
            \mid
            S_{\mathrm{full}}
        \right)
    \right] \\
    &\ge
    \frac{e^{-C_6h\log n}}{\sqrt{nh}}
    \mu_{\widehat H_s^G,Y,\tau}
    \left(
        |S|\le C_4h\sqrt{n\log n}
    \right) \\
    &\ge
    \frac{e^{-Ch\log n}}{\sqrt{nh}},
\end{aligned}
\]
after absorbing the fixed constant $c_4$ into the exponential term.
\end{proof}

\begin{lemma}
\label{lem:exact-gadget-annealed-normalization}
For the fixed choices of $\theta,\psi$ supplied by
Lemma~\ref{lem:all_prelim}, define
\[
    \mathcal A_n:=\mathbb E_G Z_G(\lambda),
\]
where the expectation is over the exact gadget distribution of
Section~\ref{sec:hardness}.  There is a constant
$C_{\mathrm{pf}}>0$, depending only on the fixed gadget parameters, such
that, with probability at least $99/100$ over the gadget $G$,
\begin{equation}
    n^{-C_{\mathrm{pf}}}\mathcal A_n
    \le Z_G(\lambda)\le
    n^{C_{\mathrm{pf}}}\mathcal A_n.
    \label{eq:exact-gadget-polynomial-window}
\end{equation}
Moreover, in polynomial time, one can obtain a deterministic $\ell_n$ such
that $|\ell_n-\log\mathcal A_n|\le1$.
\end{lemma}

We defer the proof of Lemma~\ref{lem:exact-gadget-annealed-normalization} to Subsection~\ref{subsec:proof-exact-gadget-annealed-normalization}, after completing the main reduction.

\subsection{Proof of Theorem~\ref{thm:N_eta_approx_reduction}}

\begin{proof}[Proof of Theorem \ref{thm:N_eta_approx_reduction}]

Let $H$ be the input graph for \textbf{MIN-BISECTION}, and construct $H^G_s$ and 
$\widehat{H}^G_s$ as above.
We write $Z^{\text{bal}}_{H^G_s}$ and $Z^{\text{bal}}_{\widehat{H}^G_s}$ for the corresponding balanced partition functions and $Z_{\widehat{H}^G_s}$ for the ordinary partition function of the decoupled graph.
For $Y \in \{+,-\}^{V(H)}$, let $Z^{\text{bal}}_{H^G_s}(Y)$ be the contribution from balanced independent sets whose gadget phase vector is $Y$, and define $Z_{\widehat{H}^G_s}(Y)$, $Z_{\widehat{H}^G_s}(Y, \tau)$, and 
$Z^{\text{bal}}_{\widehat{H}^G_s}(Y, \tau)$ analogously.

We begin by decomposing the balanced partition function according to the induced phase vector:
\[
Z^{\mathrm{bal}}_{H^G_s}
=
\sum_{Y} Z^{\mathrm{bal}}_{H^G_s}(Y)
=
\sum_{Y:D(Y)=0} Z^{\mathrm{bal}}_{H^G_s}(Y)
+
\sum_{Y:D(Y)\neq 0} Z^{\mathrm{bal}}_{H^G_s}(Y).
\]
Thus it suffices to understand separately the contribution of the phase vectors with the correct macroscopic balance and those with the wrong macroscopic balance.

For a fixed phase vector $Y$, the only difference between $H^G_s$ and 
$\widehat{H}^G_s$ is the set $\mathcal{E}$ of inter-gadget edges. Hence, decomposing according to the global terminal pattern $\tau \in \{0,1\}^T$,

\begin{align}\label{eq: Z_bal_HGsY_to_Z_bal_HatHGsY}
    Z^{\mathrm{bal}}_{H^G_s}(Y)
    =
    \sum_{\tau\in \{0,1\}^{T}} Z^{\mathrm{bal}}_{\widehat H^G_s}(Y,\tau) \prod_{uv\in \mathcal{E}}(1-\tau_u\tau_v),
\end{align}
since $\tau$ is compatible with the inter-gadget edges exactly when no edge $uv\in E$ has both endpoints occupied.

Next we separate the exact balancedness constraint from the unconstrained gadget weight. For fixed $Y$ and $\tau$,
\[
    Z^{\mathrm{bal}}_{\widehat H^G_s}(Y,\tau)
    = 
    Z_{\widehat H^G_s}(Y,\tau)\cdot \mu_{\widehat H^G_s,Y,\tau}(B=0).
\]

First suppose $D(Y ) = 0$ and by Lemma~\ref{lem:decoupled-imbalance}, we have that:
there is a constant $C>0$ such that for every sufficiently large $n$ the following holds uniformly in $\tau$:
\[
    \mu_{\widehat H^G_s,Y,\tau}(B=0)
    \le 
    \frac{C}{\sqrt{nh}}.
\]
Hence, for the good phase vectors $Y$ satisfying $D(Y)=0$,
\[
    Z^{\mathrm{bal}}_{\widehat H^G_s}(Y,\tau)
    \le
    \frac{C}{\sqrt{nh}}\,
    Z_{\widehat H^G_s}(Y,\tau).
\]
Substituting this into (\ref{eq: Z_bal_HGsY_to_Z_bal_HatHGsY}) and summing over $Y$ with $D(Y) = 0$ give
\begin{align} \label{ineq: upper_bound_sum_Zbal_HGsY}
    \sum_{Y:D(Y)=0} Z^{\mathrm{bal}}_{H^G_s}(Y)
    \le
    \frac{C}{\sqrt{nh}}
    \sum_{Y:D(Y)=0}
    \sum_{\tau\in \{0,1\}^{T}}
    Z_{\widehat H^G_s}(Y,\tau)\,
    \prod_{uv \in \mathcal{E}}(1-\tau_u\tau_v).
\end{align}

Now we apply the phase-conditioned nearly-independent terminal law to upper bound $Z_{\widehat H^G_s}(Y,\tau)$ for each such $Y$. By Lemma ~\ref{lem:all_prelim}(i) (and the fact that $(1+\BigO(n^{-2\theta}))^h = 1+o(1)$), we have
\[
    Z_{\widehat H^G_s}(Y,\tau)
    =
    Z_{\widehat H^G_s}(Y)\,
    \mu_{\widehat H^G_s,Y}(\sigma_T=\tau)
    \leq
     Z_{\widehat H^G_s}(Y) (1+o(1))\,Q_T^Y(\tau),
\]
where $ Q_T^Y(\tau)=\prod_{x\in V(H)} Q_{T_x}^{Y_x}(\tau_{T_x})$
is the product of the single-gadget terminal laws in the prescribed phases. 
Absorbing the $1+o(1)$ factor into the constant $C$, we obtain
\[
    \sum_{Y:D(Y)=0} Z^{\mathrm{bal}}_{H^G_s}(Y)
    \le
    \frac{C}{\sqrt{nh}}
    \sum_{Y:D(Y)=0}
    Z_{\widehat H^G_s}(Y)
    \sum_{\tau\in \{0,1\}^{T}}
    Q_T^Y(\tau)
    \prod_{uv \in \mathcal{E}}(1-\tau_u\tau_v).
\]

It remains to evaluate the compatibility probability under the product law $Q_T^Y$. 
Define
$$\Gamma := 1 - q_+ q_-, \quad \Theta_+ := 1 - q_+^2, \quad \Theta_- := 1 - q_-^2, \quad \Theta := \sqrt{\Theta_+ \Theta_-}.$$
If $Y_x = Y_y$, each of the $2k$ linking edges corresponding to $xy \in E(H)$ 
is legal with probability $\Gamma$. If $Y_x \neq Y_y$, one of the two matchings 
contributes $\Theta_+^k$ and the other contributes $\Theta_-^k$. Since no terminal 
is used more than once, the compatibility probability factors over $E(H)$, giving
\[
    \sum_{\tau\in \{0,1\}^{T}}Q_T^Y(\tau)\prod_{uv\in \mathcal{E}}(1-\tau_u\tau_v)
    =
    \Gamma^{2k|E(H)|}\left(\frac{\Theta}{\Gamma}\right)^{2k\,\mathrm{cut}(Y)}.
\]
Also note that $\Theta < \Gamma$, since $\Gamma^2 - \Theta^2 = (q_+ - q_-)^2 > 0.$

Substituting this identity into (\ref{ineq: upper_bound_sum_Zbal_HGsY}) yields
\[
    \sum_{Y:D(Y)=0} Z^{\mathrm{bal}}_{H^G_s}(Y)
    \le
    \frac{C}{\sqrt{nh}}
    \sum_{Y:D(Y)=0}
     Z_{\widehat H^G_s}(Y)\,
    \Gamma^{2k|E(H)|}
    \left(\frac{\Theta}{\Gamma}\right)^{2k\,\mathrm{cut}(Y)}.
\]
If $D(Y ) = 0$, then $Y$ induces a bisection of H, by definition of the minimum bisection value $b$, we have $\mathrm{cut}(Y)\geq b$, which gives
\[
    \left(\frac{\Theta}{\Gamma}\right)^{2k\,\mathrm{cut}(Y)}
    \le
    \left(\frac{\Theta}{\Gamma}\right)^{2kb}.
\]
Together with an obvious upper bound
$
    \sum_{Y:D(Y)=0}Z_{\widehat H^G_s}(Y) \le Z_{\widehat H^G_s},
$
we have
\[
    \sum_{Y:D(Y)=0} Z^{\mathrm{bal}}_{H^G_s}(Y)
    \le
    \frac{C}{\sqrt{nh}}\, \Gamma^{2k|E(H)|}
    \left(\frac{\Theta}{\Gamma}\right)^{2kb}
    Z_{\widehat H^G_s}.
\]
Now we control the contribution of phase vectors $Y$ with $D(Y) \neq 0$. By definition,
\[
    \nu_0(Y)= (\alpha_+-\alpha_-)n\,D(Y).
\]
Since \(\alpha_+>\alpha_-\) and \(D(Y)\neq0\), we have
$|\nu_0(Y)|\ge(\alpha_+-\alpha_-)n$. Moreover,
$|B_{\mathrm{rem}}|=o(n)$. Hence, on the event $\{B=0\}$ and for all
sufficiently large $n$,
\[
\left|(B-B_{\mathrm{rem}})-\nu_0(Y)\right|
=\left|-B_{\mathrm{rem}}-\nu_0(Y)\right|
\ge \delta n
\]
for a constant $\delta>0$. Therefore
\[
    \{B=0\}\subseteq \{|(B-B_{\mathrm{rem}})-\nu_0(Y)|\ge \delta n\}.
\]
Applying Lemma~\ref{lem:hardcore13} yields
\[
    \mu_{\widehat H_s^G,Y,\tau}(B=0)
    \le
    \mu_{\widehat H_s^G,Y,\tau}(|(B-B_{\mathrm{rem}})-\nu_0(Y)|\ge \delta n)
    \le e^{-\Omega(n/h)}.
\]
Consequently,
\begin{align*}
    \sum_{Y:D(Y) \neq 0} Z^{\mathrm{bal}}_{H_s^G}(Y)
    &= 
    \sum_{Y:D(Y) \neq 0} \sum_{\tau\in \{0,1\}^{T}} Z^{\mathrm{bal}}_{\widehat H^G_s}(Y,\tau) \prod_{uv\in \mathcal{E}}(1-\tau_u\tau_v)
    \\
    &=
    \sum_{Y:D(Y) \neq 0} \sum_{\tau\in \{0,1\}^{T}} Z_{\widehat H_s^G}(Y,\tau)\,\mu_{\widehat H_s^G,Y,\tau}(B=0) \prod_{uv\in \mathcal{E}}(1-\tau_u\tau_v) \\
    &\leq
    e^{-\Omega(n/h)} 
    \sum_{Y:D(Y) \neq 0} \sum_{\tau\in \{0,1\}^{T}}Z_{\widehat H_s^G}(Y,\tau) \\
    &\leq
     e^{-\Omega(n/h)} Z_{\widehat H_s^G}.
\end{align*}
Thus the total contribution of the non-balanced phase vectors is exponentially small.

Combining this with the estimate for the balanced phase vectors gives
\begin{equation}\label{eq:upper_total_HGs_bal}
    Z^{\mathrm{bal}}_{H_s^G}
    \le
    \left(\frac{C}{\sqrt{nh}}\, \Gamma^{2k|E(H)|}
    \left(\frac{\Theta}{\Gamma}\right)^{2kb} + e^{-\Omega(n/h)}\right) Z_{\widehat H_s^G}.
\end{equation}

We now prove the matching lower bound. 
Let \(Y^\star\in\{+,-\}^{V(H)}\) be any phase vector with \(D(Y^\star)=0\) and \(\mathrm{cut}(Y^\star)=b\), i.e. \(Y^\star\) encodes a minimum bisection of \(H\). 
By the lower bound of Lemma~\ref{lem:decoupled-imbalance}, for every $\tau \in \{0,1\}^{T}$ and some constant $C>0$
\[
    \mu_{\widehat H_s^G,Y^\star,\tau}(B=0)\ge \frac{e^{-Ch\log n}}{\sqrt{nh}}.
\]
Using the phase-conditioned terminal law as before, we obtain the following lower bound
\begin{align} \label{eq:lower_total_HGs_bal}
    Z^{\mathrm{bal}}_{H_s^G}
    &\ge
    Z^{\mathrm{bal}}_{H_s^G}(Y^\star) \nonumber\\
    &=
    \sum_{\tau\in \{0,1\}^{T}}
    Z^{\mathrm{bal}}_{\widehat H_s^G}(Y^\star,\tau)
    \prod_{uv\in \mathcal{E}}(1-\tau_u\tau_v) \nonumber\\
    &= \sum_{\tau\in \{0,1\}^{T}}
    Z_{\widehat H_s^G}(Y^\star,\tau)\,
    \mu_{\widehat H_s^G,Y^\star,\tau}(B=0) 
    \prod_{uv\in \mathcal{E}}(1-\tau_u\tau_v) \nonumber\\
    &\ge
    \frac{e^{-Ch\log n}}{\sqrt{nh}}
    \sum_{\tau\in \{0,1\}^{T}}
    Z_{\widehat H_s^G}(Y^\star,\tau)\, 
    \prod_{uv\in \mathcal{E}}(1-\tau_u\tau_v) \nonumber\\
    &\ge
    \frac{e^{-Ch\log n}}{\sqrt{nh}}
    Z_{\widehat H_s^G}(Y^\star)
    \sum_{\tau\in \{0,1\}^{T}}
    Q_T^{Y^\star}(\tau)
    \prod_{uv\in \mathcal{E}}(1-\tau_u\tau_v)\nonumber \\
    &=
    \frac{e^{-Ch\log n}}{\sqrt{nh}}\,
    Z_{\widehat H_s^G}(Y^\star)\,
    \Gamma^{2k|E(H)|}
    \left(\frac{\Theta}{\Gamma}\right)^{2kb}.
\end{align}

We now compare $Z_{\widehat H_s^G}(Y^\star)$ estimates to $Z_{\widehat{H}^G_s}$. 
Since $\widehat{H}^G_s$ is the disjoint union of $h$ gadget copies and $s$ isolated  vertices,
$$Z_{\widehat{H}^G_s} = (1 + \lambda)^s (Z_{G,+} + Z_{G,-})^h.$$
For the minimum-bisection phase vector $Y^\star$, which has exactly $h/2$ plus phases  and $h/2$ minus phases, $$Z_{\widehat{H}^G_s}(Y^\star) = (1 + \lambda)^s Z_{G,+}^{h/2} Z_{G,-}^{h/2}.$$
Using the ratio $Z_{G,+} / Z_{G,-} = n^{O(1)}$ from Lemma~\ref{lem:all_prelim}, we get
$
    Z_{\widehat{H}^G_s}(Y^\star) \ge e^{-C_0h\log n} Z_{\widehat{H}^G_s},
$
for some constants \(C_0>0\).

Dividing \eqref{eq:upper_total_HGs_bal} and \eqref{eq:lower_total_HGs_bal} by $Z_{\widehat{H}^G_s}$, and absorbing the factor $\sqrt{nh}$ into $e^{O(h\log n)}$, gives constants $C_1, C_2 < \infty$ such that
\begin{equation}\label{eq:ratio_bounds_main}
    e^{-C_1 h\log n} \Gamma^{2k|E(H)|} \left( \frac{\Theta}{\Gamma} \right)^{2kb} \le \frac{Z^{\text{bal}}_{H^G_s}}{Z_{\widehat{H}^G_s}} \le e^{C_1 h\log n} \Gamma^{2k|E(H)|} \left( \frac{\Theta}{\Gamma} \right)^{2kb} + e^{-C_2 n / h}.
\end{equation}

The additive term is negligible on the logarithmic scale because $k|E(H)| = O(kh^2) = o(n/h)$. Therefore,
$$
    \log Z^{\text{bal}}_{H^G_s} - \log Z_{\widehat{H}^G_s} = 2k|E(H)| \log \Gamma + 2kb \log \left( \frac{\Theta}{\Gamma} \right) + O(h\log n).
$$
Equivalently,
\begin{equation}\label{eq:reduceing_b_from_Z}
    b = \frac{\log Z_{\widehat{H}^G_s} - \log Z^{\text{bal}}_{H^G_s} + 2k|E(H)| \log \Gamma}{2k \log(\Gamma / \Theta)} + O\left( \frac{h\log n}{k} \right).
\end{equation}

Let $N = |V(H^G_s)|$. Since $s = \Theta(nh)$ and
$h=\Theta(n^{\theta/4})$, we have
\[
    N=\Theta(nh)=\Theta(n^{1+\theta/4}),
    \qquad
    k=\Theta(n^{3\theta/4}).
\]
Fix
\begin{equation}
    0<\zeta<\frac{3\theta}{4+\theta}.
    \label{eq:airtight-zeta-choice}
\end{equation}
Then $h\log n=o(k)$ and $N^\zeta=o(k)$.

Assume that we have the randomized polynomial-time $e^{N^\zeta}$-factor approximation algorithm from the statement of the theorem. Applying it to the coupled graph $H^G_s$ gives $\widetilde{Z}^{\text{bal}}_{H^G_s}$ with
$$\left| \log \widetilde{Z}^{\text{bal}}_{H^G_s} - \log Z^{\text{bal}}_{H^G_s} \right| \le N^\zeta$$
with probability at least $2/3$.

For the ordinary decoupled partition function, let $\ell_n$ be supplied by
Lemma~\ref{lem:exact-gadget-annealed-normalization},
and define
\begin{equation}
    L_{\mathrm{dec}}
    :=s\log(1+\lambda)+h\ell_n.
    \label{eq:airtight-decoupled-center}
\end{equation}
Since
\[
    Z_{\widehat H_s^G}=(1+\lambda)^s Z_G^h,
\]
Lemma~\ref{lem:exact-gadget-annealed-normalization} gives
\begin{equation}
    \left|L_{\mathrm{dec}}-\log Z_{\widehat H_s^G}\right|
    \le h(1+C_{\mathrm{pf}}\log n)
    =O(h\log n).
    \label{eq:airtight-decoupled-log-error}
\end{equation}

Now define
\[
    \widetilde{b}
    :=
    \frac{L_{\mathrm{dec}}
    -\log \widetilde{Z}^{\text{bal}}_{H^G_s}
    +2k|E(H)|\log\Gamma}
    {2k\log(\Gamma/\Theta)}.
\]
By Lemma~\ref{lem:exact-gadget-annealed-normalization},
\eqref{eq:airtight-decoupled-log-error}, and
\eqref{eq:reduceing_b_from_Z},
\[
    |\widetilde b-b|
    =O\left(\frac{N^\zeta}{k}+\frac {h\log n}k\right)
    =o(1).
\]
For all sufficiently large $h$ this error is less than $1/3$; finitely
many smaller inputs can be handled by brute force. Since $b$ is an
integer, rounding $\widetilde b$ recovers the exact minimum-bisection value.

The conclusions of Lemma~\ref{lem:all_prelim} hold with probability at
least $99/100$.  The simultaneous left--right conclusions needed from
Corollary~\ref{cor:moment_mgf_bounds_B} hold with probability at least
$9/10$.  Their intersection consequently has probability at least
$89/100$.  Conditional on every resulting simple graph, the assumed
approximation algorithm succeeds with probability at least $2/3$.
Thus one execution of the counting reduction succeeds with probability at
least
\[
    \frac{89}{100}\cdot\frac23=\frac{89}{150}>\frac12.
\]
Independent repetition and majority vote amplify this probability in the
standard way. This proves the counting part of the theorem. In particular,
an FPRAS for $Z^{\text{bal}}_F(\lambda)$ would also imply such an algorithm,
since an FPRAS gives a much stronger approximation than an
$e^{N^\zeta}$-factor approximation.

The same estimates also rule out an efficient approximate sampler.  Comparing
the non-balanced-phase upper bound with \eqref{eq:lower_total_HGs_bal} gives
\[
\mu^{\mathrm{bal}}_{H_s^G}(D(Y)\ne0)
\le \exp\!\left(-\Omega(n/h)+O(kh^2+h\log n)\right)=o(1),
\]
since $kh^2+h\log n=o(n/h)$.  Among phase vectors with $D(Y)=0$, increasing the
cut value by one multiplies the terminal-compatibility factor by
$(\Theta/\Gamma)^{2k}$, while all balance-point-probability, phase-counting,
and multiplicity losses contribute only $e^{O(h\log n)}$ in total.  Hence
\[
\mu^{\mathrm{bal}}_{H_s^G}(D(Y)=0,\ \mathrm{cut}(Y)>b)
\le e^{O(h\log n)}(\Theta/\Gamma)^{2k}=o(1),
\]
because $k\gg h\log n$.
Therefore, after independently tossing a fair auxiliary sign for every tied
gadget copy, a sampler within, say, total variation distance $1/10$ would,
with probability bounded away from zero, yield a phase vector that is a
minimum bisection of $H$. Repeating the sampler polynomially many times would
solve \textbf{MIN-BISECTION} in randomized polynomial time.
\end{proof}

\subsection{Proof of Lemma~\ref{lem:exact-gadget-annealed-normalization}}
\label{subsec:proof-exact-gadget-annealed-normalization}

\begin{proof}[Proof of Lemma~\ref{lem:exact-gadget-annealed-normalization}]
As specified in the gadget construction, $\mathcal A_n$ and all
probabilities below refer to the unconditioned matching law; suppressing
parallel copies preserves every partition function used here.

By Lemma~\ref{lem:all_prelim}\textnormal{(iii)}, with probability at least
$99/100$, simultaneously for $\varsigma\in\{+,-\}$,
\[
    n^{-C'}\mathbb E Z_{G,\varsigma}
    \le Z_{G,\varsigma}\le
    n^{C'}\mathbb E Z_{G,\varsigma}.
\]
Since $Z_G=Z_{G,+}+Z_{G,-}$, summing these inequalities gives
\[
    n^{-C'}\mathbb E Z_G
    \le Z_G\le
    n^{C'}\mathbb E Z_G.
\]
Thus \eqref{eq:exact-gadget-polynomial-window} holds with
$C_{\mathrm{pf}}=C'$. This is the small-subgraph-conditioning conclusion
for the exact gadget, including the appended trees, already recorded in
Lemma~\ref{lem:all_prelim}; no comparison with a different random-graph
ensemble is involved.

It remains to obtain the stated estimate $\ell_n$. 
Put $d:=\Delta-1$.
For a complete rooted $d$-ary tree of depth $d_n$ appearing in the gadget,
we use leaf-marked independence polynomials to keep track of the number of occupied leaves.
For $s\in\{0,1\}$, let $P_j^s(x)$ denote the partition function of a
depth-$j$ rooted $d$-ary tree, conditioned on the root having occupation
state $s$. Each occupied non-leaf vertex contributes its usual activity
$\lambda$, while each occupied leaf contributes the marking variable $x$.
Consequently, the coefficient of $x^r$ in $P_j^s(x)$ is the total weight of
configurations with root state $s$ and exactly $r$ occupied leaves.

At depth zero, the tree consists of a single leaf, so
\[
P_0^0(x)=1,
\qquad
P_0^1(x)=x.
\]
The activity of an occupied leaf is not included here because each leaf is
also a vertex of the matching core, and its activity will be accounted for
separately below. For $j\ge 0$, conditioning on the state of the root gives
\begin{equation}
P_{j+1}^0(x)
=\bigl(P_j^0(x)+P_j^1(x)\bigr)^d,
\qquad
P_{j+1}^1(x)
=\lambda\bigl(P_j^0(x)\bigr)^d.
\label{eq:exact-tree-recursion}
\end{equation}
Indeed, when the root is unoccupied, each child root may be either occupied
or unoccupied, whereas an occupied root forces all of its children to be
unoccupied.

There are $m$ such trees on each side of the gadget. We therefore write
\begin{equation}
\bigl(P_{d_n}^0(x)+P_{d_n}^1(x)\bigr)^m
=\sum_{w=0}^{m'} c_w x^w.
\label{eq:exact-forest-polynomial}
\end{equation}
Thus, $c_w$ is the total weight contributed by the internal vertices of the
forest on one side, summed over all configurations in which exactly $w$ of
the $m'$ leaves shared with the matching core are occupied.

Recall that $M = n+m'$ is the size of each bipartition class of the matching core $G^{'}$.
For $a,b\ge0$, set
\[
    \rho_M(a,b):=
    \begin{cases}
    \displaystyle\frac{\binom{M-a}{b}}{\binom Mb},
        &a+b\le M,\\[1.2ex]
    0,&a+b>M.
    \end{cases}
\]
This is the probability that a uniform perfect matching contains no edge
between a fixed $a$-set on its left and a fixed $b$-set on its right.
Expose the last of the $\Delta$ matchings in the construction of $G'$ and
relabel its endpoints. Its $m'$ deleted edges leave a fixed size-$n$
matching between $U^L$ and $U^R$, while the other $\Delta-1$ matchings
remain independent uniform perfect matchings on $M$ vertices per side.
Consequently,
\begin{equation}
    \mathcal A_n
    =
    \sum_{\substack{u,v\ge0\\u+v\le n}}
    \binom nu\binom{n-u}{v}
    \sum_{w,z=0}^{m'}
    c_wc_z\,\lambda^{u+v+w+z}
    \rho_M(u+w,v+z)^{\Delta-1}.
    \label{eq:exact-annealed-first-moment}
\end{equation}
Indeed, $u,v$ are the occupation numbers in $U^L,U^R$;
$\binom nu\binom{n-u}{v}$ counts the choices compatible with the
surviving size-$n$ matching; $w,z$ count the occupied shared leaves; and
each remaining matching contributes the displayed avoidance probability.

The recursion \eqref{eq:exact-tree-recursion} and the convolution in
\eqref{eq:exact-forest-polynomial} involve polynomials of degree at most
$m'$. Equation~\eqref{eq:exact-annealed-first-moment} has
$O(n^2(m')^2)$ nonnegative summands. Since
$m'=O(n^{\theta+\psi})$, the displayed recursions and sum yield the required
$\ell_n$ in polynomial time.
\end{proof}

\section{Hardness for Fixed Slice Sampling} \label{sec:density_hardness}
In this section we prove Theorem~\ref{Thm:FixedDensityHardness}. The reduction uses the same phase-coexistence gadget construction and terminal-compatibility calculation as Section~\ref{sec:hardness}; the new ingredient is the centering
of a two-dimensional fixed slice.

Let
\(\lambda_0:=\lambda(\alpha)\), and set
\[
    d_0:=\alpha_+(\lambda_0)-\alpha_-(\lambda_0),
    \qquad
    \delta:=\alpha_L-\alpha_R .
\]
The assumption in Theorem~\ref{Thm:FixedDensityHardness} is equivalent to
\(|\delta|<d_0\). Choose \(\lambda_\star\in(\lambda_c(\Delta),\lambda_0)\) sufficiently close to
\(\lambda_0\), and write
\[
    a_\pm:=\alpha_\pm(\lambda_\star),
    \qquad
    d:=a_+-a_-,
    \qquad
    p:=\frac{\lambda_\star}{1+\lambda_\star}.
\]
Since \(\lambda_\star<\lambda_0\) is close to \(\lambda_0\), we have
\((a_++a_-)/2<\alpha\) and \(p>\alpha\). Define
\[
    u:=\frac{\alpha-(a_++a_-)/2}{p-\alpha}>0,
    \qquad
    \gamma:=\frac12\left(1+\frac{\delta(1+u)}{d}\right).
\]
By taking \(\lambda_\star\) sufficiently close to \(\lambda_0\), and perturbing
slightly if necessary, we may assume \(\gamma\in(0,1)\cap\mathbb Q\). Finally,
set $\beta_L:=\gamma a_+ +(1-\gamma)a_-$ and $\beta_R:=\gamma a_- +(1-\gamma)a_+.$ Then we have that
\begin{equation}
\label{eq:fixed_density_centering}
    \alpha_L(1+u)=\beta_L+up,
    \qquad
    \alpha_R(1+u)=\beta_R+up .
\end{equation}

The $\textbf{NP-hard}$ problem we reduce from is \(\gamma\)-\textsc{MEBC} which is defined below.

\begin{fact}[\(\gamma\)-\textsc{MEBC} {\cite[p.~644]{feige2003cutting}}]
\label{fact:gamma_MEBC_hard}
Fix any rational $\gamma\in(0,1).$ Given a graph $H$ on $h$ vertices with
$\gamma h\in\mathbb Z,$ it is $\textbf{NP-hard}$ to compute
\[
    b_\gamma(H)
    :=
    \min\left\{
        |E_H(S,V(H)\setminus S)|:
        |S|=\gamma h
    \right\}.
\]
\end{fact}

Let $H$ be an input graph on $h$ vertices with $\gamma h\in\mathbb Z$. Set
\[
    j_0:=\gamma h,
    \qquad
    D_0:=2j_0-h=(2\gamma-1)h.
\]
Except for the choice of fugacity and the number of isolated vertices, we use
the same construction as in Section~\ref{sec:hardness}. In particular, the graph $H_s^G$, the decoupled graph $\widehat H_s^G$, the global terminal set
$T$, the inter-gadget edge set $\mathcal E$, the product terminal law $Q_T^Y$,
and the compatibility factors $\Gamma,\Theta$ are defined as in
Section~\ref{sec:hardness}, but now with the gadget analyzed at fugacity
$\lambda_\star$. The only new centering condition is that the fixed slice selects
phase vectors with $D(Y)=D_0$, rather than phase vectors with $D(Y)=0$.
As there, independent fair auxiliary signs resolve tied gadgets, and every
phase-restricted quantity averages those signs with weight $2^{-h}$.

We choose the gadget size $n = n(h)$ polynomially large in $h$, with the 
constant in the polynomial chosen so that
$$h = \Theta\left(n^{\theta/4}\right), \quad kh \le m, \quad \frac{h\log n}{k} = o(1),$$
where
$$k := \left\lfloor n^{3\theta/4} \right\rfloor, \quad m := \left\lfloor n^\theta \right\rfloor.$$
We also choose the constants so that $kh^2=o(n/h)$ and $n^\theta=o(n/h^2).$
Given the gadget graph $G$ as above, write $g_n := |L(G)|=|R(G)|.$ By construction, $g_n=n+o(\sqrt n).$ 
Define the even integer $s := 2\lfloor uhn\rfloor.$
The graph $H_s^G$ is obtained from $h$ copies of $G$, together with $s/2$
isolated vertices on each side, by adding the same terminal matchings as in
Section~\ref{sec:hardness}. Both sides of $H_s^G$ have size
\[
    M:=h g_n+s/2.
\]
The target fixed slice is
\begin{equation} \label{eq:fixed_density_target_slice}
    r_L := \lfloor \alpha_L M \rfloor, \quad r_R := \lfloor \alpha_R M \rfloor.
\end{equation}

We retain the balance notation from Section~\ref{sec:hardness}. For each gadget copy $G_x$, let
\[
    X_{L,x}:=|I\cap U_x^L|,
    \qquad
    X_{R,x}:=|I\cap U_x^R|,
    \qquad
    B_x:=X_{L,x}-X_{R,x},
\]
so $X_{L,x},X_{R,x}$ and $B_x$ are the core occupation variables and core imbalance used in Section~\ref{sec:hardness}. Let
\[
    X_{L,x}^{\mathrm{rem}}
    :=
    |I\cap(L(G_x)\setminus U_x^L)|,
    \qquad
    X_{R,x}^{\mathrm{rem}}
    :=
    |I\cap(R(G_x)\setminus U_x^R)|.
\]
Accordingly, the non-core imbalance $B_{\mathrm{rem}}$ defined in Section~\ref{sec:hardness} is
\[
    B_{\mathrm{rem}}
    =
    \sum_{x\in V(H)}
    \left(
        X_{L,x}^{\mathrm{rem}}-X_{R,x}^{\mathrm{rem}}
    \right).
\]
Let $J_L$ and $J_R$ be the occupation counts of the isolated vertices on the left and right sides, respectively, and set
\[
    B_{\mathrm{iso}}:=J_L-J_R.
\]
The full side-occupation counts on $\widehat H_s^G$ are
\[
    X_L
    :=
    \sum_{x\in V(H)}
    \left(X_{L,x}+X_{L,x}^{\mathrm{rem}}\right)+J_L,
    \qquad
    X_R
    :=
    \sum_{x\in V(H)}
    \left(X_{R,x}+X_{R,x}^{\mathrm{rem}}\right)+J_R,
\]
and we write $\mathbf X:=(X_L,X_R)$. As in Section~\ref{sec:hardness}, $B$ denotes the true full imbalance, so
\[
    B:=X_L-X_R
    =
    \sum_{x\in V(H)}B_x+B_{\mathrm{iso}}+B_{\mathrm{rem}}.
\]
Thus the event $\mathbf X=(r_L,r_R)$ is exactly the target fixed-slice event, and on this event $B=r_L-r_R$.
The moment and exponential-moment estimates from Lemma~\ref{lem:moment_mgf_bounds_X} and Corollary~\ref{cor:moment_mgf_bounds_B} apply to the core variables, while the full constraint is expressed through $B$ and $\mathbf X$.

To adapt the hardness reduction utilized in the previous section, we need two inputs: an exponential penalty when the global phase is not aligned with the fixed slice and a sub-exponential penalty when the global phase is aligned with it. For the first input, we have the following.

\begin{lemma} \label{lem:fix_exp_penalty}
There exists $c>0$ such that the following holds for all sufficiently large $n$. Fix a phase vector $Y\in\{+,-\}^{V(H)},$ a terminal configuration $\tau\in\{0,1\}^T,$ and a target fixed slice $(r_L,r_R)\in\mathbb Z^2.$ Suppose that
$$
    \left|
        ndD(Y)
        -
        (r_L-r_R)
    \right|
    \ge
    c_0 n
$$
for some fixed constant $c_0>0$, and suppose that $G$ satisfies
Corollary~\ref{cor:moment_mgf_bounds_B} at the value $t=t_{c_0/2}$ from
Lemma~\ref{lem:hardcore13}. Then
$$
    \mu_{\widehat H_s^G,Y,\tau}
    \left(
        \mathbf X=(r_L,r_R)
    \right)
    \le
    \exp(-cn/h).
$$
In particular, for the slice \eqref{eq:fixed_density_target_slice}, this holds for every $Y$ with $D(Y)\ne D_0.$
\end{lemma}

\begin{proof}
Applying Lemma~\ref{lem:hardcore13} at fugacity $\lambda_\star$, with the
phase-density difference $\alpha_+-\alpha_-$ there equal to $d=a_+-a_-$ here
and with $\eta=c_0/2$, gives
\[
    \mu_{\widehat H_s^G,Y,\tau}
    \left(
        \left|(B-B_{\mathrm{rem}})-ndD(Y)\right|
        \ge
        \frac{c_0}{2}n
    \right)
    \le
    \exp(-cn/h)
\]
for some $c=c(c_0)>0$.

On the event $\mathbf X=(r_L,r_R)$, the full imbalance satisfies $B=r_L-r_R$. Hence, using $|B_{\mathrm{rem}}|=o(n)$, $
    \left|(B-B_{\mathrm{rem}})-ndD(Y)\right|
    \ge
    \frac{c_0}{2}n
$, for all sufficiently large $n$.
Therefore $
    \{\mathbf X=(r_L,r_R)\}
    \subseteq
    \left\{
        \left|(B-B_{\mathrm{rem}})-ndD(Y)\right|
        \ge
        \frac{c_0}{2}n
    \right\},
$
and Lemma~\ref{lem:hardcore13} yields
\[
    \mu_{\widehat H_s^G,Y,\tau}
    \left(\mathbf X=(r_L,r_R)\right)
    \le
    \exp(-cn/h).
\]

It remains to verify the final assertion. From~\eqref{eq:fixed_density_target_slice},
$g_n=n+o(\sqrt n)$, and $s/2=uhn+O(1)$, we have
\[
\begin{aligned}
    r_L-r_R
    &=
    (\alpha_L-\alpha_R)(hg_n+s/2)+O(1) \\
    &=
    \delta hn(1+u)+O(h\sqrt n) \\
    &=
    dnD_0+O(h\sqrt n),
\end{aligned}
\]
where the last identity follows from the definition of $\gamma$. If
$D(Y)\ne D_0$, then $|D(Y)-D_0|\ge 2$, and hence
\[
    |ndD(Y)-(r_L-r_R)|
    \ge
    dn-O(h\sqrt n)
    \ge
    \frac d2 n
\]
for all sufficiently large $n$.
\end{proof}

We now show that when $D(Y) = D_0$, we incur only a sub-exponential penalty.

\begin{lemma}
\label{lem:fix_poly_penalty}
There exists a constant \(C<\infty\) such that, for all sufficiently large
\(n\), provided $G$ satisfies the simultaneous left--right moment estimates
used in the proof of Corollary~\ref{cor:moment_mgf_bounds_B}, uniformly over
every phase vector \(Y\in\{+,-\}^{V(H)}\) with \(D(Y)=D_0\) and
every terminal configuration \(\tau\in\{0,1\}^T\),
\[
    \frac{e^{-Ch\log n}}{nh}
    \le
    \mu_{\widehat H_s^G,Y,\tau}
    \left(
        \mathbf X=(r_L,r_R)
    \right)
    \le
    \frac{C}{nh}.
\]
\end{lemma}

\begin{proof}
Define the core gadget counts
\[
    S_L:=\sum_{x\in V(H)}X_{L,x},
    \qquad
    S_R:=\sum_{x\in V(H)}X_{R,x},
\]
and the non-core gadget counts
\[
    S_L^{\mathrm{rem}}
    :=
    \sum_{x\in V(H)}X_{L,x}^{\mathrm{rem}},
    \qquad
    S_R^{\mathrm{rem}}
    :=
    \sum_{x\in V(H)}X_{R,x}^{\mathrm{rem}}.
\]
Thus
\[
    X_L=S_L+S_L^{\mathrm{rem}}+J_L,
    \qquad
    X_R=S_R+S_R^{\mathrm{rem}}+J_R.
\]
Write
\[
    \widetilde S_L:=S_L+S_L^{\mathrm{rem}},
    \qquad
    \widetilde S_R:=S_R+S_R^{\mathrm{rem}}
\]
for the full contributions of all gadget vertices. 
The isolated variables
$J_L,J_R\stackrel{\mathrm{ind}}{\sim}\operatorname{Bin}(s/2,p)$ are independent of $(\widetilde S_L,\widetilde S_R)$ under $\mu_{\widehat H_s^G,Y,\tau}$.

For the upper bound, conditioning on the full gadget contribution gives
\[
\begin{aligned}
    \mu_{\widehat H_s^G,Y,\tau}\left(\mathbf X=(r_L,r_R)\right)
    &=
    \mathbb E_{\widehat H_s^G,Y,\tau}
    \left[
        \mathbb P(J_L=r_L-\widetilde S_L)\,
        \mathbb P(J_R=r_R-\widetilde S_R)
    \right] \\
    &\le
    \sup_a\mathbb P(J_L=a)\,
    \sup_b\mathbb P(J_R=b)
    \le
    \frac{C}{nh},
\end{aligned}
\]
because $s=\Theta(nh)$ and the maximum point mass of a binomial $\operatorname{Bin}(s/2,p)$ random variable is $O(s^{-1/2})$.

For the lower bound, use $D(Y)=D_0$. Then $Y$ has exactly $j_0=\gamma h$ plus
phases and $h-j_0=(1-\gamma)h$ minus phases. Hence the leading centers of the
    core gadget contributions are $hn\beta_L$ and $hn\beta_R$. By the
    simultaneous left--right version of Lemma~\ref{lem:moment_mgf_bounds_X}
    established in the proof of Corollary~\ref{cor:moment_mgf_bounds_B}, after
    summing over the independent gadget copies,
\[
    \left|
        \mathbb E_{\widehat H_s^G,Y,\tau}S_L-hn\beta_L
    \right|
    \le C_0h\sqrt{n\log n},
    \qquad
    \left|
        \mathbb E_{\widehat H_s^G,Y,\tau}S_R-hn\beta_R
    \right|
    \le C_0h\sqrt{n\log n},
\]
and
\[
    \operatorname{Var}_{\widehat H_s^G,Y,\tau}(S_L)
    +
    \operatorname{Var}_{\widehat H_s^G,Y,\tau}(S_R)
    \le C_0nh\log n.
\]
Choose $A>0$ sufficiently large. Chebyshev's inequality and a union bound imply that, for some constant $c_0>0$,
\[
    \mu_{\widehat H_s^G,Y,\tau}(\mathcal G_S)
    \ge c_0,
\]
where
\[
    \mathcal G_S:=
    \left\{
        |S_L-hn\beta_L|\le Ah\sqrt{n\log n},
        \ |S_R-hn\beta_R|\le Ah\sqrt{n\log n}
    \right\}.
\]
Since each gadget has $O(n^{\theta+\psi})$ non-core vertices on each side,
\[
    0\le S_L^{\mathrm{rem}}\le Chn^{\theta+\psi},
    \qquad
    0\le S_R^{\mathrm{rem}}\le Chn^{\theta+\psi},
\]
and $hn^{\theta+\psi}=o(h\sqrt n)$. Therefore, after increasing $A$ if necessary, the event $\mathcal G_S$ implies
\[
    |\widetilde S_L-hn\beta_L|\le 2Ah\sqrt{n\log n},
    \qquad
    |\widetilde S_R-hn\beta_R|\le 2Ah\sqrt{n\log n}.
\]

On the other hand, by~\eqref{eq:fixed_density_centering}, $s/2=uhn+O(1)$, $g_n=n+o(\sqrt n)$, and the definition of $M$, $
    r_L
    =
    \alpha_L(hg_n+s/2)+O(1) =
    hn\beta_L+(s/2)p+O(h\sqrt n),
$
and similarly $
    r_R=hn\beta_R+(s/2)p+O(h\sqrt n)
$.

After increasing $A$ once more if necessary, the event $\mathcal G_S$ implies
\[
    |r_L-\widetilde S_L-(s/2)p|
    \le
    3Ah\sqrt{n\log n},
    \qquad
    |r_R-\widetilde S_R-(s/2)p|
    \le
    3Ah\sqrt{n\log n}.
\]
A standard local lower bound for binomial random variables gives constants $c_1,C_1>0$, depending only on $p$, such that whenever
$|j-(s/2)p|\le 3Ah\sqrt{n\log n}$,
\[
    \mathbb P(\operatorname{Bin}(s/2,p)=j)
    \ge
    \frac{c_1}{\sqrt s}
    \exp\left(-C_1\frac{(j-(s/2)p)^2}{s}\right).
\]
Since $s=\Theta(nh)$, the exponent above is $O(h\log n)$ throughout this window.
Thus, on $\mathcal G_S$,
\[
    \mathbb P(J_L=r_L-\widetilde S_L)
    \ge
    \frac{e^{-C_2h\log n}}{\sqrt{nh}},
    \qquad
    \mathbb P(J_R=r_R-\widetilde S_R)
    \ge
    \frac{e^{-C_2h\log n}}{\sqrt{nh}}.
\]
Therefore
\[
\begin{aligned}
    \mu_{\widehat H_s^G,Y,\tau}\left(\mathbf X=(r_L,r_R)\right)
    &\ge
    \mathbb E_{\widehat H_s^G,Y,\tau}
    \left[
        \mathbb P(J_L=r_L-\widetilde S_L)\,
        \mathbb P(J_R=r_R-\widetilde S_R)\,
        \mathbf 1_{\mathcal G_S}
    \right] \\
    &\ge
    \frac{e^{-C_3h\log n}}{nh}\,
    \mu_{\widehat H_s^G,Y,\tau}(\mathcal G_S)
    \ge
    \frac{e^{-Ch\log n}}{nh},
\end{aligned}
\]
\end{proof}

\begin{proof}[Proof of Theorem~\ref{Thm:FixedDensityHardness}]
Sample $G$ from the gadget distribution. Let $\mathcal G$ be the intersection
of the event in Lemma~\ref{lem:all_prelim} and the simultaneous left--right
event used in the proof of Corollary~\ref{cor:moment_mgf_bounds_B}, at the fixed
value $t=t_{d/4}$ needed for the final assertion of
Lemma~\ref{lem:fix_exp_penalty}.
Then $\mathbb P(\mathcal G)\ge 89/100$. We condition on $\mathcal G$ until the
success-probability calculation below.

Let \(H\) be the input graph for \(\gamma\)-\textsc{MEBC}, and write
\(b:=b_\gamma(H)\). We use the graph \(H_s^G\), the decoupled graph
\(\widehat H_s^G\), and the target slice \((r_L,r_R)\) constructed above. For
notational convenience, in the rest of the proof we write \(H^G\) and
\(\widehat H^G\) for \(H_s^G\) and \(\widehat H_s^G\).

All partition functions and measures below are taken at fugacity
\(\lambda_\star\). We write \(Z^{\mathrm{fix}}_{H^G}\) for the
\(\lambda_\star\)-weighted partition function over independent sets satisfying
\[
    |I\cap L(H^G)|=r_L,
    \qquad
    |I\cap R(H^G)|=r_R.
\]
Equivalently, the corresponding unweighted fixed-slice count differs from this
weighted quantity by the known factor \(\lambda_\star^{r_L+r_R}\).

For a phase vector \(Y\in\Sigma_H=\{+,-\}^{V(H)}\), write
\(Z^{\mathrm{fix}}_{H^G}(Y)\) for the contribution from fixed-size independent
sets whose induced phase vector is \(Y\), and write \(Z_{\widehat H^G}(Y)\) for
the unconstrained contribution on \(\widehat H^G\) from independent sets with
phase vector \(Y\). For a terminal occupation pattern
\(\tau\in\Omega_T:=\{0,1\}^T\), define
\(Z^{\mathrm{fix}}_{\widehat H^G}(Y,\tau)\) and
\(Z_{\widehat H^G}(Y,\tau)\) analogously.

We begin by decomposing the fixed-size partition function according to the induced phase vector:
\[
Z^{\mathrm{fix}}_{H^G}
=
\sum_{Y\in \Sigma_H} Z^{\mathrm{fix}}_{H^G}(Y)
=
\sum_{Y:D(Y)=D_0} Z^{\mathrm{fix}}_{H^G}(Y)
+
\sum_{Y:D(Y)\neq D_0} Z^{\mathrm{fix}}_{H^G}(Y).
\]

Fix a phase vector $Y$. Since the only edges connecting distinct gadget copies are the edges in \(\mathcal E\), we may decompose \(Z^{\mathrm{fix}}_{H^G}(Y)\) according to the occupation pattern on the global terminal set \(T\):

\begin{align}
    Z^{\mathrm{fix}}_{H^G}(Y)
    &=
    \sum_{\tau\in \Omega_T}
    Z^{\mathrm{fix}}_{\widehat H^G}(Y,\tau)
    \prod_{uv\in \mathcal E}(1-\tau_u\tau_v) \nonumber\\
    &=
    \sum_{\tau\in \Omega_T}
    Z_{\widehat H^G}(Y,\tau)\,
    \mu_{\widehat H^G,Y,\tau}\!\left(\mathbf X=(r_L,r_R)\right)
    \prod_{uv\in \mathcal E}(1-\tau_u\tau_v).
    \label{eq:Z_fix_HGY_to_Z_HatHGY}
\end{align}
since \(\tau\) is compatible with the inter-gadget edges exactly when no edge
\(uv\in\mathcal E\) has both endpoints occupied. Here,
\(\mu_{\widehat H^G,Y,\tau}\) is the hard-core measure on \(\widehat H^G\)
conditioned on the phase vector being \(Y\) and the terminal pattern being
\(\tau\).

Now, by Lemma~\ref{lem:fix_poly_penalty} and Lemma~\ref{lem:fix_exp_penalty}, there is a constant \(C>0\) such that for every sufficiently large \(n\) the following bounds hold uniformly in \(Y\) and \(\tau\):
\[
    \mu_{\widehat H^G,Y,\tau}\!\left(\mathbf X=(r_L,r_R)\right)
    \le 
    \frac{C}{nh}
    \qquad\text{whenever } D(Y)=D_0,
\]
whereas
\[
    \mu_{\widehat H^G,Y,\tau}\!\left(\mathbf X=(r_L,r_R)\right)
    \le 
    \exp(-\Omega(n/h))
    \qquad\text{whenever } D(Y)\neq D_0.
\]

Hence, for the good phase vectors \(Y\) satisfying \(D(Y)=D_0\),
\[
    Z^{\mathrm{fix}}_{\widehat H^G}(Y,\tau)
    \le
    \frac{C}{nh}\,
    Z_{\widehat H^G}(Y,\tau).
\]
Substituting this into \eqref{eq:Z_fix_HGY_to_Z_HatHGY} and summing over \(Y\) with \(D(Y)=D_0\) gives
\begin{align} \label{ineq:upper_bound_sum_Zfix_HGY}
    \sum_{Y:D(Y)=D_0} Z^{\mathrm{fix}}_{H^G}(Y)
    \le
    \frac{C}{nh}
    \sum_{Y:D(Y)=D_0}
    \sum_{\tau\in\Omega_T}
    Z_{\widehat H^G}(Y,\tau)\,
    \prod_{uv \in \mathcal E}(1-\tau_u\tau_v).
\end{align}

Now we apply the phase-conditioned nearly-independent terminal law to upper bound \(Z_{\widehat H^G}(Y,\tau)\) for each such \(Y\). By Lemma~\ref{lem:all_prelim}(i) and the fact that \((1+\BigO(n^{-2\theta}))^h=1+o(1)\), we have
\[
    Z_{\widehat H^G}(Y,\tau)
    =
    Z_{\widehat H^G}(Y)\,
    \mu_{\widehat H^G,Y}(\sigma_T=\tau)
    =
    Z_{\widehat H^G}(Y)(1+o(1))\,Q_T^Y(\tau).
\]

Absorbing the \(1+o(1)\) factor into the constant \(C\), and by the terminal-compatibility calculation from Section~\ref{sec:hardness}, 
\begin{align*}
\sum_{Y:D(Y)=D_0} Z^{\mathrm{fix}}_{H^G}(Y)
    &\le
    \frac{C}{nh}
    \sum_{Y:D(Y)=D_0}
    Z_{\widehat H^G}(Y)
    \sum_{\tau\in\Omega_T}
    Q_T^Y(\tau)
    \prod_{uv \in \mathcal E}(1-\tau_u\tau_v) \\
    &\leq
    \frac{C}{nh}
    \sum_{Y:D(Y)=D_0}
    Z_{\widehat H^G}(Y)
    \Gamma^{2k|E(H)|}
    \left(\frac{\Theta}{\Gamma}\right)^{2k\mathrm{cut}(Y)}.
\end{align*}
Since \(D(Y)=D_0\), the phase vector \(Y\) has exactly \(j_0=\gamma h\) plus
phases, and hence encodes a feasible solution to \(\gamma\)-\textsc{MEBC}. By
definition of \(b=b_\gamma(H)\), we have \(\mathrm{cut}(Y)\ge b\) and hence
\[
    \sum_{Y:D(Y)=D_0} Z^{\mathrm{fix}}_{H^G}(Y)
    \le
    \frac{C}{nh}\, \Gamma^{2k|E(H)|}
    \left(\frac{\Theta}{\Gamma}\right)^{2kb}
    Z_{\widehat H^G}.
\]

It remains to control the contribution of phase vectors \(Y\) with \(D(Y)\neq D_0\). By Lemma~\ref{lem:fix_exp_penalty}, for every such \(Y\) and every terminal pattern \(\tau\in\Omega_T\), we have
\begin{align}
    \mu_{\widehat H^G,Y,\tau}\!\left(\mathbf X=(r_L,r_R)\right)
    \le
    \exp(-\Omega(n/h)).
\end{align}
Consequently,
\begin{align*}
    \sum_{Y:D(Y)\neq D_0} Z^{\mathrm{fix}}_{H^G}(Y)
    &=
    \sum_{Y:D(Y)\neq D_0}
    \sum_{\tau\in\Omega_T}
    Z_{\widehat H^G}(Y,\tau)\,
    \mu_{\widehat H^G,Y,\tau}\!\left(\mathbf X=(r_L,r_R)\right)
    \prod_{uv\in \mathcal E}(1-\tau_u\tau_v)
    \\
    &\le
    \exp(-\Omega(n/h))
    \sum_{Y:D(Y)\neq D_0}
    \sum_{\tau\in\Omega_T}
    Z_{\widehat H^G}(Y,\tau)
    \prod_{uv\in \mathcal E}(1-\tau_u\tau_v)
    \\
    &\le
    \exp(-\Omega(n/h))
    \sum_{Y:D(Y)\neq D_0}
    \sum_{\tau\in\Omega_T}
    Z_{\widehat H^G}(Y,\tau)
    \\
    &\le
    \exp(-\Omega(n/h)) Z_{\widehat H^G}.
\end{align*}
Thus the total contribution of the phase vectors with \(D(Y)\neq D_0\) is exponentially small. Combining this with the estimate for the phase vectors satisfying \(D(Y)=D_0\) gives
\begin{equation}\label{eq:upper_total_HG_fix}
    Z^{\mathrm{fix}}_{H^G}
    \le
    \left(\frac{C}{nh}\, \Gamma^{2k|E(H)|}
    \left(\frac{\Theta}{\Gamma}\right)^{2kb} + e^{-\Omega(n/h)}\right) Z_{\widehat H^G}.
\end{equation}

We now prove the matching lower bound. 
Let \(Y^\star\in\Sigma_H\) be any phase vector with
\(D(Y^\star)=D_0\) such that \(\mathrm{cut}(Y^\star)=b\), i.e. \(Y^\star\)
encodes a minimum \(\gamma\)-fixed-cardinality cut of \(H\). By
Lemma~\ref{lem:fix_poly_penalty}, for every \(\tau\in\Omega_T\),
\begin{align}
    \mu_{\widehat H^G,Y^\star,\tau}\!\left(\mathbf X=(r_L,r_R)\right)
    \ge
    \frac{e^{-O(h\log n)}}{nh}.
\end{align}
Using the phase-conditioned terminal law as before, we obtain the following lower bound
\begin{align} \label{eq:lower_total_HG_fix}
    Z^{\mathrm{fix}}_{H^G}
    &\ge
    Z^{\mathrm{fix}}_{H^G}(Y^\star) \nonumber\\
    &=
    \sum_{\tau\in\Omega_T}
    Z^{\mathrm{fix}}_{\widehat H^G}(Y^\star,\tau)
    \prod_{uv\in \mathcal E}(1-\tau_u\tau_v) \nonumber\\
    &=
    \sum_{\tau\in\Omega_T}
    Z_{\widehat H^G}(Y^\star,\tau)\,
    \mu_{\widehat H^G,Y^\star,\tau}\!\left(\mathbf X=(r_L,r_R)\right)
    \prod_{uv\in \mathcal E}(1-\tau_u\tau_v) \nonumber\\
    &\ge
    \frac{e^{-O(h\log n)}}{nh}
    \sum_{\tau\in\Omega_T}
    Z_{\widehat H^G}(Y^\star,\tau)\,
    \prod_{uv\in \mathcal E}(1-\tau_u\tau_v) \nonumber\\
    &\ge
    \frac{(1-o(1))e^{-O(h\log n)}}{nh}
    Z_{\widehat H^G}(Y^\star)
    \sum_{\tau\in\Omega_T}
    Q_T^{Y^\star}(\tau)
    \prod_{uv\in \mathcal E}(1-\tau_u\tau_v) \nonumber\\
    &\ge
    \frac{e^{-O(h\log n)}}{nh}\,
    Z_{\widehat H^G}(Y^\star)\,
    \Gamma^{2k|E(H)|}
    \left(\frac{\Theta}{\Gamma}\right)^{2kb}.
\end{align}

Next we compare these bounds with the fixed-size partition function on \(\widehat H^G\).
Because $\widehat H^G$ consists of identical gadget copies, whenever
$D(Y)=D_0$ we have the exact identity
\[
    Z_{\widehat H^G}(Y)
    =(1+\lambda_\star)^s Z_{G,+}^{j_0}Z_{G,-}^{h-j_0}.
\]
Denote this common value by $W_0$.
Since there are \(\binom{h}{j_0}\) such phase vectors, Lemma~\ref{lem:fix_poly_penalty} and the estimate
for \(D(Y)\neq D_0\) imply
\begin{equation}\label{eq:hatHG_fix_comparison}
    \frac{e^{-O(h\log n)}}{nh} \binom{h}{j_0} W_0
    \le
    Z^{\mathrm{fix}}_{\widehat H^G}
    \le
    \frac{e^{O(h\log n)}}{nh} \binom{h}{j_0} W_0.
\end{equation}
Here the contribution of \(D(Y)\neq D_0\) is absorbed into the upper bound,
because it is at most \(e^{-\Omega(n/h)}Z_{\widehat H^G}\), while
\(Z_{\widehat H^G}\le e^{O(h\log n)}\binom{h}{j_0}W_0\) by
the ratio bound \textup{(8)} in Lemma~\ref{lem:all_prelim}, and \(n/h\gg h\log n\).

As \(Z_{\widehat H^G}(Y)=W_0\) for every \(Y\) with \(D(Y)=D_0\), the upper bound for the contribution of such phase vectors gives
\[
    \sum_{Y:D(Y)=D_0} Z^{\mathrm{fix}}_{H^G}(Y)
    \le
    \frac{C}{nh}
    \binom{h}{j_0}
    W_0\,
    \Gamma^{2k|E(H)|}
    \left(\frac{\Theta}{\Gamma}\right)^{2kb}.
\]
Together with the estimate for \(D(Y)\neq D_0\), and using
$kh^2+h\log n=o(n/h)$ to absorb that contribution, this gives
\[
    Z^{\mathrm{fix}}_{H^G}
    \le
    \frac{e^{O(h\log n)}}{nh}
    \binom{h}{j_0}
    W_0\,
    \Gamma^{2k|E(H)|}
    \left(\frac{\Theta}{\Gamma}\right)^{2kb}.
\]
Since \(Z_{\widehat H^G}(Y^\star)=W_0\), dividing
\eqref{eq:lower_total_HG_fix} and the preceding upper bound by
\eqref{eq:hatHG_fix_comparison} yields
\begin{equation}\label{eq:ratio_bounds_fix}
e^{-O(h\log n)}
\Gamma^{2k|E(H)|}
\left(\frac{\Theta}{\Gamma}\right)^{2kb}
\le
\frac{Z^{\mathrm{fix}}_{H^G}}{Z^{\mathrm{fix}}_{\widehat H^G}}
\le
e^{O(h\log n)}
\Gamma^{2k|E(H)|}
\left(\frac{\Theta}{\Gamma}\right)^{2kb}.
\end{equation}
Taking logarithms gives
\[
\log Z^{\mathrm{fix}}_{H^G}
-
\log Z^{\mathrm{fix}}_{\widehat H^G}
=
2k|E(H)|\log \Gamma
+
2kb \log\!\left(\frac{\Theta}{\Gamma}\right)
+
O(h\log n).
\]
Since \(\Theta<\Gamma\), we have \(\log(\Gamma/\Theta)>0\), and therefore
\[
b
=
\frac{
\log Z^{\mathrm{fix}}_{\widehat H^G}
-
\log Z^{\mathrm{fix}}_{H^G}
+
2k|E(H)|\log \Gamma
}{
2k\log(\Gamma/\Theta)
}
+
O\!\left(\frac{h\log n}{k}\right).
\]

By construction \(k=\lfloor n^{3\theta/4}\rfloor\) and
\(h=\Theta(n^{\theta/4})\), so \(h\log n/k=o(1)\). Amplify the two FPRAS calls so
that, jointly with probability at least $9/10$, both have a fixed small
relative error. This gives an \(O(1)\) additive error in each logarithm, and
therefore an \(o(1)\) error in the recovered value of \(b_\gamma(H)\). Hence one
can recover the integer \(b_\gamma(H)\) by rounding. On the fixed slice, the
\(\lambda_\star\)-weighted partition functions differ from the corresponding
unweighted fixed-slice counts by the known factor
\(\lambda_\star^{r_L+r_R}\), so an FPRAS for
\(\mathrm{FixedSlice}(\alpha_L,\alpha_R)\) would give the required
approximations. Therefore such an FPRAS would imply a randomized polynomial-time
algorithm for \(\gamma\)-\textsc{MEBC}: together with the gadget event, its
success probability is at least $(89/100)(9/10)>1/2$, and standard repetition
amplifies it. Standard search-to-decision followed by deterministic
verification yields an RP algorithm. Thus no such FPRAS exists unless
$\mathbf{NP}=\mathbf{RP}$.

On the same event $\mathcal G$, the same estimates rule out an efficient
sampler. Under the exact fixed-slice measure on $H^G$,
\[
\mu^{\mathrm{fix}}_{H^G}(D(Y)\ne D_0)
\le \exp\!\left(-\Omega(n/h)+O(kh^2+h\log n)\right)=o(1).
\]
Among phase vectors with $D(Y)=D_0$, increasing $\mathrm{cut}(Y)$ by one multiplies the terminal-compatibility factor by $(\Theta/\Gamma)^{2k}$, while
all phase-counting, point-probability, and gadget-ratio losses contribute only
$e^{O(h\log n)}$. Hence
\[
\mu^{\mathrm{fix}}_{H^G}(D(Y)=D_0,\ \mathrm{cut}(Y)>b)
\le e^{O(h\log n)}(\Theta/\Gamma)^{2k}=o(1),
\]
since $k\gg h\log n$. Given a sampled independent set, the phase vector
is computed from~\eqref{eq:phase_of_gadget} on each copy $G_x$, tossing an
independent fair auxiliary sign whenever that copy is tied.
Therefore an efficient sampler within total variation distance,
say, $1/10$ from the fixed-slice distribution would output a minimum
\(\gamma\)-cardinality cut with probability bounded away from zero. Repetition would solve \(\gamma\)-\textsc{MEBC} in randomized polynomial time. This completes the proof.
\end{proof}

\section*{Statement of AI use}

ChatGPT 5.5 Plus was used for checking mathematical proofs and for assistance with the calculations in the proof of Theorem~\ref{thm:ssm}. The authors assume responsibility for all content.

\section*{Acknowledgments}

WP supported in part by NSF grant CCF-2309708.

\end{document}